\documentclass[11pt]{article}
\usepackage{bbm}
\usepackage{geometry}
\usepackage{color}
\usepackage{amsfonts}
\usepackage{algorithm}
\usepackage{algorithmic}
\usepackage{latexsym, amssymb, amsmath, amscd, amsthm, amsxtra}
\usepackage{mathtools}
\usepackage{enumerate}
\usepackage[all]{xy}
\usepackage{mathrsfs}
\usepackage{fancyhdr}
\usepackage{listings}
\usepackage{hyperref}
\usepackage{enumitem}

\def\E{\mathbb{E}}

\def\11{\mathbbm{1}}

\newtheorem{Theorem}{Theorem}[section]

\newtheorem{Definition}[Theorem]{Definition}
\newtheorem{Lemma}[Theorem]{Lemma}
\newtheorem{Corollary}[Theorem]{Corollary}
\newtheorem{Proposition}[Theorem]{Proposition}
\newtheorem{Remark}[Theorem]{Remark}
\newtheorem{Claim}[Theorem]{Claim}

\numberwithin{equation}{section}

\makeatletter
\newenvironment{breakablealgorithm}
{% \begin{breakablealgorithm}
		\begin{center}
			\refstepcounter{algorithm}% New algorithm
			\hrule height.8pt depth0pt \kern2pt% \@fs@pre for \@fs@ruled
			\renewcommand{\caption}[2][\relax]{% Make a new \caption
				{\raggedright\textbf{\ALG@name~\thealgorithm} ##2\par}%
				\ifx\relax##1\relax % #1 is \relax
				\addcontentsline{loa}{algorithm}{\protect\numberline{\thealgorithm}##2}%
				\else % #1 is not \relax
				\addcontentsline{loa}{algorithm}{\protect\numberline{\thealgorithm}##1}%
				\fi
				\kern2pt\hrule\kern2pt
			}
		}{% \end{breakablealgorithm}
		\kern2pt\hrule\relax% \@fs@post for \@fs@ruled
	\end{center}
}
\makeatother

\hypersetup{colorlinks=true,linkcolor=blue}

\title{A polynomial time iterative algorithm for matching Gaussian matrices with non-vanishing correlation}
\author{Jian Ding \\ Peking University  \and  Zhangsong Li \\ Peking University}
\date{\today}

\begin{document}

\maketitle

\begin{abstract}
Motivated by the problem of matching vertices in two correlated Erd\H{o}s-R\'enyi graphs, we study the problem of matching two correlated Gaussian Wigner matrices. We propose an iterative matching algorithm, which succeeds in polynomial time as long as the correlation between the two Gaussian matrices does not vanish. Our result is the first polynomial time algorithm that solves a graph matching type of problem when the correlation is an arbitrarily small constant. 
\end{abstract}

\section{Introduction}

In this work, we study the problem of matching two correlated Wigner matrices, and we consider the case of symmetric matrices in order to be consistent with the graph matching problem. More precisely, for two sets $V$ and $\mathsf{V}$ with cardinality $n$, define $E(V)$ to be the set of unordered pairs $(u,v)$ with $u,v \in V, u \not = v$ and define $\mathsf{E(V)}$ similarly with respect to $\mathsf{V}$. Let $\pi$ be a uniform bijection between $V$ and $\mathsf{V}$. Let $G$ and $\mathsf{G}$ be two symmetric random matrices indexed by $V$ and $\mathsf{V}$ respectively. In addition, conditioned on $\pi$ we have $\left(G_{u,v}, \mathsf{G}_{\pi(u), \pi(v)}\right)\sim \mathbf F$ independent among all unordered pairs $(u, v) \in E(V)$, where $\mathbf F$ is the law for a pair of correlated random variables. Then, $G$ and $\mathsf{G}$ can be viewed as complete graphs with correlated edge weights.

\begin{Theorem}\label{thm-main}
Let $\mathbf F$ be the law of a pair of standard bivariate normal variables (i.e., with mean 0 and variance 1) with correlation $\epsilon > 0$. Then there exists a constant $C = C(\epsilon)>0$ and an algorithm (see Algorithm~\ref{algo:matching}) with $O(n^C)$-running time that takes $(G, \mathsf G)$ as input and outputs the latent matching $\pi$ with probability tending to 1 as $n\to \infty$.
\end{Theorem}

\subsection{Backgrounds and related works}

Our work is closely related to the problem of matching two correlated Erd\H{o}s-R\'enyi graphs, when $\mathbf F$ is the law of a pair of Bernoulli variables with mean $ps$ and covariance $p(1-p)s^2$. Recently, the random graph matching problem has been extensively studied with important motivations from applied fields such as social network analysis \cite{NS08,NS09}, computer vision \cite{CSS07,BBM05}, computational biology \cite{SXB08,VCP15} and natural language processing \cite{HNM05}. From the collective efforts of the community  \cite{CK16, CK17, HM20, WXY20+,WXY21+, GML21, DD22+, DD22+b}, it is fair to say that up to now we have fairly complete understanding on the information thresholds for the problem of correlation detection as well as vertex matching for Erd\H{o}s-R\'enyi graph models. In what follows, we elaborate on the progress on the computational aspect, which is the main focus of the present work. 

The existing algorithms are essentially of two types, the optimization-based method that relies on ``convex relaxation and rounding'' \cite{FMWX22a, FMWX22b} and the signature-based method that relies on ``computing and comparing signatures'' \cite{PG11, YG13, LFP14, KHG15, FQRM+16, SGE17, BCL19, DMWX21, BSH19, CKMP19,DCKG19, MX20,  GM20, MRT21+, MWXY21+, GMS22+, MWXY22+}. The optimization-based method has the appeal that it directly addresses the problem of maximizing the overlap. In a couple of very impressive works \cite{FMWX22a, FMWX22b}, employing a clever spectral relaxation method and a novel probabilistic analysis, the authors obtained a polynomial time matching algorithm when the correlation approaches 1 at a rate polylog in $n$. That being said, the signature-based method seems to have been pushed much further. In \cite{BCL19}, the authors proposed a quasi-polynomial time algorithm (based on subgraph counts) which succeeds when the correlation is non-vanishing; in \cite{DMWX21}, the authors proposed a polynomial time algorithm (based on degree profile)  which succeeds when the correlation approaches 1 at a rate polylog in $n$. It is fair to say that both methods are of inspiration to future works in this line. In a later breakthrough \cite{MRT21+}, the authors found the first polynomial time algorithm (based on some sophisticated partition tree) that succeeds for exact matching with constant correlation; in a recent breakthrough \cite{MWXY22+} (see also \cite{GMS22+} for a remarkable result on partial recovery of similar flavor when the average degree is $O(1)$), the authors substantially improved \cite{MRT21+} and obtained a polynomial time algorithm which succeeds as long as the correlation is above some threshold given by the Otter's constant which is around $\sqrt{0.338}$ (their algorithm is based on a carefully curated family of rooted trees called chandeliers, and this also covers a much wider parameter regime than that in \cite{MRT21+}).

\subsection{Our contributions}

Our Theorem~\ref{thm-main} is on matching correlated Gaussian Wigner matrices, and this model is formally simpler than random graph matching since we can apply thresholding to reduce our model to a graph matching model. As in \cite{DMWX21, FMWX22a}, the assumption of Gaussianity provides a substantial technical simplification: on the one hand, we do believe that the phenomenon we reveal and the method we obtain \emph{should} apply to the random graph matching problem (with a caveat that there shall be an assumption on the lower bound of the edge density); on the other hand, we do acknowledge that it is of substantial challenge to extend our analysis to graph matching, and this is a natural future direction. With these clarified, we  wish to emphasize that our work has a number of conceptual novelties which may offer useful insights for this class of problems.
\begin{itemize}
\item While in \cite{MWXY22+} (see also \cite{GMS22+}) polynomial time matching algorithms were obtained when the correlation is above the threshold from the Otter's constant,  our work strongly suggests a polynomial time algorithm for any non-vanishing correlation, at least in the regime for fixed $p>0$. In fact, we believe a polynomial time algorithm should succeed as long as $p \geq n^{-\alpha}$ for a fixed constant $\alpha < 1$. 
\item It seems that the power for the running time in \cite{MWXY22+} tends to $\infty$ as the correlation approaches the threshold from the Otter's constant. For our algorithm the power only tends to $\infty$ as the correlation tends to $0$, and we are under the feeling that this is the best possible.
\item Our work provides an iterative algorithm which seems to have novel features. While the message-passing algorithm in \cite{GMS22+, PSSZ22} is also iterative, it is completely different from ours and in particular their algorithm exploits the local tree structure in a crucial way. In addition, an iterative greedy algorithm has been employed recently in \cite{DDG22+} to obtain a polynomial time approximation scheme for the maximal overlap between two \emph{independent} Erd\H{o}s-R\'enyi graphs. The iteration in \cite{DDG22+} is also completely different from the one used in this work; this is not surprising since in \cite{DDG22+} we are dealing with a pure optimization problem without a planted signal. In a broader context, iterative algorithms have been extensively applied, and we wish to emphasize one novel feature of our algorithm: in the usual application of an iterative algorithm, the outcomes converge to the planted truth as the iteration evolves; but in our algorithm, along our iteration we obtain signals in a vector with increasing dimensions where the signal carried at each coordinate decreases. It seems at least somewhat surprising that an iterative algorithm with signal per coordinate weakened at every step could eventually work, and it seems not obvious at all to harness the increase on the number of  coordinates for compensation.
\item Partly related to the aforementioned novel feature of our iteration and partly related to the fact that our algorithm handles a pair of correlated matrices simultaneously, the analysis of our iterative algorithm is of substantial challenge and this challenge seems to be of a rather novel type (e.g., compared to that in \cite{BM10}). We believe that our method of analysis will be of some inspirations for future works, possibly even outside the scope of graph matching problems.
\end{itemize}

\subsection{Discussions and perspectives}

Our work reiterates a number of future research directions as we discuss below.

\underline{Computational phase transition for random graph matching.} As we suggested above, it would be very interesting to extend our iterative algorithm to random graph matching for $p = n^{-\alpha}$ with $\alpha < 1$. In addition, we feel that this is tight in the following sense: as the correlation $\epsilon \to 0$, no polynomial time algorithm with a fixed power would be able to match two random graphs. One way to provide an evidence is to employ the framework of low degree polynomial as in \cite{SW22, MWXY21+}, although it seems that one needs to develop a version of low degree polynomial with suitable truncation in order to obtain the sharp phase transition. Another possibility is to use the framework of overlap gap property \cite{Garmarnik21}. Finally, we point out that for $p = \lambda/n$, it was conjectured in \cite{GML22} that the computational threshold is indeed given by the Otter's constant; this conjecture is consistent with our intuition. 

\underline{Robust algorithms.} Currently, essentially all matching algorithms are proposed for specific graph models (i.e., correlated  Erd\H{o}s-R\'enyi graphs) or at least the analysis of the algorithms crucially relies on the specific properties of the model. It would be of great importance to develop robust algorithms that would apply to a wide class of random graph models, and it would be a great success even if the proposed robust algorithms do not necessarily achieve the presumed sharp computational threshold for any specific model. For instance, one may consider a general correlated random graph model where the two graphs are independently subsampled from a mother graph. It would be really exciting if some \emph{minimal} assumptions can be posed on the mother graph under which an efficient algorithm can be developed for graph matching.   
 
\underline{Other important random graph models.} Another important direction is to understand computational phase transitions for matching other important correlated random graph models, such as the random geometric graph model \cite{WWXY22+}, the random growing graph model \cite{RS20+} and the stochastic block model \cite{RS21}. We emphasize that it is also important to propose and study correlated graph models based on important real-world and scientific problems, albeit the models do not appear to be ``canonical'' from a mathematical point of view.
 
\underline{Other matching problems.} Perhaps an even more canonical matching problem is to match correlated random vectors, where we observe two pools of random vectors and each pair of vectors under the latent matching are correlated (and one may assume that different pairs are independent of each other). In many practical problems, the correlation is only positive in a small unknown set of coordinates (see \cite{CJMNZ22+} for study on a closely related model as well as its applications on single cell problems). It is a very interesting question to obtain a computational phase transition for this model, where the parameters are naturally the number and the dimension of these vectors,  the strength of the correlation (when it is positive) and the number of coordinates with positive correlation. It is possible that our iterative algorithm would shed some light on this model too.
 
\subsection{Notations}

We record in this subsection some notation conventions.

In the rest of the paper, we assume that $\mathbf F$ is the law of a pair of standard bivariate normal variables as in Theorem~\ref{thm-main}.

Given two random variables $X,Y$ and a $\sigma$-algebra $\mathfrak{S}$, the notation $ X|{\mathfrak{S}} \overset{d}{ = } Y|{\mathfrak{S}}$ means
that for any integrable function $\phi$ and  for any bounded random variable $Z$ measurable on $\mathfrak{S}$, we have $\mathbb{E}[ \phi(X)Z ] = \mathbb{E}[ \phi(Y)Z ]$. In words, $X$ is equal in distribution to $Y$ conditioned on $\mathfrak{S}$. When $\mathfrak{S}$ is the trivial $\sigma$-field, we simply write $X \overset{d}{=} Y$.

We also need some standard notations in linear algebra. For an $m*m$ matrix $A=(a_{ij})_{m*m}$, if $A$ is symmetric we let $\varsigma_1(A) \geq \varsigma_2(A) \geq \ldots \geq \varsigma_m(A)$ be the eigenvalues of $A$. Denote by $\mathrm{rank}(A)$ the rank of the matrix $A$. We define the Hilbert-Schmidt norm (i.e., the 2-Frobenius norm), operator norm, 1-norm and $\infty$-norm of $A$ respectively by
\begin{align*}
    &\| A \|^2_{\mathrm{HS}}=\sum_{i,j} a_{ij}^2 = \mathrm{tr}(AA^{*}) = \mathrm{tr}(A^{*}A)\,,\\ %= \sum_{k=1}^{m} \varsigma_k(A)^2\,, \\
    &\| A \|_{\mathrm{op}}= \max_{x \not = 0} \left \{  \frac{\|  Ax \|_2}{\| x \|_2} \right \}\,,\\
    % = \varsigma_1(A)\,, \\
    &\| A  \|_1 = \max_{x \not = 0} \left \{  \frac{ \| Ax  \|_{\infty} }{ \| x \|_{\infty}  }  \right \} = \max_{1 \leq k \leq m} \left \{ \sum_{i=1}^{m} |a_{k,i}| \right  \}\,,\\
    &\| A \|_{\infty} =  \max_{x \not = 0} \left \{ \frac{ \| Ax  \|_{1} }{ \| x \|_{1}  } \right \} = \max_{1 \leq k \leq m} \left \{ \sum_{i=1}^{m} |a_{i,k}|  \right \}\,,
\end{align*}
where $\mathrm{tr}(\cdot)$ is the trace for a squared matrix. 
Note that $\| A\|_{\mathrm{op}}, \|A  \|_{1}$ and $\| A \|_{\infty}$ are the norms of $A$ regarded as an operator on different Banach spaces, i.e., on $(\mathbb{R}^m, \| \cdot \|_2),(\mathbb{R}^m, \| \cdot\|_{\infty})$ and $(\mathbb{R}^m, \| \cdot \|_1)$ respectively.

For two vectors $\gamma, \mu \in \mathbb{R}^d$, we say $|\gamma| \geq |\mu|$ if the entries satisfy $|\gamma(i)| \geq |\mu(i)|$ for $1 \leq i \leq d$; we define $|\gamma| \leq |\mu|, \gamma \leq \mu$ similarly. In addition, for $\alpha \in \mathbb R$, we write $|\gamma| \leq \alpha$ if $|\gamma(i)| \leq \alpha$ for $1 \leq i \leq d$.

We will use $\mathrm{I}_{d \times d}$ to denote the $d \times d$ identity matrix (and we drop the subscript if the dimension is clear from the context). Similarly, we denote $\mathrm{0}_{m \times d}$  the $m \times d$ zero matrix and denote $\mathsf{J}_{m \times d}$  the $m \times d$ matrix with all entries being 1. The indicator function of sets $A$ is denoted by $\mathbf{I}_{A}$. For a matrix or a vector $A$, we will use $A^*$ to denote its transpose.

\medskip
 
{\bf Acknowledgment.} We thank Zongming Ma, Yihong Wu and Jiaming Xu for extensive and stimulating discussions on random graph matching problems; we thank James Johndrow and Nancy Zhang for stimulating discussions which led to the formulation of the aforementioned problem of matching correlated random vectors; we thank Fan Yang for stimulating discussions on an early stage of the project.

\section{An iterative matching algorithm}\label{sec:algorithm-description}
We first describe the underlying heuristics of our algorithm. For $t \geq 0$ and $K_t \geq 1$ (to be specified), our wish is to iteratively construct a sequence of paired sets $( \Gamma^{(t)}_k,\Pi^{(t)}_k)_{1 \leq k \leq K_t}$ where $\Gamma^{(t)}_k \subset V$, $\Pi^{(t)}_k \subset \mathsf V$ and $|\pi(\Gamma^{(t)}_k) \cap \Pi^{(t)}_k| \geq (1+\varepsilon_t) \frac{|\Gamma^{(t)}_k| \cdot |\Pi^{(t)}_k|}{n}$, i.e., $\Gamma^{(t)}_k \times \Pi^{(t)}_k$ contains more true pairs of form $(v, \pi(v))$ than that when the two sets are sampled uniformly randomly.  

The initialization can be achieved if we have $K_0$ true pairs available to us as seeds, since we can then define $(\Gamma^{(0)}_k, \Pi^{(0)}_k)$ as in \eqref{equ_initial_set}. In fact, if indeed we have $K_0$ true pairs as seeds, then the running time of our algorithm can be reduced to $O(n^{2+o(1)})$ (see Proposition~\ref{prop_time_complexity}). In order to address the fact that we do not have seeds, we essentially just take arbitrary $K_0$ vertices from $V$ and try all possible pairings to these vertices; this increases the running time by a factor of $n^{K_0}$. 

The core challenge is on the iteration. Since each pair $(\Gamma^{(t)}_k,\Pi^{(t)}_k)$ carries some signal (i.e., $\varepsilon_t>0$ as we suppose by induction), we then hope to construct paired sets for $t+1$ by checking the total edge weights (say denote as $D_{v, k}^{(t)}$) between each $v\in V$ (respectively $\mathsf v\in \mathsf V$) and $\Gamma^{(t)}_k$ (respectively, $\Pi^{(t)}_k$). A moment of thinking convinces us that in this way $\varepsilon_t$ (i.e., the signal) will have to decrease in $t$ (see \eqref{equ_def_iter_varepsilon}). In order to address this, here comes our main (albeit simple in retrospect) observation: we may take advantage of \emph{many} linear combinations of $\{ D_{v, k}^{(t)}: 1 \leq k \leq K_t \}$ (this is why we choose $K_{t+1}$ recursively as in \eqref{equ_iter_K}) and by Proposition~\ref{prop_random_vector} these linear combinations are effectively independent of each other. In one sentence, we use the increase in the number of paired sets to compensate the loss that the signal carried in each pair decreases. As we hope, once the iteration progresses to time $t = t^*$ we would have accumulated enough total signal so that we can just complete the matching directly in the next step, as described in Section~\ref{sec:ending-iteration}. 

At this point, it seems the ``only'' remaining challenge is to control the correlation among different iterative steps. However, let us stress that this challenge has quite some novel features. A natural attempt is to employ Gaussian projections to remove the influence of conditioning on outcomes in previous steps. This is indeed very useful since all the conditioning can be expressed as conditioning on linear combinations of Gaussian variables. Although highly-nontrivial, this is possible to deal with as have done in e.g., \cite{BM10} (and this is the main reason why Gaussian model simplifies the analysis here). However, there is a new difficulty in our model since there is a latent matching which is \emph{a priori} inaccessible by our algorithm, and as a result the part of projection arising from the correlation between the two matrices cannot be subtracted from the algorithmic point of view (of course in the analysis we can still do this). We also note that in a recent work \cite{DDG22+} an iterative greedy algorithm was proposed on maximizing the overlap between two \emph{independent} Erd\H{o}s-R\'enyi graphs. Thanks to independence, the major challenge of correlation between two matrices (through the latent matching) was not present in \cite{DDG22+}; even so, the analysis of \cite{DDG22+} is difficult and delicate. We refer to Section~\ref{sec:analysis-outline} for an overview discussion on how such correlations are dealt with in the analysis of the algorithm. 

Next, we describe in detail our iterative algorithm, which consists of a few steps including preprocessing (see Section~\ref{sec:preprocessing}), initialization (see Section~\ref{sec:initialization}), iteration (see Section~\ref{sec:iteration}) and finishing (see Section~\ref{sec:ending-iteration}). We discuss in Section~\ref{sec:sampling} the random sampling procedure employed in the iteration since a resampling may be necessary (see Remark~\ref{remark-sampling}). We formally present our algorithm in Section~\ref{sec:formal-algorithm}. In Section~\ref{sec:matrix-rank} we prove a lemma regarding to the ranks of some matrices arising from our algorithm (which then ensures that our algorithm is well-defined). In Section~\ref{sec:runtime-analysis} we analyze the time complexity of the algorithm.

\subsection{Preprocessing}\label{sec:preprocessing}

In order to facilitate analysis later, we first employ some preprocessing for $G$ and $\mathsf{G}$. Sample i.i.d.\ standard (i.e., mean 0 and variance 1) Gaussian variables $\{ \Tilde{G}_{u,v}, \Tilde{\mathsf{G}}_{\mathsf{u,v}} : (u,v) \in E(V),(\mathsf{u,v}) \in \mathsf{E}(\mathsf{V})\}$. Arbitrarily assign an orientation to each edge in $E$ and $\mathsf E$, and this gives two sets of directed edges  $ \overrightarrow{E}(V)=  \{ \overrightarrow{(u,v)}:(u,v) \in E(V)  \}$ and $\overrightarrow{\mathsf{E}}(\mathsf{V})=  \{ \overrightarrow{(\mathsf{u},\mathsf{v})}:\mathsf{(u},\mathsf{v)} \in \mathsf{E(V)}  \}$. Then we define
\begin{align*}
    & \Hat{G}_{u,v}=\frac{G_{u,v}+\Tilde{G}_{u,v}}{\sqrt{2}} \mbox{ and } \Hat{G}_{v,u}= \frac{G_{u,v} - \Tilde{G}_{u,v}}{\sqrt{2}} \textup{ for } \overrightarrow{(u,v)} \in \overrightarrow{E}(V)\,,\\
    & \Hat{\mathsf{G}}_{\mathsf{u,v}}= \frac{\mathsf{G}_{\mathsf{u,v}}+\Tilde{\mathsf{G}}_{\mathsf{u,v}}}{\sqrt{2}} \mbox{ and }\Hat{\mathsf{G}}_{\mathsf{v,u}}= \frac{\mathsf{G}_{\mathsf{u,v}}-\Tilde{\mathsf{G}}_{\mathsf{u,v}}}{\sqrt{2}} \textup{ for } \overrightarrow{(\mathsf{u,v})} \in \overrightarrow{\mathsf{E}}(\mathsf{V})\,.
\end{align*}
With these modifications, instead of being a symmetric matrix, $\Hat{G}$ (respectively $\Hat{\mathsf {G}})$ is a matrix with independent entries. This is useful since for instance we now have that the sum of random weights on all outgoing edges from $u$ and from $v$ are independent. It is straightforward to verify that $\{  \Hat{G}_{u,v}  \}_{u \not = v}$ and $\{  \Hat{\mathsf{G}}_{\mathsf{u,v}}  \}_{\mathsf{u} \not = \mathsf{v}}$ are two families of i.i.d. Gaussian variables. Also, we have
\begin{align*}
    \mathrm{Cov}(\Hat{G}_{u,v},\Hat{\mathsf{G}}_{\pi(u),\pi(v)}) =\mathrm{Cov} (\Hat{G}_{u,v},\Hat{\mathsf{G}}_{\pi(v),\pi(u)})
    =\frac{\epsilon}{2}\,.  
\end{align*}
This means that, the strength of the signal is weakened, but only by a factor of 2 which is not an issue for our purpose.

\subsection{Initialization}\label{sec:initialization}
For a pair of standard bivariate normal variables $(X, Y)$ with correlation $u$, we define $\phi: [-1, 1] \mapsto [0, 1]$ by (below the number 10 is somewhat arbitrarily chosen)
\begin{align}
    \phi(u) = \mathbb{P} [|X| \geq 10 ,|Y| \geq 10]\,.
    \label{equ_def_func_phi}
\end{align}
In addition, we define $\iota = \frac{1}{2} \phi''(0)$ and we write $\alpha = \mathbb{P}[|X| \geq 10]$. Let $\kappa = \kappa(\epsilon)$ be a sufficiently large constant depending on $\epsilon$ whose exact value will be decided later in \eqref{eq-kappa-choice}. Set $K_0 = \kappa$. We then arbitrarily choose a sequence $A = (u_1, u_2, \ldots , u_{K_0})$ where $u_i$'s are distinct vertices in $V$, and list all the sequences of length $K_0$ with distinct elements in $\mathsf{V}$  as $\mathsf{A}_1, \mathsf{A}_2, \ldots , \mathsf{A}_{\mathtt{M}}$ where $\mathtt{M} = \mathtt{M}(n,\epsilon) = n (n-1) \ldots (n-K_0+1)$. As hinted earlier, for each $1 \leq \mathtt{m} \leq \mathtt{M}$, we will run a procedure of initialization and iteration and we know that for one of them (although we cannot decide which one it is \emph{a priori}) we are running an algorithm as if we have $K_0$ true pairs as seeds. For notation convenience, when describing the initialization and iteration we will drop $\mathtt m$ from notations, but we should keep in mind that this procedure is applied to each $\mathsf A_{\mathtt m}$. With this clarified, we take a fixed $\mathtt m$ and denote $\mathsf{A}_{\mathtt{m}} = \{ \mathsf{u}_1, \mathsf{u}_2, \ldots, \mathsf{u}_{K_0}\}$. In what follows, we abuse the notation and write $V \setminus A$ when regarding $A$ as a set (similarly for $\mathsf A_{\mathtt m}$). Define for $1 \leq k \leq K_0$,
\begin{equation}
    \Gamma^{(0)}_k = \{ v \in V \backslash A : |\Hat G_{v, u_k}| \geq 10\}  \mbox{ and }
    \Pi^{(0)}_k = \{ \mathsf{v} \in \mathsf{V} \backslash \mathsf{A}_{\mathtt{m}} : |\Hat{\mathsf{G}}_{\mathsf{v},\mathsf{u}_k}| \geq 10\} \,.
    \label{equ_initial_set}
\end{equation}
(In the above and in the iteration below we have used the absolute value of a Gaussian instead of a Gaussian itself, and the purpose is to introduce more symmetry in order to facilitate our analysis. For instance, this would be useful in controlling \eqref{eq-3.47-part-2} later.)
In addition, we define $\Phi^{(0)}, \Psi^{(0)}$ to be $K_0 \times K_0$ matrices by
\begin{align}
    \Phi^{(0)} = \mathrm{I} \mbox{ and }
    \Psi^{(0)} = \frac{\phi(\epsilon) - \alpha^2 }{\alpha-\alpha^2} \mathrm{I}\,,
    \label{equ_initial_matrix}
\end{align}
and in the iterative steps we will also construct $\Phi^{(t)}$ and $\Psi^{(t)}$ for $t\geq 1$.

\subsection{Iteration}\label{sec:iteration}
We emphasize again that in this subsection we are describing the iteration for a fixed $1\leq \mathtt m \leq \mathtt M$ and eventually this iterative procedure will be applied to each $\mathtt m$.
Define
\begin{align}
    \varepsilon_{0} = \frac{\phi(\epsilon) - \phi(0)}{\alpha - \alpha^2} \mbox{ and } K_{t+1} = \frac{1}{ \varkappa  } K_t^2 \mbox{ for } t \geq 0\,,
    \label{equ_iter_K}
\end{align}
where we set $\varkappa = \varkappa(\epsilon) =  \frac{10^{20} \epsilon^{-20}}{\iota^2 (\alpha - \alpha^2)^2 }$. As one may expect, we will define our iteration step in an inductive manner. Now suppose that $( \Gamma^{(s)}_k,\Pi^{(s)}_k)_{1 \leq k \leq K_s}$ has been constructed for $s \leq t$. For $v \in V \backslash A,\mathsf{v} \in \mathsf{V} \backslash \mathsf{A}_{\mathtt{m}}$,  define $D^{(t)}_v, \mathsf{D}^{(t)}_{\mathsf{v}} \in \mathbb{R}^{K_t}$ to be the ``normalized degrees'' of $v$ to $\Gamma^{(t)}_k$ and of $\mathsf{v}$ to $\Pi^{(t)}_k$ as follows:
\begin{equation}
    \begin{aligned}
        D^{(t)}_v(k) &= \frac{1}{\sqrt{(\alpha - \alpha^2)n}} \sum_{u \in V \backslash A }( \mathbf{I}_{u \in \Gamma^{(t)}_k} - \alpha) \Hat{G}_{v,u}\,, \\
        \mathsf{D}^{(t)}_{\mathsf{v}}(k) &= \frac{1}{\sqrt{(\alpha-\alpha^2)n}} \sum_{\mathsf{u} \in \mathsf{V} \backslash \mathsf{A}_{\mathtt{m}} }(\mathbf{I}_{\mathsf{u} \in \Pi^{(t)}_k} -\alpha) \Hat{\mathsf{G}}_{\mathsf{v,u}}\,.
        \label{equ_degree}
    \end{aligned}
\end{equation}
Recalling \eqref{equ_initial_matrix}, in order to further describe our iterative step we will also need to use an important property on eigenvalues for $\Phi^{(t)}$ and $\Psi^{(t)}$ (see \eqref{equ_def_iter_matrix} for their definitions), as incorporated in the next lemma.
\begin{Lemma}{\label{lemma_matrix_eigenvalue}}
Let $(\Phi^{(t)}, \Psi^{(t)})$ be initialized as in \eqref{equ_initial_matrix} and inductively defined as in \eqref{equ_def_iter_matrix}. Then, $\Phi^{(t)}$ has $\frac{3}{4}K_t$ eigenvalues between $0.9$ and $1.1$, and $\Psi^{(t)}$ has $\frac{3}{4}K_t$ eigenvalues between $ 0.9 \varepsilon_t$ and $1.1 \varepsilon_t$.
\end{Lemma}
\begin{Remark}
Lemma~\ref{lemma_matrix_eigenvalue} will be proved in Section \ref{sec:matrix-rank}. The proof is by induction: once it was proved for the $t$-th step, then the iterative construction for the $(t+1)$-th step makes sense and formally it is only at this point \eqref{equ_def_iter_matrix} is well-defined. In what follows, we will ignore this subtlety since it is only an issue of formality. 
\end{Remark}
Assuming Lemma~\ref{lemma_matrix_eigenvalue}, we can then write $\Phi^{(t)}$ and $\Psi^{(t)}$ as their spectral decompositions:
\begin{equation}\label{eq-spectral-decomposition}
    \Phi^{(t)}=\sum^{K_t}_{i=1} \lambda^{(t)}_i
    \left({\nu^{(t)}_i} \right)^{*} \left(\nu^{(t)}_i \right) \mbox{ and }
    \Psi^{(t)}= \sum_{i=1}^{K_t} \mu^{(t)}_i \left({\xi^{(t)}_i} \right)^{*} \left(\xi^{(t)}_i \right)
\end{equation}
where 
\begin{equation}\label{eq-lambda-mu-bound}
\lambda^{(t)}_i \in (0.9,1.1), \mu^{(t)}_i \in ( 0.9 \varepsilon_t, 1.1 \varepsilon_t ) \mbox{ for } 1\leq  i \leq \frac{3K_t}{4}
\end{equation} and $\nu_i,\xi_i$ are the unit eigenvectors with respect to $\lambda_i,\mu_i$ respectively. In addition, for $s, t$ we define
$\mathrm{M}_{\Gamma}^{(t,s)},\mathrm{M}_{\Pi}^{(t,s)},\mathrm{P}_{\Gamma,\Pi}^{(t,s)}$ to be  $K_t*K_s$ matrices by
\begin{equation}
    \begin{aligned}
        \mathrm{M}_{\Gamma}^{(t,s)}(i,j) & = \frac{  |\Gamma^{(t)}_i \cap \Gamma^{(s)}_j | - \alpha |\Gamma^{(t)}_i | - \alpha |\Gamma^{(s)}_j | + \alpha^2  n }{ (\alpha - \alpha^2) n} \,,\\
        \mathrm{M}_{\Pi}^{(t,s)} (i,j)  & = \frac{ |\Pi^{(t)}_i \cap \Pi^{(s)}_j| - \alpha |\Pi^{(t)}_i| - \alpha |\Pi^{(s)}_j| + \alpha^2 n} {(\alpha - \alpha^2) n}   \,, \\
        \mathrm{P}_{\Gamma,\Pi}^{(t,s)}(i,j) & = \frac{ |\pi(\Gamma^{(t)}_i) \cap \Pi^{(s)}_j| - \alpha |\Gamma^{(t)}_i| - \alpha |\Pi^{(s)}_j| + \alpha^2 n }{(\alpha - \alpha^2) n}\,.
        \label{equ_martix_M_P}
    \end{aligned}
\end{equation}
Note that $\mathrm{M}_{\Gamma},\mathrm{M}_{\Pi}$ are accessible by the algorithm but $\mathrm{P}_{\Gamma,\Pi}$ is not (since it relies on the latent matching). We further define two linear subspaces as follows:
\begin{equation}
    \begin{aligned}
        \mathrm{W}^{(t)} & \overset{\triangle}{=}  \left \{ x \in \mathbb{R}^{K_t} :
        x \mathrm{M}_{\Gamma}^{(t,s)} = 0, 
        x \mathrm{M}_{\Pi}^{(t,s)} = 0,  \mbox{ for all } s<t \right\} \,, \\
        \mathrm{V}^{(t)} & \overset{\triangle}{=}  \mathrm{span} \left \{ \nu_1,\nu_2,\ldots,\nu_{\frac{3}{4}K_t} \right\}  \cap \mathrm{span} \left\{  \xi_1,\xi_2,\ldots,\xi_{\frac{3}{4}K_t}   \right\}  \cap \mathrm{W}^{(t)} \,.
        \label{equ_linear_space}
    \end{aligned}
\end{equation}
We refer to Remark~\ref{remark-3.5} for underlying reasons of the definition above. 
Note that the number of linear restrictions posed on $\mathrm{W}^{(t)}$ are at most $ 2\sum_{i=1}^{t}K_{i-1} < 3K_{t-1}$, so $\dim ( \mathrm{V}^{(t)}) \geq \frac{1}{2}K_t - 3K_{t-1}$. We claim we can choose $\eta^{(t)}_1,\eta^{(t)}_2,\ldots,\eta^{(t)}_{\frac{1}{12}K_t}$ from $\mathrm{V}^{(t)}$ such that
\begin{align}
    & \eta^{(t)}_i \mathrm{M}_{\Gamma}^{(t,t)} \left(\eta^{(t)}_j \right)^{*}
    =\eta^{(t)}_i \mathrm{M}_{\Pi}^{(t,t)} \left(\eta^{(t)}_j \right)^{*}
    =\eta^{(t)}_i \Psi^{(t)} \left(\eta^{(t)}_j \right)^{*} =0 \,, \label{equ_vector_orthogonal}
    \\
    & \eta^{(t)}_i \Phi^{(t)} \left(\eta^{(t)}_i \right)^{*} =1, \quad  2 \varepsilon_t \geq   \eta^{(t)}_i \Psi^{(t)} \left(\eta^{(t)}_i \right)^{*} \geq 0.5 \varepsilon_t\,. \label{equ_vector_unit}
\end{align}
We first verify that \eqref{equ_vector_orthogonal} can be satisfied. To this end, we may choose an arbitrary $\eta^{(t)}_1 \in \mathrm{V}^{(t)}$. Supposing we have chosen valid $\eta^{(t)}_1,\eta^{(t)}_2,\ldots, \eta^{(t)}_k \in \mathrm{V}^{(t)}$, we will show that there is a valid choice for $\eta^{(t)}_{k+1}$ as long as $k < \frac{1}{12} K_t$ (and this completes the verification of \eqref{equ_vector_orthogonal}). Under the assumption of $k < \frac{1}{12} K_t$, the orthogonal space of $\{ \eta^{(t)}_1,\eta^{(t)}_2, \ldots, \eta^{(t)}_k \}$ with respect to $\mathrm{M}^{(t,t)}_{\Gamma}, \mathrm{M}^{(t,t)}_{\Pi}$ and $\Psi^{(t)}$ has dimension $K_t - 3k> \frac{3}{4}K_t$, and thus has a non-empty intersection with $\mathrm{V}^{(t)}$ which results in a choice of $\eta^{(t)}_{k+1}$ satisfying \eqref{equ_vector_orthogonal}. We next verify that we can simultaneously satisfy \eqref{equ_vector_unit}. Since $\eta^{(t)}_i \in \mathrm{span} \{ \nu_1,\nu_2,\ldots,\nu_{\frac{3}{4}K_t} \}  \cap \mathrm{span} \{  \xi_1,\xi_2,\ldots, \xi_{\frac{3}{4}K_t}\}$,
we can write
\begin{align*}
    \eta^{(t)}_i = \sum_{i=1}^{\frac{3}{4}K_t} x_i \nu_i
    = \sum_{j=1}^{\frac{3}{4}K_t}y_i \xi_i \mbox{ for } x_i, y_i \in \mathbb R\,.
\end{align*}
By the orthogonality of $\{  \xi_i \}$ and $\{ \nu_i \}$, we then have that
\begin{align*}
    \eta^{(t)}_i \Phi^{(t)} \Big(\eta^{(t)}_i \Big)^{*}
    & = \Big( \sum_{i=1}^{\frac{3}{4}K_t} x_i \nu^{(t)}_i \Big)
    \Big(\sum_{i=1}^{K_t} \lambda_i \Big(\nu^{(t)}_i \Big)^{*} \nu^{(t)}_i  \Big)
    \Big(\sum_{i=1}^{\frac{3}{4}K_t} x_i \Big(\nu^{(t)}_i \Big)^{*}  \Big)
    = \sum_{i=1}^{\frac{3}{4}K_t} \lambda_i x_i^2  \,, \\
    \eta^{(t)}_i \Psi^{(t)} \Big(\eta^{(t)}_i \Big)^{*}
    & = \Big( \sum_{i=1}^{\frac{3}{4}K_t} y_i \xi^{(t)}_i \Big)
    \Big(\sum_{i=1}^{K_t} \mu_i \Big(\xi^{(t)}_i \Big)^{*} \xi^{(t)}_i  \Big)
    \Big(\sum_{i=1}^{\frac{3}{4}K_t} x_i \Big(\xi^{(t)}_i  \Big)^{*} \Big)
    = \sum_{i=1}^{\frac{3}{4}K_t} \mu_i y_i^2\,.
\end{align*}
Since $\| \eta^{(t)}_i  \|_2^2 = \sum_{i=1}^{\frac{3}{4}K_t} x_i^2 = \sum_{i=1}^{\frac{3}{4}K_t} y_i^2$, recalling \eqref{eq-lambda-mu-bound} we can satisfy \eqref{equ_vector_unit} by properly choosing the norms of these $\eta$'s. Furthermore, we must have $\| \eta^{(t)}_i \| \in (\frac{1}{2},2)$. Define
\begin{align}
    \varepsilon_{t+1} = \frac{\iota}{(\alpha- \alpha^2) } \Big( \frac{\epsilon}{2} \frac{12}{K_t} \sum_{j=1}^{\frac{K_t}{12}} \eta^{(t)}_j \Psi^{(t)} \left( \eta^{(t)}_j \right)^{*}   \Big)^2\,.
    \label{equ_def_iter_varepsilon}
\end{align}
By \eqref{equ_vector_unit}, we have that
\begin{equation}\label{equ_epsilon_t_bound}
\varepsilon_{t+1} \in \big[ \frac{\epsilon^2 \iota}{4(\alpha-\alpha^2)}  (0.5 \varepsilon_t)^2, \frac{\epsilon^2 \iota}{4(\alpha-\alpha^2)}  (2 \varepsilon_t)^2 \big]\,.
\end{equation}
Next, we sample $\beta^{(t)}_k(j)$ as i.i.d.\ uniform variables on $\{-1, 1\}$ (see Remark~\ref{remark-sampling} for a minor modification for this) and we define
\begin{align}
    \sigma_k^{(t)}=\frac{1}{\sqrt{\frac{1}{12}K_t}} \sum_{j=1}^{\frac{1}{12}K_t} \beta^{(t)}_k(j)  \eta_j^{(t)} \mbox{ for } k=1,2,\ldots,K_{t+1}\,.
    \label{equ_def_sigma}
\end{align}
Write $\Hat{\beta}^{(t)}_k(j)= \sqrt{\eta^{(t)}_j \Psi^{(t)} {( \eta^{(t)}_j )}^{*}} \beta^{(t)}_k(j)$ and define $\Phi^{(t+1)},\Psi^{(t+1)}$ to be $K_{t+1}*K_{t+1}$ matrices such that
\begin{equation}
    \begin{aligned}
    & \Phi^{(t+1)}(i,j) = (\alpha-\alpha^2)^{-1} \left \{  \phi \left( \frac{12}{K_t} \langle {\beta}^{(t)}_i,{\beta}^{(t)}_j \rangle \right) - \alpha^2 \right \}  \,,  \\
    & \Psi^{(t+1)}(i,j)= (\alpha-\alpha^2)^{-1} \left \{ \phi \left(\frac{\epsilon}{2} \frac{12}{K_t} \langle \Hat{\beta}^{(t)}_i , \Hat{\beta}^{(t)}_j \rangle \right) - \alpha^2    \right \}\,.
    \label{equ_def_iter_matrix}
    \end{aligned}
\end{equation}
Later we will show that $\frac{12}{K_t} \langle {\beta}^{(t)}_i,{\beta}^{(t)}_j \rangle$ is the ``typical'' correlation between $\langle \sigma^{(t)}_i, D^{(t)}_v \rangle $ and $\langle \sigma^{(t)}_j, D^{(t)}_v \rangle$, and that $\frac{\epsilon}{2} \frac{12}{K_t} \langle \Hat{\beta}^{(t)}_i , \Hat{\beta}^{(t)}_j \rangle$ is the ``typical'' correlation between $\langle \sigma^{(t)}_i, D^{(t)}_v \rangle $ and $\langle \sigma^{(t)}_j, \mathsf{D}^{(t)}_{\pi(v)} \rangle$. Thus, we can expect $\mathrm{M}^{(t+1,t+1)}_{\Gamma}, \mathrm{M}^{(t+1,t+1)}_{\Pi}$ and $\mathrm{P}_{\Gamma,\Pi}^{(t+1,t+1)}$ to concentrate around $\Phi^{(t+1)}$, $\Phi^{(t+1)}$ and $\Psi^{(t+1)}$ respectively. Finally, we complete our iteration by setting
\begin{equation}
         \Gamma^{(t+1)}_k= \{  v \in V \backslash A :   | \langle \sigma^{(t)}_k,D^{(t)}_v\rangle |      \geq 10  \}\,; \,\Pi^{(t+1)}_k=  \{  \mathsf{v} \in \mathsf{V} \backslash \mathsf{A}_{\mathtt{m}} :  | \langle \sigma^{(t)}_k,\mathsf{D}^{(t)}_{\mathsf{v}} \rangle | \geq 10  \}\,.
        \label{equ_def_iter_sets}
\end{equation}

\subsection{Finishing}\label{sec:ending-iteration}

In this subsection we describe how we find the matching once we accumulate enough signal along the iteration. To this end, define
\begin{equation*}
    t^{*}=\min \{  t \geq 0: K_t \geq e^{ (\log \log n)^2 }\}\,.
\end{equation*}
By \eqref{equ_iter_K}, we have $K_t = (K_0^{2^t})/ (\varkappa^{2^t-1})$ and thus $t^* \sim 2 \log_2 \log \log n$. So we have that $K_{t^*} \leq e^{ 2 (\log \log n)^2 } \ll n^{0.0001}$. Recalling \eqref{equ_epsilon_t_bound}, we have $\varepsilon_t \geq (\frac{\iota \epsilon^2}{16(\alpha - \alpha^2)})^{2^t-1} \varepsilon_0^{2^t}$, so we may choose 
\begin{equation}\label{eq-kappa-choice}
K_0 = \kappa = \kappa (\epsilon) \geq 1000 (\alpha - \alpha^2)^2 \iota^{-2} \epsilon^{-4} \varepsilon_{0}^{-2} \varkappa
\end{equation} 
such that for an absolute constant $c>0$ (using $t^* \sim 2 \log_2 \log \log n$ again) 
\begin{align}
    K_{t^*} \varepsilon_{t^*}^2 \geq \Big( \frac{ K_0 \iota^2 \epsilon^4 \varepsilon^2_0 }{ 256 (\alpha-\alpha^2)^2 \varkappa } \Big)^{2^{t^*}} = \exp \{ c (\log \log n)^2 \}\,.
    \label{equ_estimation_K_t_Varepsilon_t}
\end{align}
For each $1 \leq \mathtt m \leq \mathtt M$, we run the procedure of initialization and then run the iteration up to time $t^{*}$, and then we construct a permutation $\pi_{\mathtt{m}}$ (with respect to  $\mathsf{A}_{\mathtt{m}}$) as follows. For $A= \{ u_1, \ldots, u_{K_0}  \}$ and $\mathsf{A}_{\mathtt{m}} = \{ \mathsf{u}_1 , \ldots, \mathsf{u}_{K_0}  \}$, set $\pi_{\mathtt{m}}(u_j) = \mathsf{u}_j$ for $1 \leq j \leq K_0$. For each $v \in V\setminus A$, we check $\mathsf v\in \mathsf V\setminus \mathsf A_{\mathtt m}$ in a prefixed ordering until we encounter the first $\mathsf v$ such that  
\begin{equation}
    \sum_{k=1}^{\frac{1}{12}K_{t^*}} \langle \eta^{(t^*)}_k, D^{(t^*)}_v \rangle \langle \eta^{(t^*)}_k, \mathsf{D}^{(t^*)}_{\mathsf{v}} \rangle
    \geq \frac{1}{100} K_{t^{*}} \varepsilon_{t^{*}}\,.
    \label{equ_def_matching_ver}
\end{equation}
If a desired $\mathsf v$ is found, then we set $\pi_{\mathtt m}(v) = \mathsf v$; otherwise, we set $\pi_{\mathtt m} = \emptyset$ which amounts to declaring failure for this corresponding $\mathsf A_{\mathtt m}$. 
Finally, we set 
\begin{align}
    \Hat{\pi} = \arg \max_{ \pi_{\mathtt{m}}: \pi_{\mathtt m} \neq \emptyset}  \{  \sum_{(u,v) \in E(V)} G_{u,v} \mathsf{G}_{\pi_{\mathtt{m}}(u),\pi_{\mathtt{m}}(v)} \}\,.
    \label{equ_def_final_pi_hat}
\end{align}
The success of of our algorithm is then guaranteed by the following theorem.
\begin{Theorem}{\label{main-thm}}
   With probability tending to 1 as $n\to \infty$, we have that $\Hat{\pi} = \pi$.
\end{Theorem}

\subsection{On the random sampling}\label{sec:sampling}
In the iterative steps we have sampled random vectors and we wish they generate enough cancellations for calculations later. This is incorporated in the following proposition.
\begin{Proposition}{\label{prop_random_vector}}
    In the $t$-th step of the iteration, with probability at least $0.5$ the random samples $\{ \beta^{(t)}_k (j) : 1 \leq k \leq K_{t+1}, 1 \leq j \leq \frac{1}{12}K_t \}$ we draw satisfy the following:
    \begin{align}
        &\frac{12}{K_t} \langle \beta^{(t)}_k, \beta^{(t)}_l \rangle \leq \frac{ 24 \sqrt{\log K_t}}{ \sqrt{K_t} } \mbox{ for } 1 \leq k < l \leq K_{t+1}\,,
        \label{equ_averaging_1} \\
        &\frac{12}{K_t} \langle \Hat{\beta}^{(t)}_k, \Hat{\beta}^{(t)}_l \rangle \leq \frac{24 \varepsilon_t \sqrt{\log K_t}}{ \sqrt{K_t} }, \mbox{ for } 1 \leq k < l \leq K_{t+1}\,,
        \label{equ_averaging_2} \\
        & \sum_{1\leq k \not = l \leq K_{t+1}} \left( \frac{12}{K_t} \langle \beta^{(t)}_k , \beta^{(t)}_l \rangle \right)^4 \leq  10000 \frac{K_{t+1}^2}{K^2_t}\,,
        \label{equ_sum_averaging_1}  \\
        & \sum_{1\leq k \not = l \leq K_{t+1}} \left( \frac{12}{K_t} \langle \Hat{\beta}^{(t)}_k , \Hat{\beta}^{(t)}_l \rangle \right)^4 \leq    100000  \varepsilon_t^{4} \frac{K_{t+1}^2}{K^2_t}\,.
        \label{equ_sum_averaging_2}
    \end{align}
\end{Proposition}

\begin{Remark}\label{remark-sampling}
    Since  $\Gamma^{(t)}_k,\Pi^{(t)}_k$ and $\Phi^{(t)},\Psi^{(t)}$ are accessible by our algorithm, we can resample $\beta$'s if any of the conditions in \eqref{equ_averaging_1}, \eqref{equ_averaging_2}, \eqref{equ_sum_averaging_1} or \eqref{equ_sum_averaging_2} is not satisfied. This will increase the sampling complexity by a constant factor thanks to Proposition \ref{prop_random_vector}. For this reason in what follows, we assume that we have performed resampling until all these conditions are satisfied. 
\end{Remark}
\begin{proof}[Proof of Proposition~\ref{prop_random_vector}]
We will upper-bound the probability for each of the desired conditions to be violated. For \eqref{equ_averaging_1}, note that $\sum_{j} \beta_k^{(t)}(j) \beta_l^{(t)}(j)$ is a sum of $\frac{1}{12}K_t$ i.i.d.\ random signs. By Azuma-Hoeffding inequality, we have
\begin{align*}
    \mathbb{P} \Big[\frac{12}{K_t} \sum_{n=1}^{\frac{1}{12}{K_t}}  \beta_k^{(t)} (j) \beta_l^{(t)} (j) > \frac{24 \sqrt{\log K_t}}{\sqrt{K_t}} \Big] & \leq
    2 \exp \Big \{ \frac{- (\frac{1}{12} K_t \cdot \frac{24 \sqrt{\log K_t}}{\sqrt{K_t}} )^2}{2 \cdot \frac{1}{12}K_t} \Big\}  \\
    & \leq 2 \exp \{  - 24 \log K_t  \} \leq K_{t+1}^{-4}\,.
\end{align*}
Also, a similar bound for \eqref{equ_averaging_2} can be derived by recalling \eqref{equ_vector_unit}. 

In addition, note that 
\begin{align*}
    & \mathbb{E} \Big \{ \sum_{1\leq k \not = l \leq K_{t+1}} \Big( \frac{12}{K_t} \langle \beta^{(t)}_k , \beta^{(t)}_l \rangle \Big)^4  \Big \} = \frac{12^4}{K_t^4} \sum_{1\leq k \not = l \leq K_{t+1}} \mathbb{E} \Big\{  \Big(  \langle \beta^{(t)}_k , \beta^{(t)}_l \rangle \Big)^4  \Big\} \\
    \leq & \frac{12^4 K_{t+1}^2}{K_t^4} \mathbb{E} \Big \{ \Big( \sum_{j=1}^{\frac{K_t}{12}} \beta^{(t)}_k(j) \beta^{(t)}_l(j) \Big)^4 \Big\} \leq \frac{12^4 K_{t+1}^2}{K_t^4} \cdot 8 \Big( \frac{K_t}{12} \Big)^2 \leq \frac{2000K_{t+1}^2}{K_t^2}\,.
\end{align*}
Thus by Markov's inequality the probability for \eqref{equ_sum_averaging_1} to fail is at most $0.2$.
Similarly,
\begin{align*}
    & \mathbb{E} \Big \{ \sum_{1\leq k \not = l \leq K_{t+1}} \Big(\frac{12}{K_t} \langle \Hat{\beta}^{(t)}_k , \Hat{\beta}^{(t)}_l \rangle \Big)^4  \Big\} = \frac{12^4}{K_t^4} \sum_{1\leq k \not = l \leq K_{t+1}} \mathbb{E} \Big\{  \Big(\langle \Hat{\beta}^{(t)}_k , \Hat{\beta}^{(t)}_l \rangle \Big)^4  \Big\} \\
    \leq & \frac{12^4 K_{t+1}^2}{K_t^4} \mathbb{E} \Big \{ \Big( \sum_{j=1}^{\frac{K_t}{12}} \eta^{(t)}_j \Psi^{(t)} \Big( \eta^{(t)}_j \Big)^{*} \beta^{(t)}_k(j) \beta^{(t)}_l(j) \Big)^4 \Big\} \leq \frac{20000 K_{t+1}^2 \varepsilon_{t}^{4} }{K_t^2}\,,
\end{align*}
where in the last inequality we recalled \eqref{equ_vector_unit}. So by Markov's inequality \eqref{equ_sum_averaging_2} fails with probability at most 0.2. The desired result then follows by a simple union bound.
\end{proof}

\subsection{Formal description of the algorithm} \label{sec:formal-algorithm}
We are now ready to present our algorithm formally. 
\begin{breakablealgorithm}
			\label{algo:matching}
			\caption{Gaussian Matrix Matching Algorithm}
	\begin{algorithmic}[1]
		\STATE Define $\Hat{G}, \Hat{\mathsf{G}}, A, \phi, \mathtt{M}, \iota, \alpha, \varkappa, \kappa$ and $\Phi^{(0)}, \Psi^{(0)}$ as above.
		\STATE List all sequences with $\kappa$ distinct elements in $\mathsf{V}$ by $\mathsf{A}_1, \mathsf{A}_2, \ldots, \mathsf{A}_{\mathtt{M}}$.
		\FOR{$\mathtt{m}=1, \ldots, \mathtt{M}$}
		\STATE Define $\Gamma^{(0)}_k, \Pi^{(0)}_k$  for $1 \leq k \leq K_0$ as in \eqref{equ_initial_set}.
		\STATE Define $\varepsilon_0, K_0$ as above.
		\STATE Set $\pi_{\mathtt{m}}(v_j) = \mathsf{v}_j$ where $v_j, \mathsf{v}_j$ are the $j$-th coordinate of $A, \mathsf{A}_{\mathtt{m}}$ respectively.
		\WHILE{ $K_t \leq \exp \{ (\log \log n)^2 \}$ }
		\STATE Calculate $K_{t+1}$ according to \eqref{equ_iter_K}.
		\STATE Calculate $\mathrm{M}^{(t,s)}_{\Gamma}, \mathrm{M}^{(t,s)}_{\Pi}$ for $0 \leq s \leq t$ according to \eqref{equ_martix_M_P}.
		\STATE Calculate the eigenvalues and eigenvectors of $\Phi^{(t)}, \Psi^{(t)}$, as in \eqref{eq-spectral-decomposition}.
		\STATE Define $\eta^{(t)}_1, \eta^{(t)}_2, \ldots, \eta^{(t)}_{ \frac{K_t}{12} }$ according to \eqref{equ_vector_orthogonal} and \eqref{equ_vector_unit}.
		\STATE Calculate $\varepsilon_{t+1}$ according to \eqref{equ_def_iter_varepsilon}.
		\STATE Sample random vectors $\beta^{(t)}$ satisfying \eqref{equ_averaging_1}, \eqref{equ_averaging_2}, \eqref{equ_sum_averaging_1} and \eqref{equ_sum_averaging_2}.
		\STATE Define $\sigma^{(t)}_k$ for $1 \leq k \leq K_{t+1}$ according to \eqref{equ_def_sigma};
		\STATE Define $\Phi^{(t+1)}, \Psi^{(t+1)}$ according to \eqref{equ_def_iter_matrix}; 
		\STATE Define $\Gamma^{(t+1)}_k, \Pi^{(t+1)}_k$ for $1 \leq k \leq K_{t+1}$ according to \eqref{equ_def_iter_sets}; 
		\ENDWHILE
		\STATE Suppose we stop at $t=t^{*}$;
		\STATE Define $\eta^{(t)}_1, \eta^{(t)}_2, \ldots, \eta^{(t)}_{ \frac{K_{t^{*}}}{12} }$ according to \eqref{equ_vector_orthogonal} and \eqref{equ_vector_unit}.
		\FOR{ $u \in V \backslash A$ }
		\STATE Define $\textup{SUC}_u =0$;
		\FOR{ $\mathsf{u} \in \mathsf{V} \backslash \mathsf{A}_{\mathtt{m}}$ }
		\IF{ $u$ and $\mathsf{u}$ satisfy (\ref{equ_def_matching_ver})}
		\STATE Define $\pi_{\mathtt{m}}(u) = \mathsf{u}$;
		\STATE Set $\textup{SUC}_u =1$;
		\ENDIF
		\ENDFOR
		\IF{ $\textup{SUC}_u =0$ }
		\STATE Set $\pi_{\mathtt{m}} = \emptyset$ and break the {\bf for} cycle.
		\ENDIF
		\ENDFOR
		\ENDFOR 
		\IF{ there exists a $\pi_{\mathtt{m}} \not = \emptyset$ } 
		\STATE Find ${\pi}_{\mathtt{m}^{*}}$ which maximizes $\sum_{(u,v) \in E(V)} G_{u,v} \mathsf{G}_{\pi(u), \pi(v)}$ among $\{\pi_{\mathtt{m}} \not = \emptyset : 1 \leq \mathtt{m} \leq \mathtt{M}\}$.
		\RETURN $\Hat{\pi}= {\pi}_{\mathtt{m}^{*}}$.
		\ELSE
		 \RETURN $\textup{FAIL}$;
		\ENDIF
	\end{algorithmic}
\end{breakablealgorithm}

\subsection{Proof of Lemma~\ref{lemma_matrix_eigenvalue}} \label{sec:matrix-rank}
We prove Lemma~\ref{lemma_matrix_eigenvalue} in this subsection, which then justifies that our algorithm is well-defined. The case when $t=0$ is trivial. In what follows, we prove the statement for $t+1$ assuming it holds up to $t$. For notation convenience, we drop $t+1$ from the superscript of $\Phi$ and $\Psi$ in this subsection.
Recall \eqref{equ_def_iter_matrix}.
Since the matrix $\mathsf{J}$ has rank 1, as we will see later the key is to study the matrices $\Tilde{\Phi} = \Phi + \frac{ \alpha^2 }{ \alpha - \alpha^2 } \mathsf{J}$ and $\Tilde{\Psi} = \Psi + \frac{ \alpha^2 }{ \alpha - \alpha^2 } \mathsf{J}$. Recall that $\phi(\beta_i,\beta_j)=\phi(\langle {\beta}_i, {\beta}_j \rangle)$, where
\begin{equation}
    \phi(u)= \int_{|x| \geq 10} \int_{|y| \geq 10} \frac{1}{2\pi \sqrt{1-u^2}} e^{-\frac{x^2-2uxy+y^2}{2(1-u^2)}} dxdy\,.
    \label{equ_intergration_form_phi}
\end{equation}
Denote the Taylor expansion of $\phi$ as $\phi(u) = \sum_{k=0}^{\infty} c_k u^k$. Recall from \eqref{equ_def_iter_varepsilon} that 
$$ \frac{12}{K_t} \| \beta^{(t)}_i \|^2 = 1 \mbox{ and }\frac{\epsilon}{2} \frac{12}{K_t} \| \Hat{\beta}^{(t)}_i \|^2 = \frac{\epsilon}{2} \frac{12}{K_t} \sum_{j} \eta^{(t)}_j \Psi^{(t)} (\eta^{(t)}_j)^{*} = \sqrt{\frac{\varepsilon_{t+1} (\alpha - \alpha^2)}{\iota}}\,.$$ So we can write $\Tilde{\Phi}$ and $\Tilde{\Psi}$ as
\begin{align}\label{eq-tilde-Phi-Psi-expression}
    \Tilde{\Phi} = \frac{1}{\alpha-\alpha^2} \Big(\alpha \mathrm{I} +  \sum_{k=0}^{\infty} c_k \Phi_k\Big) \mbox{ and } \Tilde{\Psi} = \frac{1}{\alpha - \alpha^2} \Big(\phi \Big( \sqrt{\frac{\varepsilon_{t+1} (\alpha - \alpha^2)}{\iota}} \Big)  \mathrm{I} + \sum_{k=0}^{\infty} c_k \Psi_k\Big)\,,
\end{align}
where $\Phi_k$ and $\Psi_k$ are matrices with diagonal entries  being 0, and non-diagonal entries given by $\left( \frac{12}{K_t} \langle \beta^{(t)}_i, \beta^{(t)}_j \rangle \right)^k$ and $\left( \frac{\epsilon}{2}\frac{12}{K_t} \langle \Hat{\beta}^{(t)}_i, \Hat{\beta}^{(t)}_j \rangle \right)^k$, respectively. Next we prove several claims characterizing $c_k$ and $\phi$. Clearly we have $c_0 = \phi(0) = \alpha^2$.
\begin{Claim}{\label{claim_Taylor_c_1}}
    We have $c_1 = 0$.
\end{Claim}
\begin{proof} 
Applying Taylor's expansion repeatedly, we get that
\begin{align*}
    & \frac{1}{2\pi \sqrt{1-u^2}}    e^{-\frac{x^2-2uxy+y^2}{2(1-u^2)}} \\
    =&  \frac{1}{2\pi \sqrt{1-u^2}}  e^{\frac{1}{2(1-u^2)}(-u^2(x^2+y^2)+2uxy)} e^{-\frac{x^2+y^2}{2}}  \\
    =&  \frac{1}{2 \pi} \sum_{z=0}^{\infty} \frac{(2z-1)!!} {(2z)!!} u^{2z} \sum_{p=0}^{\infty} \frac{1}{p!} \left(  \frac{uxy}{ 1-u^2 } - \frac{u^2(x^2+y^2)}{2(1-u^2)}   \right)^p e^{- \frac{x^2+y^2}{2}}  \\
    =& \frac{1}{2 \pi} \sum_{z=0}^{\infty} \frac{(2z-1)!!} {(2z)!!} u^{2z} \sum_{p=0}^{\infty} \frac{1}{p!} \sum_{k=0}^{p} \frac{\binom{p}{k} (uxy)^k \left( -\frac{u^2 (x^2+y^2)}{2} \right)^{p-k} }{(1-u^2)^p} e^{- \frac{x^2+y^2}{2}}  \\
    =&  \sum_{z=0}^{\infty}  \sum_{p=0}^{\infty}  \sum_{k=0}^{p} \sum_{r=0}^{\infty} \frac{1}{2 \pi} \frac{\binom{p}{k}}{p!} \frac{(2z-1)!!} {(2z)!!} (xy)^k (-\frac{x^2+y^2}{2})^{l} \binom{p+r-1}{r} u^{2p - k +2z +2r} e^{- \frac{x^2+y^2}{2}} \\
    =& \sum_{z=0}^{\infty}  \sum_{\ell=0}^{\infty} \sum_{k=0}^{\infty}  \sum_{r=0}^{\infty} \frac{1}{2 \pi} \frac{\binom{l+k}{k}}{(l+k)!} \frac{(2z-1)!!} {(2z)!!} (xy)^k (-\frac{x^2+y^2}{2})^{l} \binom{k+l +r-1}{r} u^{k +2l +2z +2r} e^{- \frac{x^2+y^2}{2}} \,.
\end{align*}
So we have $c_m = \int_{|x| \geq 10} \int_{|y| \geq 10} f_m(x,y) e^{- \frac{x^2+y^2}{2}} dxdy$ for $m\geq 1$, where $ f_m(x, y) $ is given by
\begin{align}\label{eq-f-m-expression}
   \sum_{ k + 2l + 2z + 2r = m } \frac{1}{2\pi (k+l)!} \frac{(2z-1)!!}{(2z)!!} \binom{k+l}{k} \binom{k+l+r-1}{r} (xy)^{k} \left(- \frac{x^2 + y^2}{2} \right)^{l}\,.
\end{align}
Therefore,
\begin{equation*}
    c_1 = \frac{1}{2\pi} \int_{|x| \geq 10} \int_{|y| \geq 10} xy e^{- \frac{x^2+y^2}{2}} dxdy =0\,. \qedhere
\end{equation*}
\end{proof}

\begin{Claim}{\label{claim_Taylor_c_n}}
    We have $|c_m| \leq (m+1) 4^m$ for $m\geq 1$.
\end{Claim}
\begin{proof} 
By \eqref{eq-f-m-expression}, we have
\begin{align*}
    |f_m(x, y)| \leq \sum_{k + 2l + 2z + 2r = m} \frac{1}{2\pi (k+l)!} \binom{k+l}{k} \binom{k+l+r-1}{r} \left( \frac{x^2+y^2}{2} \right)^{k+l}\,.
\end{align*}
Note that
\begin{align*}
    & \frac{1}{2\pi (k+l)!} \int_{|x| \geq 10} \int_{|y| \geq 10} \left( \frac{x^2+y^2}{2} \right)^{k+l} e^{- \frac{x^2 + y^2}{2}} dxdy \\
    \leq & \frac{1}{2\pi(k+l)!} \int_{\mathbb{R}} \int_{\mathbb{R}} \left( \frac{x^2+y^2}{2} \right)^{k+l} e^{- \frac{x^2 + y^2}{2}} dxdy =  \frac{1}{2\pi(k+l)!} 2 \pi (k+l+1)! = k+l+1\,.
\end{align*}
Therefore,
\begin{align*}
    |c_m| & \leq \int_{|x| \geq 10} \int_{|y| \geq 10} |f_m(x, y)| e^{- \frac{x^2 + y^2}{2}} dxdy \\
    & \leq \sum_{k + 2l + 2z + 2r = m} (k+l+1) \binom{k+l}{k} \binom{k+l+r-1}{r} \leq (m+1) 4^m\,. \qedhere
\end{align*}
\end{proof}

\begin{Claim}{\label{claim_difference_phi}}
 For an absolute constant $u_0>0$ we have $\phi(u) - \phi(0) \in (0.99  \iota u^2, 1.01\iota u^2)$ for $|u|\leq u_0$.
\end{Claim}
\begin{proof}
Recall $c_1 =0$ and $\iota = \frac{1}{2} \phi''(0) = c_2$. By Taylor's expansion, 
\begin{align*}
    \phi(u) - \phi(0) = \iota u^2 + O(u^3), \mbox{ when } u\to 0\,.
\end{align*}
So it suffices to show $\iota = c_2 > 0$. By \eqref{eq-f-m-expression}, we have
\begin{equation*}
    c_2 = \frac{1}{2\pi} \int_{|x| \geq 10} \int_{|y| \geq 10} \left( \frac{x^2 y^2}{2} - \frac{x^2 + y^2}{2} + 1 \right) e^{- \frac{x^2 + y^2}{2}} dxdy > 0 \,.\qedhere
\end{equation*}
\end{proof}

\begin{Remark}\label{rem-epsilon-u-0}
    For simplicity of our proof, in what follows we wish to assume $\epsilon \leq u_0$ as required by Claim \ref{claim_difference_phi}. Indeed, if $\epsilon > u_0$ we can deliberately add i.i.d.\ noise to each $\{ \Hat{G}_{u,v} \}$ and $\{ \Hat{\mathsf{G}}_{\mathsf{u,v}} \}$ in the step of preprocessing such that the correlation between the modified $\Hat{G}$ and $\Hat{\mathsf{G}}$ becomes $u_0$.  That being said, we also point out that the above assumption is nonessential and is only for technical simplicity. Actually, we can modify our  algorithm by setting $\varepsilon_{t+1} = \frac{1}{\alpha - \alpha^2} ( \phi( \frac{\epsilon}{2} \frac{12}{K_t} \sum_{j=1}^{\frac{K_t}{12}} \eta^{(t)}_j \Psi^{(t)} (\eta^{(t)}_j )^{*} ) - \phi(0) )$. This would require a slightly more complicated analysis and more importantly would result in even more cumbersome notation, which is why formally we made the assumption of $\epsilon \leq u_0$ as explained above. 
 \end{Remark}

We also need some standard lemmas in linear algebra.
\begin{Lemma}{\label{lemma_symmetric_decomposition}}
    For any $d*d$ symmetric matrix $A$, recall that we denote by $\varsigma_k(A)$ its $k$-th largest eigenvalue. Suppose $| \{ k : |\varsigma_k(A)| \geq M  \} | \leq m$.  Then $A$ can be decomposed as $A=B+C$, where $B,C$ are symmetric matrices such that $\mathrm{rank}(B) \leq m$ and $\| C \|_{\mathrm{op}}  \leq M$.
\end{Lemma}
\begin{proof}
Since $A$ is symmetric, $A$ can be diagonalized by an orthogonal matrix $U$, i.e. we can write $A= U D U^{*}$, where $D=\mathrm{diag}(\varsigma_1, \varsigma_2, \ldots, \varsigma_d)$ and $\varsigma_1 \geq \varsigma_2 \geq \ldots \geq \varsigma_d$ are the ordered eigenvalues of $A$. Then we can set $D_1 = \mathrm{diag}( \delta_1 \varsigma_1, \delta_2 \varsigma_2, \ldots, \delta_d \varsigma_d)$ and $D_2 = D-D_1$, where $\delta_k \in \{ 0,1 \}$ and $\delta_k=1$ if and only if $|\varsigma_k| \geq M$. It is straightforward to verify that $B = U D_1 U^{*}$ and $C= U D_2 U^{*}$ satisfy required properties.
\end{proof}

\begin{Lemma}{\label{lemma_sum_symmetric_eigenvalue}}
    If $A,B$ are $d*d$ symmetric matrices with $\mathrm{rank}(B) \leq k$, then we have $\varsigma_{m-k}(A+B) \geq \varsigma_{m}(A)$ and $\varsigma_{m-k}(A) \geq \varsigma_{m}(A+B)$ for $k<m$.
\end{Lemma}
\begin{proof}
Consider an orthogonal diagonalization $A=U D U^{*}$, where $D=\mathrm{diag}( \varsigma_1, \varsigma_2, \ldots, \varsigma_d)$. Set $D_1 = \mathrm{diag}( \varsigma_1, \varsigma_2, \ldots, \varsigma_m,  \varsigma_m, \ldots, \varsigma_m )$, $D_2 = D_1 - D$ and $C= U D_1 U^{*}, H = U D_2 U^{*}$. Thus, $C \geq \varsigma_m \mathrm{I}$ (i.e., $C - \varsigma_m \mathrm{I}$ is semi-positive definite)  and $\mathrm{rank}(H) \leq d-m$. So $A+B= C+(B-H)$, where $C \geq \varsigma_m \mathrm{I}$ and $\mathrm{rank}(B-H) \leq k+d-m$. Since $A+B \geq \varsigma_m \mathrm{I} + (B-H)$, we have
\begin{align*}
    \varsigma_{m-k}(A+B) \geq \varsigma_{m-k}(\varsigma_m \mathrm{I} + (B-H)) \geq \varsigma_m = \varsigma_m(A)\,,
\end{align*}
where the last inequality holds since $\mathrm{rank}(B-H) \leq d-m+k$ and thus $\varsigma_m \mathrm{I} + (B-H)$ has at least $(m-k)$ eigenvalues (which are all equal to) $\varsigma_m$. Replacing $A,B$ with $A+B,-B$ we see that the second inequality also holds.
\end{proof}

\begin{Lemma}{\label{lemma_rank_perturbed}}
    Let $A=(a_{ij})_{1 \leq i,j \leq n}$ be a $d*d$ symmetric matrix with $a_{ii}=0$ and $ \sum_{i,j} a_{ij}^2 < \frac{d \delta^2 }{10^6 }$. Then at least $0.99d$ eigenvalues of $A$ are in $[-0.01 \delta ,0.01 \delta ]$.
\end{Lemma}
\begin{proof}
We adopt the same approach as in \cite{Alon09}. Denote $\lambda_1,\lambda_2,\ldots,\lambda_d$ the eigenvalues of $A$. Then we have $\sum_{i=1}^d \lambda_i= 0$ and $\sum_{i=1}^d \lambda_i^2 =\sum_{i,j=1}^d a_{ij}^2 \leq 10^{-6}d \delta^2$. Thus we can conclude by Chebyshev's inequality
\begin{equation*}
    \left| \{ i:\lambda_i \in [-0.01 \delta ,0.01 \delta ] \} \right| \geq d- \frac{10^{-6}d \delta^2}{(0.01 \delta)^2} = 0.99d\,. \qedhere
\end{equation*}
\end{proof}
We are now finally ready to provide the proof of Lemma~\ref{lemma_matrix_eigenvalue}.
\begin{proof}[Proof of Lemma~\ref{lemma_matrix_eigenvalue}]
We first consider $\Phi$.
By \eqref{eq-tilde-Phi-Psi-expression} and Claim~\ref{claim_Taylor_c_1}, we can write $\Tilde{\Phi}$ as
\begin{align*}
    & \Tilde{\Phi} = \frac{1}{\alpha - \alpha^2} \left( \phi(1) - \phi(0) \right) \mathrm{I} + \frac{1}{\alpha - \alpha^2} c_0 \mathsf{J} + \frac{1}{\alpha - \alpha^2} \sum_{k=2}^{\infty} c_k \Phi_k  \,.
\end{align*}
Since $\phi(1) = \alpha$ and $c_0 = \phi(0) = \alpha^2$, we have
\begin{align*}
    \Tilde{\Phi} = \mathrm{I} + \frac{ \alpha^2 }{\alpha - \alpha^2} \mathsf{J} + \frac{1}{\alpha - \alpha^2} \sum_{k=2}^{\infty} c_k \Phi_k\,.
\end{align*}
Recalling $\Phi = \Tilde{\Phi} - \frac{\alpha^2}{\alpha - \alpha^2} \mathsf{J}$, we have $\Phi = \mathrm{I} + \frac{1}{\alpha - \alpha^2} \sum_{k=2}^{\infty} c_k \Phi_k$. By Proposition \ref{prop_random_vector}, we have
\begin{align*}
    \sum_{i,j} (\frac{c_2}{\alpha - \alpha^2} \Phi_2(i,j))^2 \leq \sum_{i \neq j} \left( \frac{c_2}{\alpha - \alpha^2} \right)^2 \left( \frac{12}{K_t} \langle \beta^{(t)}_i, \beta^{(t)}_j \rangle  \right)^4 \leq \frac{10^8}{ (\alpha - \alpha^2)^2} \frac{K_{t+1}^2}{K^2_t}\,.
\end{align*}
By  our choice that $K_{t+1} = \frac{1}{\varkappa} K_t^2$ (recall \eqref{equ_iter_K}) and Lemma \ref{lemma_rank_perturbed}, we see that
\begin{equation}\label{eq-Phi-2-eigenvalue-bound}
|\{ l:  |\varsigma_l( \Phi_2 )| \geq 0.01  \} | \leq 0.01K_{t+1}\,.
\end{equation}
Similarly, by Proposition \ref{prop_random_vector} and Claim~\ref{claim_Taylor_c_n}, we have
$$\|\frac{1}{\alpha - \alpha^2} \sum_{k=3}^{\infty} c_k \Phi_k \|_\infty \leq \sum_{k=3}^{\infty} \frac{(k+1)4^k}{\alpha-\alpha^2} \left( \frac{ 24 \sqrt{ \log K_t}}{\sqrt{K_t}}  \right)^k \leq \frac{ 10^6 (\log K_t)^2 }{ K_t^{1.5} }\,.$$
Thus,
$$\|\frac{1}{\alpha - \alpha^2} \sum_{k=3}^{\infty} c_k \Phi_k \|^2_{\mathrm{HS}} \leq K_{t+1}^2 \frac{10^{12} (\log K_t)^4}{K_t^3} \leq 10^{-6} K_{t+1}$$
where we have used \eqref{equ_iter_K} again.
By Lemma \ref{lemma_rank_perturbed}, we have
\begin{equation}\label{eq-Phi-3-eigenvalue-bound}
    |\{  l: |\varsigma_l ( \sum_{k=3}^{\infty} \frac{c_k}{\alpha - \alpha^2} \Phi_k )| \geq 0.01   \} | \leq 0.01 K_{t+1}\,.
\end{equation}
 
Applying Lemma \ref{lemma_symmetric_decomposition} with \eqref{eq-Phi-2-eigenvalue-bound} and \eqref{eq-Phi-3-eigenvalue-bound} , we can write $\Phi_2 = C_1 + D_1$ and $\sum_{k=3}^{\infty} \frac{c_k}{\alpha - \alpha^2} \Phi_k = C_2 + D_2$, where $\| C_1 \|_{\mathrm{op}}, \| C_2 \|_{\mathrm{op}} \leq 0.01$ and $\mathrm{rank}(D_1), \mathrm{rank}(D_2) \leq 0.01 K_{t+1}$. Noting $\Phi = (\mathrm{I} + C_1 + C_2) + (D_1 + D_2)$, we apply Lemma \ref{lemma_sum_symmetric_eigenvalue} and get that
\begin{align*}
    & \varsigma_{0.98 K_{t+1}} (\Phi) \geq \varsigma_{K_{t+1}} ( \mathrm{I} + C_1 + C_2 ) \geq 0.98\,,  \\
    & \varsigma_{0.02 K_{t+1}+1} (\Phi) \leq \varsigma_{1} ( \mathrm{I} + C_1 + C_2 ) \leq 1.02\,.
\end{align*}
This shows that $\Phi$ has at least $0.96K_{t+1}$ eigenvalues  in $(0.98,1.02)$.

We deal with $\Psi$ in a similar way. By \eqref{eq-tilde-Phi-Psi-expression} and Claim~\ref{claim_Taylor_c_1}, we can write $\Tilde{\Psi}$ as
$$ \Tilde{\Psi} = \frac{1}{\alpha - \alpha^2} \left( \phi \left( \sqrt{\frac{\varepsilon_{t+1} (\alpha - \alpha^2)}{\iota}} \right) - \phi(0) \right) \mathrm{I} + \frac{ \alpha^2 }{\alpha - \alpha^2} \mathsf{J} + \frac{1}{\alpha - \alpha^2} \sum_{k=2}^{\infty} c_k \Psi_k\,.$$
By Claim~\ref{claim_difference_phi} we know that $\phi \left( \sqrt{\frac{ \varepsilon_{t+1} (\alpha - \alpha^2) }{\iota} }\right) - \phi(0) \in ( 0.99 \varepsilon_{t+1} (\alpha - \alpha^2), 1.01 \varepsilon_{t+1} (\alpha - \alpha^2) )$. Recalling $\Psi = \Tilde{\Psi} - \frac{\alpha^2}{\alpha - \alpha^2} \mathsf{J}$, we can write $\Psi$ as 
\begin{align*}
    \Psi = \omega \mathrm{I} + \sum_{k=2}^{\infty} \frac{c_k}{ \alpha- \alpha^2 } \Psi_k\,,
\end{align*}
where $\omega \in (0.99 \varepsilon_{t+1}, 1.01 \varepsilon_{t+1})$. Again by Proposition \ref{prop_random_vector}, we have
\begin{align*}
    \sum_{i,j} (\frac{c_2}{\alpha - \alpha^2} \Psi_2(i,j))^2 \leq \sum_{i \neq j} \left( \frac{c_2}{\alpha - \alpha^2} \right)^2 \left( \frac{\epsilon}{2} \frac{12}{K_t} \langle \Hat{\beta}^{(t)}_i, \Hat{\beta}^{(t)}_j \rangle  \right)^4  \\
    \leq \frac{10^{9} \varepsilon^4_t}{ (\alpha - \alpha^2)^2} \frac{K_{t+1}^2}{K^2_t} \overset{\eqref{equ_epsilon_t_bound}}{\leq} \frac{10^{12} \varepsilon^2_{t+1}}{\epsilon^4 \iota^2 } \frac{K_{t+1}}{\varkappa}  \leq 10^{-6} \varepsilon^2_{t+1} K_{t+1}\,,
\end{align*}
where we have used \eqref{equ_iter_K} again.
By Lemma \ref{lemma_rank_perturbed}, $\frac{c_2}{\alpha - \alpha^2} \Psi_2$ has at most $0.01 K_{t+1}$ eigenvalues with absolute values larger than $0.01 \varepsilon_{t+1}$. By Proposition \ref{prop_random_vector} and Claim~\ref{claim_Taylor_c_n}, 
$$\|\frac{1}{\alpha - \alpha^2} \sum_{k=3}^{\infty} c_k \Psi_k \|_\infty \leq \sum_{k=3}^{\infty} \frac{(k+1)4^k }{\alpha-\alpha^2} \left( \frac{\epsilon}{2} \frac{ 24  \varepsilon_t \sqrt{ \log K_t}}{\sqrt{K_t}}  \right)^k \leq \frac{ 10^6 \varepsilon_t^3 (\log K_t)^2 }{ K_t^{1.5} }\,.$$
Thus,
$$\|\frac{1}{\alpha - \alpha^2} \sum_{k=3}^{\infty} c_k \Psi_k \|^2_{\mathrm{HS}} \leq K_{t+1}^2 \frac{10^{12} \varepsilon_t^6 (\log K_t)^4}{K_t^3} \leq 10^{-6} K_{t+1} \varepsilon_{t+1}^2\,.$$
By Lemma \ref{lemma_rank_perturbed} the matrix $\sum_{k=3}^{\infty} \frac{c_k}{\alpha- \alpha^2} \Psi_k$ has at most $0.01K_{t+1}$ eigenvalues with absolute values larger than $0.01\varepsilon_{t+1}$. By Lemma \ref{lemma_symmetric_decomposition}, we can write $\frac{c_2}{\alpha - \alpha^2}\Psi_2 = C_1 + D_1$ and $\sum_{k=3}^{\infty} \frac{c_k}{\alpha- \alpha^2} \Psi_k = C_2 + D_2$, where $\| C_1  \|_{\mathrm{op}}, \| C_2  \|_{\mathrm{op}} \leq 0.01 \varepsilon_{t+1}$ and $\mathrm{rank}(D_1), \mathrm{rank}(D_2) \leq 0.01 K_{t+1}$. By Lemma \ref{lemma_sum_symmetric_eigenvalue}, we know $\Psi = ( \omega \mathrm{I} + C_1 + C_2 ) + (D_1 + D_2)$ satisfies $\varsigma_{0.98K_{t+1}} (\Psi) \geq 0.97 \varepsilon_{t+1}$ and $\varsigma_{0.02K_{t+1}+1} (\Psi) \leq 1.03 \varepsilon_{t+1}$. This completes the proof of the lemma.
\end{proof}

\subsection{Running time analysis}\label{sec:runtime-analysis}
In this subsection, we show that Algorithm~\ref{sec:formal-algorithm} runs in polynomial time. 
\begin{Proposition}{\label{prop_time_complexity}}
     The running time for computing each $\pi_{\mathtt m}$ is $O(n^{2+o(1)})$. Furthermore, the running time for Algorithm \ref{sec:formal-algorithm}  is $O(n^{\kappa +2 + o(1)})$.
\end{Proposition}
\begin{proof} 
We first prove the first claim. We can compute $\Gamma^{(0)}_k, \Pi^{(0)}_k$ in $O( \kappa n )$ time. In addition, the iteration has $O( \log \log \log n )$ steps, and in each step for $t \leq t^*$ the running time can be bounded as follows: calculating $\mathrm{M}_{\Gamma}, \mathrm{M}_{\Pi}$ takes $O(  K_t (\sum_{s \leq t-1} K_s) n )$ time; calculating $\{D^{(t)}_v, \mathsf{D}^{(t)}_{\mathsf{v}}\}$ takes $O(K_t n^2)$ time; calculating the spectral decomposition of $\Phi^{(t)}$ and $\Psi^{(t)}$ takes $O(K_t^3)$ time; choosing $\eta$'s and calculating $\sigma$'s take $O(K_t^3)$ time; calculating $\Gamma^{(t+1)},\Pi^{(t+1)}$ (when $t \leq t^*-1$) takes $O(K_{t+1} n)$ time. Furthermore, in the finishing step calculating $\pi_{\mathtt{m}}$ takes $O(K_{t+1}^2 n^2)$ time. Therefore, the total amount of time spent on computing each $\pi_{\mathtt{m}}$ is upper-bounded by
\begin{align*}
    O(\kappa n) + \sum_{t \leq t^*}  O(K_{t} n^2)  + O(K_{t^*}^2 n^2) = O(n^{2 + o(1)})\,.
\end{align*}
We now prove the second claim. Since $\mathtt M \leq n^{\kappa}$, the running time for computing all $\pi_{\mathtt{m}}$ is $O(n^{\kappa + 2 +o(1)})$. In addition, finding $\Hat{\pi}$ from $\{ \pi_{\mathtt{m}} \}$ takes $O(n^{\kappa + 2})$ time. So the total running time is $O(n^{\kappa +2 + o(1)})$.
\end{proof}

We complete this section by pointing out that Theorem~\ref{thm-main} follows directly from Theorem~\ref{main-thm} and Proposition~\ref{prop_time_complexity}.

\section{Analysis of the algorithm}
The main goal of this section is to prove Theorem~\ref{main-thm}, and the crucial input is the following proposition. We say a pair of sequences $A = (u_1,u_2,\ldots, u_{K_0})$ and $\mathsf A = (\mathsf{u}_1, \mathsf{u}_2, \ldots, \mathsf{u}_{K_0} )$ is a good pair if
\begin{equation}\label{eq-correct-seeds}
\mathsf{u}_j = \pi(u_j) \mbox{ for } 1 \leq j \leq K_0\,.
\end{equation}
\begin{Proposition}{\label{main-prop}}
For a pair of sequences $(A, \mathsf A)$, define $\pi(A, \mathsf A) = \pi_{\mathsf m}$ if $\mathsf A = \mathsf A_{\mathtt m}$. 
    If $(A, \mathsf A)$ is a good pair, then
    \begin{align*}
        \mathbb{P}[\pi(A, \mathsf A) = \pi ] \geq 1 - o(1)\,.
    \end{align*}
\end{Proposition}
Recall that for each pair $(A, \mathsf A_{\mathtt m})$ with $1\leq \mathtt m\leq \mathtt M$, Algorithm~\ref{sec:formal-algorithm} outputs a $\pi_{\mathtt m} = \pi(A, \mathsf A_{\mathtt m})$. Note that there is one (and only one) good pair, and the corresponding output is the true matching with probability tending to 1 by Proposition~\ref{main-prop}. At this point, we apply \cite[Theorem 1]{WXY21+} which in particular says that with probability tending to 1, the maximizer of $\max_{\pi^{\prime}} \{ \sum_{(u,v) \in E(V)} G_{u,v} \mathsf{G}_{\pi^{\prime}(u), \pi^{\prime}(v)}\}$ is unique and is the true matching $\pi$. Therefore, Theorem~\ref{main-thm} follows.

The rest of the paper is devoted to the proof of Proposition~\ref{main-prop}.

\subsection{Outline of proof}\label{sec:analysis-outline}

We fix a good pair $(A, \mathsf A)$.
The basic intuition is that each pair of $( \Gamma^{(t)}_k,\Pi^{(t)}_k)$ carries signal of strength at least $\varepsilon_t$, and thus the total signal strength of all $K_t$ pairs will grow in $t$ (recall \eqref{equ_estimation_K_t_Varepsilon_t}). Of course, this also requires to show that the signals carried by different pairs are essentially non-repetitive. In order to prove this, we pose the following \emph{admissible} conditions on $( \Gamma^{(t)}_k,\Pi^{(t)}_k)$ and we hope that this will allow us to verify these admissible conditions by induction. 

\begin{Definition}{\label{def-admissible}}
    For $t \geq 0$ and a collection of pairs $( \Gamma^{(s)}_k, \Pi^{(s)}_k)_{1 \leq k \leq K_s,0 \leq s \leq t}$ with $\Gamma^{(s)}_k \subset V \backslash A$ and  $\Pi^{(s)}_k \subset \mathsf{V} \backslash \mathsf{A}$, we say $( \Gamma^{(s)}_k,\Pi^{(s)}_k )_{1 \leq k \leq K_s,0 \leq s \leq t}$ is $t$-admissible if the following hold:
    \begin{enumerate}
        \item[(i.)] $\Big|\frac{|\Gamma^{(s)}_k|}{n} -\alpha  \Big |<n^{-0.1}(\log n)^{10s} \prod_{i \leq s} K_i^{100}$ for  $0 \leq s \leq t$;
        \item[(ii.)] $\Big| \frac{|\Pi^{(s)}_k|}{n} -\alpha  \Big |<n^{-0.1}(\log n)^{10s} \prod_{i \leq s} K_i^{100}$ for $0 \leq s \leq t$ ;
        \item[(iii.)] $\Big | \frac{|\Gamma^{(s)}_k \cap \Gamma^{(s)}_{l}|}{n} -\phi( \frac{12}{K_{s-1}}  \langle {\beta}^{(s-1)}_k ,{\beta}^{(s-1)}_l \rangle) \Big| < n^{-0.1}(\log n)^{10s}\prod_{i \leq s} K_i^{100}$ for  $1 \leq s \leq t$;
        \item[(iv.)]$\Big | \frac{|\Gamma^{(0)}_k \cap \Gamma^{(0)}_{l}|}{n} - \alpha^2  \Big| < n^{-0.1}$;
        \item[(v.)] $\Big | \frac{|\Pi^{(s)}_k \cap \Pi^{(s)}_{l}|}{n} -\phi(  \frac{12}{K_{s-1}}  \langle {\beta}^{(s-1)}_k , {\beta}^{(s-1)}_l \rangle) \Big| < n^{-0.1}(\log n)^{10s}\prod_{i \leq s} K_i^{100}$ for  $1 \leq s \leq t$;
        \item[(vi.)] $\Big | \frac{|\Pi^{(0)}_k \cap \Pi^{(0)}_{l}|}{n} - \alpha^2  \Big| < n^{-0.1}$;
        \item[(vii.)] $\Big | \frac{| \pi(\Gamma^{(s)}_k) \cap \Pi^{(s)}_{l}|}{n} -\phi( \frac{ \epsilon }{2}  \frac{12}{K_{s-1}}   \langle \Hat{\beta}^{(s-1)}_k , \Hat{\beta}^{(s-1)}_l \rangle) \Big| < n^{-0.1}(\log n)^{10s} \prod_{i \leq s} K_i^{100}$ for $1 \leq s \leq t$;
        \item[(viii.)] $\Big | \frac{| \pi(\Gamma^{(0)}_k) \cap \Pi^{(0)}_{k}|}{n} - \phi(\epsilon)  \Big| < n^{-0.1}$ and  $\Big | \frac{| \pi(\Gamma^{(0)}_k) \cap \Pi^{(0)}_{l}|}{n} - \alpha^2  \Big| < n^{-0.1} $;
        \item[(ix.)] $\Big | \frac{|\Gamma^{(s)}_k \cap \Gamma^{(r)}_{l}|}{n} - \alpha^2 \Big| < n^{-0.1} (\log n)^{10s} \prod_{i \leq s} K_i^{100}$ for $0 \leq r < s \leq t$;
        \item[(x.)] $\Big | \frac{|\Pi^{(s)}_k \cap \Pi^{(r)}_{l}|}{n} - \alpha^2 \Big| < n^{-0.1}(\log n)^{10s} \prod_{i \leq s} K_i^{100}$ for $0 \leq r < s \leq t$;
        \item[(xi.)] $\Big | \frac{| \pi(\Gamma^{(s)}_k) \cap \Pi^{(r)}_{l}|}{n} - \alpha^2 \Big| < n^{-0.1}(\log n)^{10 \max(s, r)} \prod_{i \leq \max(s, r)} K_i^{100}$ for $0 \leq r \not = s \leq t$.
    \end{enumerate}
    Here $\alpha,K_t,\phi,\beta$ and $\Hat{\beta}^{(t)}$ are defined previously in Section~\ref{sec:algorithm-description}.
\end{Definition}
\begin{Remark}\label{remark-3.5}
   A comment is in order for this seemingly daunting definition of admissibility.  Intuitively we hope that our iteration is ``nearly'' independent between different steps: for instance, knowing $v \in \Gamma^{(t)}_k$ should not have much influence on the iteration in the $(t+1)$-th step. To confirm this intuition, we try to calculate the correlation between the random variables $\langle \eta^{(t)}_k , D^{(t)}_v \rangle$ and $\langle \eta^{(s)}_k , D^{(s)}_u \rangle$ and we even employ a simplification by regarding the random sets $\Gamma^{(t)}_k, \Pi^{(t)}_k$ as fixed and applying a Gaussian computation. Then we see that this ``informal correlation'' is $O(\frac{1}{n})$ when $u \not = v$. Now we focus on the informal correlation when $u = v$. By a simple calculation we get that the informal correlation $\mathbb{E}[ \langle \eta^{(t)}_k , D^{(t)}_v \rangle \langle \eta^{(s)}_l , D^{(s)}_v \rangle ] \approx \eta^{(t)}_k \mathrm{M}_{\Gamma}^{(t,s)} ( \eta^{(s)}_l  )^{*}$ and $\mathbb{E}[ \langle \eta^{(t)}_k , D^{(t)}_v \rangle \langle \eta^{(s)}_l , \mathsf{D}^{(s)}_{ \pi(v) } \rangle ] \approx \eta^{(t)}_k \mathrm{P}_{\Gamma,\Pi}^{(t,s)}( \eta^{(s)}_l  )^{*}$. By \eqref{equ_linear_space} and \eqref{equ_vector_orthogonal},  the informal  correlation between $\langle \eta^{(t)}_k , D^{(t)}_v \rangle$ and $\langle \eta^{(s)}_l , D^{(s)}_v \rangle$ would be zero. However, since the matrix $\mathrm{P}_{\Gamma,\Pi}$ is inaccessible by the algorithm, our algorithm cannot make a choice such that $\langle \eta^{(t)}_k , D^{(t)}_v \rangle$ is uncorrelated with $\langle \eta^{(s)}_l , \mathsf{D}^{(s)}_{\pi(v)} \rangle$. In order to address this, we assume that $\mathrm{P}^{(t,s)}_{\Gamma,\Pi}$ concentrates around some deterministic matrix (which is then accessible by the algorithm). Then, we can choose $\eta$'s such that the aforementioned correlation is close to 0. This deterministic matrix turns out to be the zero matrix for $t \not= s$ and to be $\Psi^{(t)}$ for $t=s$. Finally, we also need to justify that the informal correlation considered above is a good approximation of the real correlation; this is why for $t$-admissibility we make assumptions for all $s\leq t$ and then the Gaussian computation (with conditioning on linear statistics) can be rigorously justified (see Section~\ref{sec:conditional-distributions}). All of these considerations contribute to the complexity of the admissible conditions above.
\end{Remark}
As mentioned above, the main purpose of the admissible conditions is to ensure the concentration of a few matrices, as we now explain in details. Write 
\begin{equation}\label{eq-def-Delta-s}
\Delta_s = n^{-0.1}(\log n)^{10s} \prod_{i \leq s} K_i^{100}
\end{equation}
 for $0 \leq s \leq t$.
Under the $t$-admissible assumption, we see that for $s \leq t$
    \begin{align}
        &\Big| \frac{ |\Gamma^{(s)}_i \cap \Gamma^{(s)}_j| - \alpha |\Gamma^{(s)}_i| - \alpha |\Gamma^{(s)}_j| + \alpha^2 n }{(\alpha - \alpha^2) n}-\Phi^{(s)}(i,j) \Big| \nonumber\\ 
        &\leq  \frac{1}{\alpha-\alpha^2} \Big|  \frac{|\Gamma^{(s)}_i \cap \Gamma^{(s)}_j|}{n} - \phi ( \frac{12}{K_t} \langle {\beta}^{(t)}_i,{\beta}^{(t)}_j \rangle )  \Big| + \frac{\alpha}{\alpha-\alpha^2} \Big|  \frac{|\Gamma^{(s)}_i|}{n} - \alpha  \Big| + \frac{\alpha}{\alpha-\alpha^2} \Big|  \frac{|\Gamma^{(s)}_j|}{n} - \alpha  \Big| \nonumber\\
        &\overset{ (i.), (iii.) }{\leq} \frac{2}{\alpha} \Delta_s\,,
        \label{equ_concentrate_M_Gamma}
    \end{align}
    where the first inequality follows from \eqref{equ_def_iter_matrix} and the triangle inequality. 
    Similarly (by applying (ii.), (v.) instead of (i.), (iii.)), we get that
    \begin{equation}
       \Big| \frac{ |\Pi^{(s)}_i \cap \Pi^{(s)}_j| - \alpha |\Pi^{(s)}_i| - \alpha |\Pi^{(s)}_j| + \alpha^2 n }{(\alpha - \alpha^2) n}-\Phi^{(s)}(i,j) \Big|  \leq\frac{2}{\alpha} \Delta_s\,.
        \label{equ_concentration_M_Pi}
    \end{equation}
    In addition, by \eqref{equ_def_iter_matrix} and the triangle inequality we have that
    \begin{align}
        &\Big| \frac{ |\pi(\Gamma^{(s)}_i) \cap \Pi^{(s)}_j| - \alpha |\Gamma^{(s)}_i| - \alpha |\Pi^{(s)}_j| + \alpha^2 n }{(\alpha - \alpha^2) n}-\Psi^{(s)}(i,j) \Big| \nonumber \\ 
        &\leq \frac{1}{\alpha-\alpha^2} \Big|  \frac{|\pi(\Gamma^{(s)}_i) \cap \Pi^{(s)}_j|}{n} - \phi( \frac{\epsilon}{2}\frac{12}{K_t} \langle \Hat{\beta}^{(t)}_i, \Hat{\beta}^{(t)}_j \rangle)  \Big| + \frac{\alpha}{\alpha-\alpha^2} \Big|  \frac{|\Gamma^{(s)}_i|}{n} - \alpha  \Big| + \frac{\alpha}{\alpha-\alpha^2} \Big|  \frac{|\Pi^{(s)}_j|}{n} - \alpha  \Big| \nonumber\\
        & \overset{ (i.), (ii.), (vii.) }{\leq} \frac{2}{\alpha} \Delta_s\,.
        \label{equ_concentration_P_Gamma_Pi}
    \end{align}
Therefore, the matrices $\mathrm{M}^{(t,t)}_{\Gamma}, \mathrm{M}^{(t,t)}_{\Pi}, \mathrm{P}^{(t,t)}_{ \Gamma,\Pi}$ concentrate around $\Phi^{(t)}, \Phi^{(t)},\Psi^{(t)}$ respectively (the case of $t=0$ can be derived separately using (iv.), (vi.) and (viii.) in a similar manner). In addition,  recalling \eqref{equ_martix_M_P} and applying the triangle inequality again, we get that for $s<t$
    \begin{equation}
        | \mathrm{M}^{(t,s)}_{\Gamma}(i,j) | \leq  \frac{1}{\alpha - \alpha^2} \Big(\Big|  \frac{|\Gamma^{(t)}_i \cap \Gamma^{(s)}_j|}{n} - \alpha^2  \Big| + \alpha \Big|  \frac{|\Gamma^{(t)}_i|}{n} - \alpha  \Big| +  \alpha\Big|  \frac{|\Gamma^{(s)}_j|}{n} - \alpha  \Big| \Big) \overset{(i.),(ix.)}{\leq} \frac{2}{\alpha} \Delta_t\,.
        \label{equ_bound_M_Gamma_t_s}
    \end{equation}
    Similarly, by applying (ii.), (x.) instead of (i.), (ix.) we get that for $s<t$
    \begin{equation}
    | \mathrm{M}^{(t,s)}_{\Pi}(i,j) |   \leq \frac{2}{\alpha} \Delta_t\,.
    \label{equ_bound_M_Pi_t_s}
    \end{equation}
Furthermore, by \eqref{equ_martix_M_P} and the triangle inequality again, we get that for $s<t$
    \begin{equation}
        | \mathrm{P}^{(t,s)}_{\Gamma,\Pi}(i,j) | 
        \leq  \frac{1}{\alpha - \alpha^2}  \Big( \Big|  \frac{|\pi(\Gamma^{(t)}_i) \cap \Pi^{(s)}_j|}{n} - \alpha^2  \Big| + \alpha \Big|  \frac{|\Gamma^{(t)}_i|}{n} - \alpha  \Big| + \alpha \Big|  \frac{|\Pi^{(s)}_j|}{n} - \alpha  \Big|\Big) \leq\frac{2}{\alpha} \Delta_t\,,
        \label{equ_bound_P_Gamma_Pi_t_s}
    \end{equation}
where the last inequality follows from (i.), (ii.) and (xi.) (and the same bound holds for $\mathrm{P}^{(s,t)}_{\Gamma,\Pi}$ by the same argument). 
Therefore, we obtain that for $s<t$ the matrices $\mathrm{M}^{(t,s)}_{\Gamma}, \mathrm{M}^{(t,s)}_{\Pi}, \mathrm{P}^{(t,s)}_{\Gamma,\Pi}, \mathrm{P}^{(s,t)}_{\Gamma,\Pi}$ have entries upper-bounded by $ 2 \alpha^{-1} \Delta_t$.

We continue to explain our proof ideas. For $t \geq 0$ , define
\begin{equation*}
    \mathcal{E}_{t} = \{     ( \Gamma^{(s)}_k,\Pi^{(s)}_k  )_{1 \leq k \leq K_s, 0 \leq s \leq t}   \mbox{ is $t$-admissible}\}\,.
\end{equation*}
The key is to prove the following:
\begin{Proposition}{\label{prop_admissible}}
We have that $\mathbb P[\mathcal E_{t^*}] \geq 1 -o(1)$.
\end{Proposition}
It is fairly obvious at this point that the proof of Proposition~\ref{prop_admissible} would be via induction. To this end, we write
\begin{align*}
    \mathbb{P}[\mathcal{E}_{t^*}^{c}] =  \mathbb{P}[\mathcal{E}_0^{c}] + \sum_{1 \leq t \leq t^*} \mathbb{P}[\mathcal{E}_t^{c} \cap \mathcal{E}_{t-1}]\,.
\end{align*}
By a straightforward concentration result, we see that $\mathbb P[ \mathcal{E}_0] \geq 1 - O((\log n)^{-3})$ (see \eqref{eq-mathcal-E-0-bound} for more detailed explanation). 
Thus, it remains to upper-bound $\mathbb{P}[ \mathcal{E}^{c}_t \cap \mathcal{E}_{t-1}]$.

%Also note that $\mathcal{E}_0$ only depends on $\{ G_{u_j,v}, \mathsf{G}_{\pi(u_j) , %\mathsf{v}} : v \in V \backslash A , \mathsf{v} \in \mathsf{V} \backslash \mathsf{A} \} $, %and the iteration depends on $\{ \Hat{G}_{v,w}, \Hat{\mathsf{G}}_{\mathsf{v,w}} :  v,w \in %V \backslash A, \mathsf{v,w} \in \mathsf{V} \backslash \mathsf{A}  \}$, so $\mathcal{E}%_t$ for $t\geq 1$ is independent with $\mathcal{E}_0$. 

Recall \eqref{equ_degree} and \eqref{equ_def_iter_sets}. As hinted earlier, the main difficulty arises from the complicate dependency among the iterative steps. For instance, $\langle \eta^{(t)}_k,D^{(t)}_v \rangle$ does not only depend on $G$, but also depends on $\{  \langle \eta^{(s)}_k,D^{(s)}_v \rangle, \langle \eta^{(s)}_k,\mathsf{D}^{(s)}_{\mathsf{v}} \rangle\}$ for $s<t$. To be more precise, it is a linear combination of entries in $\hat G$ whereas the coefficients depend on $( \Gamma^{(t)}_k,\Pi^{(t)}_k)$ which in turn depends on $\{ \langle \eta^{(t-1)}_k,D^{(t-1)}_v \rangle , \langle \eta^{(t-1)}_k,\mathsf{D}^{(t-1)}_{\mathsf{v}} \rangle\}$. In order to address this, we can regard $\{ \Gamma^{(0)}_k, \Pi^{(0)}_k  \}$ as a fixed admissible realization: since the initialization only uses $\{ G_{u, v}, \mathsf{G}_{\mathsf u,\mathsf{v}}: u\in A, \mathsf u\in\mathsf A \}$,  conditioning on the initialization will not change the law of $\{ \Hat{G}_{v,w} , \Hat{\mathsf{G}}_{\mathsf{v,w}} : v,w \in V \backslash A, \mathsf{v,w} \in \mathsf{V} \backslash \mathsf{A} \}$. The much harder part is to deal with the correlation in the iteration. One way to address this, as carried out in \cite{BM10} for a different model (which in particular only involves one matrix), is to consider instead the conditional distribution of $(\Hat{G},\Hat{\mathsf{G}})$ under $\mathfrak{S}_{t-1}$, where $\mathfrak{S}_{t-1}= \sigma \{ \langle \eta^{(s)}_k,D^{(s)}_v \rangle, \langle \eta^{(s)}_k,\mathsf{D}^{(s)}_{\mathsf{v}} \rangle : 0 \leq s \leq t-1, v\in V\setminus A, \mathsf v\in \mathsf V \setminus \mathsf A \}$. Since $ \Gamma^{(s)}_k,\Pi^{(s)}_k$ are measurable with respect to  $\mathfrak{S}_{t-1}$ for $s \leq t$, the conditioning is equivalent to compute the conditional distribution of $\hat G$ under several linear constraints as given in \eqref{equ_degree} and \eqref{equ_def_iter_sets}. At this point, Lemma~\ref{lemma_Gaussian_condition} (below) will play a very useful role in removing correlations by taking projections.

However, there are various obstacles in making the above outline into a rigorous proof, and these will be explained and treated in the rest of the section.

\subsection{A few probability inequalities}

We collect in this subsection a few useful probability inequalities that will be used in later analysis repeatedly. 

The following Hanson-Wright inequality is useful in controlling the quadratic form of sub-Gaussian variables (see \cite{HW71, Wright73, Latala06, DKN10, BM13, FR13, RV14} and in particular see \cite{RV14} for an outstanding account on this topic and for an extremely nice proof inspired by \cite{Bourgain99}). For a random variable $\xi$, we define $\|\xi\|_{\psi_2} = \sup_{p\geq 1} p^{-1/2}\E |\xi^p|^{1/p}$.
\begin{Lemma}{\label{lemma_Hanson_Wright}}
    \textup{(Hanson-Wright Inequality)} Let $X=(X_1,\ldots,X_m)$ be a random vector with independent components which satisfy  $\mathbb{E}[X_i]=0$ and  $\| X_i \|_{\psi_2} \leq K$ for all $1\leq i\leq m$. If $A$ is an $m*m$ symmetric matrix, then for an absolute constant $c>0$ we have
    \begin{equation}
        \mathbb{P} [ \left|X A X^{*} - \mathbb{E}[X A X^{*} ] \right| > s  ] \leq 2 \exp \Big(  -c  \min \Big\{ \frac{s^2}{ K^4 \|A\|^2_{\mathrm{HS}}}, \frac{s}{ K^2 \|A\|_{\mathrm{op}}}  \Big\} \Big)\,.
        \label{equ_Hanson_Wright_tail}
    \end{equation}
\end{Lemma}
\begin{Corollary}{\label{corollary_Hanson_Wright_Gaussian}}
    For any Gaussian vector $X=(X_1,\ldots,X_m)$ and an $m*m$ symmetric matrix $A$, the following holds for an absolute constant $c>0$
    \begin{align}
        \mathbb{P} \left[ |X A X^{*} - \mathbb{E}[ X A X^*]| > s \right] \leq 2 \exp \Big( -c \min \Big\{ \frac{s^2}{\mathbb{E}[(X A X^{*})^2]} , \frac{s}{\sqrt{\mathbb{E}[(X A X^{*})^2]}} \Big\}  \Big)\,.
        \label{equ_corollary_Hanson_Wright_tail}
    \end{align}
\end{Corollary}
\begin{proof}
We may write $X =  Y T$, where $T$ is an $l*m$ matrix (for some $l\leq m$), $Y = (Y_1, \ldots, Y_l)$ and $Y_1, \ldots, Y_l$ are i.i.d.\ standard normal variables. Then $X A X^{*} = Y T A T^{*} Y^{*}$. Thus,
\begin{align*}
    \mathbb{P}[ |X A X^{*} - \mathbb{E}[ X A X^*]| > s ] &= \mathbb{P}[ |Y T A T^{*} Y^{*} - \mathbb{E}[ Y T A T^* Y^{*}]| > s ]  \\
    &\leq \exp \Big \{ -c \min \Big(  \frac{ s^2 }{ \| T A T^{*} \|^2_{\mathrm{HS}} }, \frac{s}{ \| T A T^{*} \|_{\mathrm{HS}} }   \Big) \Big \}\,,
\end{align*}
where the inequality follows from  Lemma~\ref{lemma_Hanson_Wright} and the fact that $\| TAT^* \|_{\mathrm{op}} \leq \| TAT^* \|_{\mathrm{HS}}$. Since $\| T A T^{*} \|^2_{\mathrm{HS}} \leq \mathbb{E} [ ( Y T A T^{*} Y^{*} )^2   ] = \mathbb{E} [ (X A X^{*})^2 ]$, the corollary follows. 
\end{proof}

\begin{Corollary}{\label{corollary_Hanson_Wright_independent_subGaussian}}
    Let $X_1, \ldots, X_\ell, Y_1, \ldots, Y_m $ be independent variables with mean-zero such that $\| X_i \|_{\psi_2}, \| Y_j \|_{\psi_2} \leq K$ for all $1\leq i\leq \ell, 1\leq j\leq m$. Let $X = (X_1, \ldots, X_\ell)$ and $Y = (Y_1, \ldots, Y_m)$.  For any $\ell* m$ matrix $A$, we have that for an absolute constant $c>0$
    \begin{align*}
        \mathbb{P} [ | X A Y^{*} | > s ] \leq 2 \exp \Big \{ - c  \min \Big( \frac{s^2}{K^4 \| A \|^2_{\mathrm{HS}}}, \frac{ s }{ K^2 \| A \|_{\mathrm{HS}} }   \Big) \Big \}\,.
    \end{align*}
\end{Corollary}
\begin{proof}
Write $Z = (X,Y)$. Then $X A Y^{*} = Z B Z^{*}$, where $B= \frac{1}{2} \begin{pmatrix} \mathrm{0}_{\ell \times \ell} & A \\ A^{*} & \mathrm{0}_{m \times m} \end{pmatrix}$. Since $\| B \|_{\mathrm{HS}} \leq \| A \|_{\mathrm{HS}}$, the corollary follows from Lemma~\ref{lemma_Hanson_Wright}.
\end{proof}

Using the decoupling technique employed in \cite{RV14} (which in turn was inspired by \cite{Bourgain99}), we derive the following version of sub-Gaussian concentration inequality.
\begin{Lemma}{\label{lemma_decoupling_subGaussian}}
    Let $X_1,\ldots,X_m,Y_1,\ldots,Y_m$ be mean-zero variables with $\| X_i \|_{\psi_2}, \| Y_i \|_{\psi_2} \leq K$ for all $1\leq i\leq m$. In addition, assume for all $i$ that $(X_i, Y_i)$ is independent with $(X_{\setminus i}, Y_{\setminus i})$ where $X_{\setminus i}$ is obtained from $X$ by dropping its $i$-th component (and similarly for $Y_{\setminus i}$). Let A be an $m \times m$ matrix with diagonal entries being 0. Then for every $s > 0$ 
    \begin{align*}
        \mathbb{P} [ | X A Y^{*} | > s ] \leq 2 \exp \Big \{  -c \min \Big(  \frac{s^2}{K^4 \| A \|^2_{\mathrm{HS}}}, \frac{ s }{ K^2 \| A \|_{\mathrm{HS}} } \Big)  \Big \}\,.
    \end{align*}
\end{Lemma}
\begin{proof}
Our proof is based on calculating the exponential moments of $S = \sum_{i \not = j} a_{ij} X_i Y_j$. Following the decoupling technique as in \cite{RV14}, we sample i.i.d.\ Bernoulli random variables $\delta_1,\delta_2,\ldots,\delta_m$ such that $\mathbb{E}[ \delta_i ] = \frac{1}{2}$. Write $\delta = (\delta_1, \ldots, \delta_m)$ and define $S_{\delta} = \sum_{i,j} \delta_i (1-\delta_j) a_{ij} X_i Y_j$. Since $S = 4 \mathbb{E}_{\delta} [ S_{\delta} ]$ (here $\mathbb E_\delta$ is the operation of taking expectation over $\delta$), by the conditional Jensen's inequality we have that $\mathbb{E}[ \exp \{ \lambda S \} ] \leq \mathbb{E}_{X, Y, \delta} [ \exp \{ 4 \lambda S_{\delta} \} ]$ for any $\lambda>0$. Now define $\Lambda_{\delta} = \{ 1\leq i\leq m : \delta_{i} =1  \}$. Conditioned on $\delta$ and $\{ X_i : i \in \Lambda_{\delta} \} $, we see that
\begin{align*}
    S_{\delta} = \sum_{j \in \Lambda_{\delta}^{c}} Y_j ( \sum_{i \in \Lambda_{\delta}} a_{ij} X_i )
\end{align*}
is a linear combination of sub-Gaussian variable $\{ Y_j : j \in \Lambda_{\delta}^{c} \}$ with fixed coefficients given by $\sum_{i \in \Lambda_{\delta}} a_{ij} X_i$. Under this conditioning we can use our assumption on $(X_i, Y_i)$ independent of $(X_{\setminus i}, Y_{\setminus i})$ to prove by induction that $\{ Y_j : j \in \Lambda_{\delta}^{c} \}$ is independent with $\{ X_i : i \in \Lambda_{\delta} \}$. By properties of sub-Gaussian random variables we know $\| S_{\delta} \|_{\psi_2} \leq C \sigma_{\delta}$, where $C>0$ is an absolute constant and $\sigma_{\delta}^2 = \sum_{ j \in \Lambda_{\delta}^{c} }(  \sum_{i \in \Lambda_{\delta}} a_{ij} X_i)^2$. Thus, for $C' >0$ depending on $C$
\begin{align*}
    \mathbb{E}_{\{ Y_j: j \in \Lambda_{\delta}^{c} \}}[ \exp \{ 4 \lambda S_{\delta}  \} ] \leq \exp \{ C^{\prime} \lambda^2 \sigma_{\delta}^2 \}\,.
\end{align*}
Taking expectation with respect to $\{ Y_i : i \in \Lambda_{\delta} \}$ and $\{X_1, \ldots, X_m\}$, we get that
\begin{align*}
    \mathbb{E}_{X,Y}[ \exp \{ 4 \lambda S_{\delta}  \} ] \leq  \mathbb{E}_{X} [ \exp \{ C^{\prime} \lambda^2 \sigma_{\delta}^2  \} ]
\end{align*}
for any given realization $\delta$. Since by our assumption $X_1, \ldots, X_m$ are independent, we can now just follow Step 3-5 in the proof of Theorem 1.1 in \cite{RV14} and obtain the desired bound.
\end{proof}

We will also use the following version of Bernstein's inequality as in \cite[Theorem 1.4]{DP09}.
\begin{Lemma}{\label{lemma_Berstein_ine}}
    Let $X=\sum_{i=1}^{m} X_i$, where $X_i$'s are independent random variables such that $|X_i| \leq K$ almost surely. Then, for $s>0$, we have
    \begin{align*}
        \mathbb{P}[ X> \mathbb{E}[X]+s  ] \leq \exp \Big \{  - \frac{s^2}{2(\sigma^2+Ks/3)} \Big\}\,,
    \end{align*}
    where $\sigma^2 = \sum_{i=1}^{m} \mathrm{Var}(X_i)$ is the variance of $X$. It follows (by applying the above for $X$ and $-X$) that for $\gamma>0$, 
    \begin{align*}
        \mathbb{P} [ |X -\mathbb{E}[X]| \geq  \sqrt{2 \sigma^2 \gamma} + 2K\gamma/3  ] \leq 2 e^{-\gamma}\,.
    \end{align*}
\end{Lemma}
The following lemma is well-known.
\begin{Lemma}{\label{lemma_Gaussian_convar}}
    Let $Z_1,Z_2,\ldots,Z_t$ be multivariate normal variables with mean zero. Assume that the covariance matrix of $Z_1,Z_2,\ldots,Z_{t-1}$ (denoted by $C$) is invertible. Then
    \begin{align*}
    \E(Z_t\mid Z_1, \ldots, Z_{t-1}) = u C^{-1} Z^* \mbox{ and }
        \mathrm{Var}(Z_t | Z_1,\ldots,Z_{t-1}) = \mathbb{E}[Z_t^2] - u C^{-1} u^*\,,
    \end{align*}
    where $u = (\mathbb{E}[Z_t Z_1], \ldots, \mathbb{E}[Z_t Z_{t-1}])$ and $Z = (Z_1, \ldots, Z_{t-1})$.
\end{Lemma}

\subsection{Conditional distributions}\label{sec:conditional-distributions}
Recall that $\mathfrak{S}_{t}= \sigma \{ \langle \eta^{(s)}_k,D^{(s)}_v \rangle, \langle \eta^{(s)}_k, \mathsf{D}^{(s)}_{\mathsf{v}} \rangle  : s \leq t, v \in V \backslash A , \mathsf{v} \in \mathsf{V} \backslash \mathsf{A} \}$. In order to calculate the law of $\langle \eta^{(t+1)}_k,D^{(t+1)}_v \rangle$ under the conditioning of $\mathfrak{S}_{t}$, we will study the conditional distribution $(\Hat{G},\Hat{\mathsf{G}}) |_{\mathfrak{S}_{t}}$, where as noted earlier the conditioning is given as a set of linear constraints. Thus, the following well-known lemma will then be useful.
\begin{Lemma}{\label{lemma_Gaussian_condition}}
    Let $Z \in \mathbb{R}^M$ be a random vector with i.i.d.\ standard Gaussian components, and let $ D \in \mathbb{R}^{m*M}$ be a linear operator with full row rank. Then for any constant vector $b \in \mathbb{R}^{m}$, the conditional distribution of $Z$ given $DZ=b$ satisfies:
    \begin{align}\label{eq-gaussian-condition-equivalent-form}
        Z|\{DZ=b\} \overset{d}{=} D^{*} (DD^{*})^{-1} b + \mathcal{P}_{D^{\perp}}(\Tilde{Z})\,,
    \end{align}
    where $D^{\perp}$ is the orthogonal space $\{x \in \mathbb{R}^n: Dx= 0 \}$, $\mathcal{P}_{D^{\perp}}$ is the orthogonal projection onto the subspace $D^{\perp}$, and $\Tilde{Z}$ is a copy of $Z$.
\end{Lemma}
We next derive a consequence of Lemma~\ref{lemma_Gaussian_condition}, and we continue to use notations therein. For two fixed vectors $\sigma, \tau \in \mathbb R^M$, conditioned on $DZ = b$ we can compute the covariance between $\langle \sigma, Z\rangle$ and $\langle \tau, Z\rangle$ as follows:
\begin{align}
    &\mathrm{Cov}(\langle \sigma , Z|\{DZ=b\} \rangle,\langle \tau,Z |\{DZ=b\} \rangle) \nonumber \\
    =& \mathrm{Cov}( \langle \sigma, D(D^{*}D)^{-1}b + \mathcal{P}_{D^{\perp}}(\Tilde{Z}) \rangle , \langle \tau, D(D^{*}D)^{-1}b + \mathcal{P}_{D^{\perp}}(\Tilde{Z}) \rangle  ) \nonumber \\
    =& \mathrm{Cov}(\langle \sigma ,  \mathcal{P}_{D^{\perp}}(\Tilde{Z}) \rangle , \langle \tau,  \mathcal{P}_{D^{\perp}}(\Tilde{Z}) \rangle   ) = \mathbb{E}[ \sigma \mathcal{P}_{D^{\perp}} \Tilde{Z}^{*} \Tilde{Z} \mathcal{P}_{D^{\perp}} \tau^{*}  ] \nonumber \\
    =& \sigma \mathcal{P}_{D^{\perp}} \mathbb{E}[ \Tilde{Z}^{*} \Tilde{Z}] \mathcal{P}_{D^{\perp}} \tau^{*} =\sigma \mathcal{P}_{D^{\perp}}  \tau^{*} = \langle \mathcal{P}_{D^{\perp}}(\sigma),  \mathcal{P}_{D^{\perp}}(\tau) \rangle   \nonumber \\
    =& \langle \sigma - \mathcal{P}_D(\sigma) , \tau - \mathcal{P}_D(\tau) \rangle = \langle \sigma, \tau \rangle - \langle \mathcal{P}_D(\sigma) , \mathcal{P}_D(\tau) \rangle\,. \label{eq-conditional-covariance-projection}
\end{align}

In later analysis we will apply a slightly modified version of Lemma~\ref{lemma_Gaussian_condition}. Consider the sets of Gaussian variables $\mathcal{G}= \{ W_{p}= \sum_{k=1}^{M} D(p,k) Z_k : 1 \leq p \leq m \}$. Note that $ D D^*$ is the covariance matrix of $\{ W_{p} : 1 \leq p \leq m \}$. For each $1 \leq j \leq M$, denote the projection of $Z_j$ onto $\mathrm{span}(\mathcal{G})$ by $\sum_{p=1}^{m} a_{j, p} W_{p} $. Then we have
\begin{align}\label{eq-equivalence-conditional-distribution}
    \{Z_j |\{DZ=b\}\} \overset{d}{=} \{\sum_{p=1}^{m} a_{j, p} b_p + \Tilde{Z}_j - \sum_{p=1}^{m}  a_{j, p}   \sum_{k=1}^{M} D(p,k) \Tilde{Z}_k\}\,,
\end{align}
where $\Tilde{Z}$ is an independent copy of $Z$. We now check that the above expression is equivalent to Lemma~\ref{lemma_Gaussian_condition}. It is obvious that $\mathcal{P}_{D^{\perp}}(Z)=Z - \mathcal{P}_{D}(Z) = Z - \mathbb{E}[ Z | \mathcal{G} ]$, and hence the $j$-th entry of $\mathcal{P}_{D^{\perp}}(\Tilde{Z}) $ is $\Tilde{Z}_j - \sum_{p=1}^{m} a_{j, p} \sum_{k=1}^{n} D(p,k) \Tilde{Z}_k$. It remains to show that the $j$-th entry of $D^{*} (DD^{*})^{-1} b$ is $\sum_{p=1}^{m} a_{j, p} b_p$. To this end, we apply Lemma~\ref{lemma_Gaussian_convar} and get
\begin{align*}
    \mathbb{E}[ Z_j | \mathcal{G} ] =
    (\mathbb{E}[Z_j W_1] \ldots \mathbb{E}[Z_j W_m]) (DD^{*})^{-1} W^{*} = e_j D^{*} (DD^{*})^{-1} W^{*}\,,
\end{align*}
 where $e_j$ is the vector in $\mathbb{R}^{d}$ with the $j$-th entry being 1 and all other entries being 0, and $W= (W_p)_{1 \leq p \leq m}$.
This completes the verification for the equivalence.

Motivated by \eqref{eq-gaussian-condition-equivalent-form}, we define an operation $\Hat{\mathbb{E}}$ as follows. For any function $g$ (of the form $g ( \Gamma, \Pi, \Hat{G}, \Hat{\mathsf{G}} )$) and  any realization $\xi$  for  $\{ \Gamma, \Pi \}$, define  $f( \xi ) = \mathbb{E}_{ \Hat{G}, \Hat{\mathsf{G}} } [ g ( \xi, \Hat{G}, \Hat{\mathsf{G}} ) ]$. Then the operator $\Hat {\mathbb E}$ is defined such that $\Hat{\mathbb{E}}[ g ( \Gamma, \Pi, \Hat{G}, \Hat{\mathsf{G}} ) ] = f( \Gamma, \Pi )$. 
For example, we have $\Hat{\mathbb{E}} [ ( \mathbf{I}_{v \in \Gamma^{(t)}_k}(|\Gamma^{(t)}_k|)^{-1/2} \Hat{G}_{v,u}  )^2 ] = \mathbf{I}_{v \in \Gamma^{(t)}_k}(|\Gamma^{(t)}_k|)^{-1}$. Later, when calculating $\Hat{\mathbb{E}}$ involving $(D^{(t)}_v, \mathsf{D}^{(t)}_{\mathsf{v}})$, we will regard $(D^{(t)}_v, \mathsf{D}^{(t)}_{\mathsf{v}})$ as a vector-valued function of $(\Gamma,\Pi,\Hat{G},\Hat{\mathsf{G}})$, where for instance $D^{(t)}_v(k) = g_{v,k}( \Gamma,\Pi,\Hat{G},\Hat{\mathsf{G}} ) = \frac{1}{\sqrt{(\alpha - \alpha^2)n} } \sum_{u \in V \backslash A} ( \mathbf{I}_{u \in \Gamma^{(t)}_k} - \alpha ) \Hat{G}_{v,u} $.

We use $\Hat{\mathbb{E}}$ to calculate the projection of $\langle \eta^{(t+1)}_k , D^{(t+1)}_v \rangle$ to $\mathfrak{S}_{t}$ as follows. First we  choose a basis for $\mathfrak{S}_{t}$. One natural choice is $\{ \langle \eta^{(s)}_k,D^{(s)}_v \rangle , \langle \eta^{(s)}_k, \mathsf{D}^{(s)}_{\mathsf{v}} \rangle :s \leq t, v \in V \backslash A ,\mathsf{v} \in \mathsf{V} \backslash \mathsf{A} , 1 \leq k \leq K_s  \}$. For convenience of definiteness, we list the basis in the following order: first we list all $\langle \eta^{(s)}_k, D^{(s)}_v \rangle$ indexed by $(s,k,v)$ in the dictionary order and then we list all $\langle \eta^{(s)}_k,\mathsf{D}^{(s)}_{\mathsf{v}} \rangle$ indexed by $(s,k,\mathsf{v})$ in the dictionary order. Under this ordering and the assumption that $\mathcal{E}_{t}$ happens, we have for all $s \leq t$
\begin{align}
    & \Hat{\mathbb{E}}[\langle \eta^{(t)}_k,D^{(t)}_v \rangle \langle \eta^{(s)}_m,D^{(s)}_u \rangle]  = \Hat{\mathbb{E}}[ \langle \eta^{(t)}_k, \mathsf{D}^{(t)}_{\mathsf{v}} \rangle \langle \eta^{(s)}_m,\mathsf{D}^{(s)}_{\mathsf{u}} \rangle ]=0  \mbox{ for } u\neq v, \mathsf u\neq \mathsf v\,, \label{equ_degree_correlation_1}\\
    & \Hat{\mathbb{E}}[\langle \eta^{(t)}_k, D^{(t)}_v \rangle \langle \eta^{(s)}_m, D^{(s)}_v \rangle]  = \eta^{(t)}_k \mathrm{M}_{\Gamma}^{(t,s)} \left( \eta^{(s)}_m \right)^{*} \overset{ \eqref{equ_linear_space}, \eqref{equ_vector_orthogonal} }{=} 0 \mbox{ for } (t,k) \neq (s,m) \,,  \label{equ_degree_correlation_2} \\ 
    & \Hat{\mathbb{E}}[\langle \eta^{(t)}_k, \mathsf{D}^{(t)}_{\mathsf{v}} \rangle \langle \eta^{(s)}_m, \mathsf{D}^{(s)}_{\mathsf{v}} \rangle] = \eta^{(t)}_k \mathrm{M}_{\Pi}^{(t,s)} \left( \eta^{(s)}_m \right)^{*} \overset{ \eqref{equ_linear_space}, \eqref{equ_vector_orthogonal} }{=} 0 \mbox{ for } (t,k) \neq (s,m)  \,, \label{equ_degree_correlation_3} 
\end{align}
where the first display follows from the fact that all $\Hat{G}_{u, w}$ over $u\neq w$ are independent Gaussian variables (and the similar version for $\hat{\mathsf G}$). Also, we have
\begin{align}
    & \Big| \Hat{\mathbb{E}}[ \langle \eta^{(t)}_k, D^{(t)}_v \rangle^2 ] - 1 \Big|  = \Big| \eta^{(t)}_k \mathrm{M}_{\Gamma}^{(t,t)} \Big( \eta^{(t)}_k \Big)^{*} -1 \Big|  \nonumber \\
    =& \Big| \eta^{(t)}_k \Phi^{(t)} \Big( \eta^{(t)}_k \Big)^{*} + \eta^{(t)}_k \Big( \mathrm{M}_{\Gamma}^{(t,t)}-\Phi^{(t)} \Big) \Big( \eta^{(t)}_k \Big)^{*} -1 \Big| \overset{ \eqref{equ_vector_unit}, \eqref{equ_concentrate_M_Gamma} }{\leq} K_t \Delta_t  \,,  \label{equ_degree_variance_1} \\ 
    &\Big| \Hat{\mathbb{E}}[\langle \eta^{(t)}_k, \mathsf{D}^{(t)}_{\mathsf{v}} \rangle^2 ] - 1 \Big| = \Big| \eta^{(t)}_k \mathrm{M}_{\Pi}^{(t,t)} \Big( \eta^{(t)}_k \Big)^{*} -1 \Big| \nonumber\\
    =& \Big| \eta^{(t)}_k \Phi^{(t)} \Big( \eta^{(t)}_k \Big)^{*} +  \eta^{(t)}_k  \Big( \mathrm{M}_{\Pi}^{(t,t)} - \Phi^{(t)} \Big) \Big( \eta^{(t)}_k \Big)^{*} - 1 \Big| \overset{ \eqref{equ_vector_unit}, \eqref{equ_concentration_M_Pi} }{\leq} K_t \Delta_t  \,.\label{equ_degree_variance_2}
\end{align}
Thus, we may write their $\Hat{\mathbb{E}}$-correlation matrix as $\begin{pmatrix} \mathbf{P}_t & \mathbf{R}_t \\ \mathbf{R}_t^{*} & \mathbf{L}_t  \end{pmatrix}$, where $\mathbf{P}_t, \mathbf{L}_t$ are diagonal matrices with diagonal entries in $(1-K_t \Delta_t,1+K_t \Delta_t)$, and $\mathbf{R}_t$ is the matrix with row indexed by $(s,k,u)$ for $0 \leq s \leq t, 1 \leq k \leq K_s, u \in V \backslash A$ and column indexed by $(s,k, \mathsf{u})$ for  $0 \leq s \leq t, 1 \leq k \leq K_s, u \in \mathsf{V} \backslash \mathsf{A}$, and entries $\mathbf{R}_t((s,k,u);(r,l,\mathsf{w}))$ given by $\Hat{\mathbb{E}}[\langle \eta^{(s)}_k,D^{(s)}_u \rangle  \langle \eta^{(r)}_l,\mathsf{D}^{(r)}_{\mathsf{w}} \rangle]$. Thus, by Lemma~\ref{lemma_Gaussian_convar} we have
\begin{align}
    \mathbb{E}[ \langle \eta^{(t+1)}_k,D^{(t+1)}_v \rangle | \mathfrak{S}_{t}] = 
    \begin{pmatrix}
        H_{t} & \mathsf{H}_{t}  
    \end{pmatrix}
    \begin{pmatrix}
        \mathbf{P}_t       & \mathbf{R}_t \\
        \mathbf{R}_t^{*} & \mathbf{L}_t
    \end{pmatrix}^{-1}
    \begin{pmatrix}
        Y_{t} \\
        \mathsf{Y}_{t}
    \end{pmatrix}\,.
    \label{equ_formulate_projection}
\end{align}
Here $H_t,\mathsf{H}_t$ and $Y_t,\mathsf{Y}_t$ are all $\sum_{0 \leq s \leq t} K_s(n - \kappa)$ dimensional vectors; $H_t$ and $Y_t$ are indexed by triple $(s,l,u): 0 \leq s \leq t, 1 \leq k \leq K_s, u \in V \backslash A$ in the dictionary order; $\mathsf{H}_t$ and $\mathsf{Y}_t$ are indexed by triple $(s,l,\mathsf{u}): 0 \leq s \leq t, 1 \leq k \leq K_s, \mathsf{u} \in \mathsf{V} \backslash \mathsf{A}$ in the dictionary order. In addition, their entries are given by
\begin{equation}
\begin{aligned}
    &Y_t (s,l,u) = \langle \eta^{(s)}_l,D^{(s)}_u \rangle, \quad  H_t ( s,l,u ) = \Hat{\mathbb{E}}[\langle \eta^{(t+1)}_k, D^{(t+1)}_v \rangle  \langle \eta^{(s)}_l,{D}^{(s)}_{{u}} \rangle]\,; \\
    &
    \mathsf{Y}_t (s,l,\mathsf{u}) = \langle \eta^{(s)}_l, \mathsf{D}^{(s)}_{\mathsf{u}} \rangle, \quad  \mathsf{H}_t (s,l,\mathsf{u}) = \Hat{\mathbb{E}}[\langle \eta^{(t+1)}_k, D^{(t+1)}_v \rangle  \langle \eta^{(s)}_l, \mathsf{D}^{(s)}_{\mathsf{u}} \rangle] \,.
\end{aligned}
\end{equation}
\begin{Remark}{\label{remark_projection_form}}
    In conclusion, we have shown that  (recall \eqref{eq-equivalence-conditional-distribution})%conditioned on $\mathfrak{S}_{t}$, the law of $\langle \eta^{(t+1)}_k, D^{(t+1)}_v \rangle$ equals to 
    \begin{align}
   \Big\{ \langle \eta^{(t+1)}_k,D^{(t+1)}_v \rangle | \mathfrak{S}_{t} \Big\}\overset{d}{=} 
    & \Big\{ \begin{pmatrix}
        H_{t} & \mathsf{H}_{t}  
    \end{pmatrix}
    \begin{pmatrix}
        \mathbf{P}_t       & \mathbf{R}_t \\
        \mathbf{R}_t^{*} & \mathbf{L}_t
    \end{pmatrix}^{-1}
    \begin{pmatrix}
        Y_{t} \\
        \mathsf{Y}_{t}  
    \end{pmatrix} \Big| \mathfrak{S}_{t}
    \label{equ_projection_part}  \\
    + & \langle \eta^{(t+1)}_k,   \Tilde{D}^{(t+1)}_v \rangle -
    \begin{pmatrix}
        H_{t} & \mathsf{H}_{t}  
    \end{pmatrix}
    \begin{pmatrix}
        \mathbf{P}_t       & \mathbf{R}_t \\
        \mathbf{R}_t^{*} & \mathbf{L}_t
    \end{pmatrix}^{-1}
    \begin{pmatrix}
        \Tilde{Y}_{t} \\
        \Tilde{\mathsf{Y}}_{t}  
    \end{pmatrix} \Big| \mathfrak{S}_{t} \Big\}\,.
    \label{equ_Gaussian_part}
\end{align}
In the above $\Tilde{Y}_t (s,l,u) = \langle \eta^{(s)}_l, \Tilde{D}^{(s)}_u \rangle , \Tilde{\mathsf{Y}}_t (s,l,\mathsf{u}) = \langle \eta^{(s)}_l, \Tilde{\mathsf{D}}^{(s)}_{\mathsf{u}} \rangle$, where
$$ \Tilde{D}^{(t)}_v(k) = \frac{1}{\sqrt{(\alpha - \alpha^2)n}} \sum_{u \in V \backslash A }( \mathbf{I}_{u \in \Gamma^{(t)}_k} - \alpha) \Hat{G}_{v,u}^{(\textup{new})}$$ is a linear combination of Gaussian variables $\{\Hat{G}_{v,u}^{(\textup{new})}\}$, with coefficients fixed under the conditioning of $\mathfrak{S}_t$, and $\{\Hat{G}_{v,u}^{(\textup{new})}\}$ is an independent copy of $\{\Hat{G}_{v,u}\}$ (and similarly for $\Tilde{\mathsf D}^{(t)}_{\mathsf v}(k)$). Thus, conditioned on any realizations of $\mathfrak{S}_t$, the term  \eqref{equ_projection_part} is determined and the term  \eqref{equ_Gaussian_part} is a linear combination of $\{\Hat{G}_{u,v}^{(\textup{new})}, \Hat{\mathsf{G}}_{\mathsf{u},\mathsf{v}}^{(\textup{new})}\}$ with determined coefficients. For notation convenience, we denote \eqref{equ_projection_part} as $\textup{PROJ}(\langle \eta^{(t+1)}_k, D^{(t+1)}_v \rangle; \mathfrak{S}_t)$, and denote (\ref{equ_Gaussian_part}) as $\langle \eta^{(t+1)}_k, \Tilde{D}^{(t+1)}_v \rangle - \textup{GAUS}(\langle \eta^{(t+1)}_k, D^{(t+1)}_v \rangle; {\mathfrak{S}_t})$.
\end{Remark}
We now explain our intuition for bounding \eqref{equ_projection_part}. On $\mathcal{E}_{t}$, we have for $u \neq v$
\begin{equation}
\begin{aligned}
    \Hat{\mathbb{E}}[\langle \eta^{(t)}_k, D^{(t)}_v \rangle \langle \eta^{(s)}_l, \mathsf{D}^{(s)}_{\pi(u)} \rangle] & = \frac{\epsilon}{2} \sum_{i=1}^{K_{t}} \sum_{j=1}^{K_s} \eta^{(t)}_k(i) \eta^{(s)}_l(j) \frac{(\mathbf{I}_{\pi(v) \in \Pi^{(s)}_j} - \alpha) (\mathbf{I}_{u \in \Gamma^{(t)}_i}-\alpha)}{(\alpha -\alpha^2) n} \\
    & \leq \sum_{i=1}^{K_{t}} \sum_{j=1}^{K_s} \frac{|\eta^{(t)}_k(i) \eta^{(s)}_l(j)|}{\alpha n} \leq \frac{ 4 \sqrt{K_t K_s} } {\alpha n} \,;  
    \label{equ_bound_correlation_dif_vertices}
\end{aligned}
\end{equation}
and also for $s<t$
\begin{equation}
\begin{aligned}
    \Hat{\mathbb{E}}[\langle \eta^{(t)}_k, D^{(t)}_v \rangle \langle \eta^{(s)}_l, \mathsf{D}^{(s)}_{\pi(v)} \rangle] & = \frac{\epsilon}{2} \eta^{(t)}_k \mathbf{P}_{\Gamma,\Pi}^{(t,s)} \left(\eta^{(s)}_l \right)^{*}   \\
    & \overset{\eqref{equ_bound_P_Gamma_Pi_t_s}}{\leq} \sum_{i=1}^{K_{t}} \sum_{j=1}^{K_s} \alpha^{-1} \Delta_{t} |\eta^{(t)}_k(i) \eta^{(s)}_l(j)| \leq \frac{ 4 \sqrt{K_t K_s}\Delta_{t} }{ \alpha } \,;
    \label{equ_bound_correlation_same_vertices}
\end{aligned}
\end{equation}
and in addition for $k \neq l$
\begin{equation}
\begin{aligned}
    \Hat{\mathbb{E}}[\langle \eta^{(t)}_k, D^{(t)}_v \rangle \langle \eta^{(t)}_l, \mathsf{D}^{(t)}_{\pi(v)} \rangle] & = \frac{\epsilon}{2} \eta^{(t)}_k \mathbf{P}_{\Gamma,\Pi}^{(t,t)} \left(\eta^{(t)}_l \right)^{*} \overset{\eqref{equ_vector_orthogonal}}{=} \frac{\epsilon}{2} \eta^{(t)}_k ( \mathbf{P}_{\Gamma,\Pi}^{(t,t)}- \Psi^{(t)} ) \left(\eta^{(t)}_l \right)^{*}  \\
    & \overset{\eqref{equ_concentration_P_Gamma_Pi}}{\leq} \sum_{i=1}^{K_{t}} \sum_{j=1}^{K_s} \alpha^{-1} \Delta_{t} |\eta^{(t)}_k(i) \eta^{(t)}_l(j)| \leq \frac{ 4 K_t \Delta_{t} }{ \alpha } \,,
    \label{equ_bound_correlation_same_vert_same_time}
\end{aligned}
\end{equation}
where we have used the fact that $\| \eta^{(t)}_k \|, \| \eta^{(s)}_l \| \leq 2$ and the Cauchy-Schwartz inequality. Replacing $t$ with $t+1$ in preceding discussions \eqref{equ_degree_correlation_1} \eqref{equ_degree_correlation_2} \eqref{equ_degree_correlation_3} \eqref{equ_bound_correlation_dif_vertices} and \eqref{equ_bound_correlation_same_vertices}, we then get that
\begin{align}
    &\|  \begin{pmatrix} H_{t} & \mathsf{H}_{t}  \end{pmatrix}   \|_1 \leq 2K_{t} n \cdot \frac{  4 \sqrt{K_{t+1} K_t}} {\alpha n} + 2K_t \cdot \frac{  4 \sqrt{K_{t+1} K_t} \Delta_{t+1} }  {\alpha} \leq \frac{ 10 \sqrt{K_{t+1}} K_t^2 } {\alpha}  \label{equ_bound_1_norm_H_t}\,, \\
    &\|  \begin{pmatrix} H_{t} & \mathsf{H}_{t}  \end{pmatrix}   \|_2^2 \leq 2K_{t} n \cdot (\frac{  4 \sqrt{K_{t+1} K_t}} {\alpha n})^2 + 2K_t \cdot (\frac{ 4 \sqrt{K_{t+1} K_t} \Delta_{t+1} }  {\alpha})^2 \leq \frac{ 100K_{t+1}K_t^2 \Delta_{t+1}^2 } {\alpha^2} \label{equ_bound_2_norm_H_t}\,.
\end{align}
However, the bound on \eqref{equ_projection_part} we get from directly applying \eqref{equ_bound_1_norm_H_t} or \eqref{equ_bound_2_norm_H_t} is not sharp and later we need a better control which would have to take advantage of cancellations between the entries of $H_t$ and $\mathsf{H}_t$. The main trouble for implementing this is that it seems difficult to write down the exact formula for the entries of $\begin{pmatrix} \mathbf{P}_t       & \mathbf{R}_t \\ \mathbf{R}_t^{*} & \mathbf{L}_t \end{pmatrix}^{-1}$ since they seem to involve extremely complicate calculations. Instead we try to control the operator norm and  the 1-norm of this matrix, as incorporated in next few lemmas. This enables us to control the cancellation; see \eqref{equ_bound_lambda_u}, \eqref{equ_bound_lambda_u_w} and \eqref{equ_bound_square_lambda_u_w} for a detailed explanation of the role of the following lemmas in later proofs. Write $\mathbf Q_t = \begin{pmatrix} \mathbf{P}_t & \mathbf{R}_t \\ \mathbf{R}_t^{*} & \mathbf{L}_t \end{pmatrix}^{-1}$.
\begin{Lemma}{\label{lemma_op_norm_Q}}
    Suppose that $\epsilon \in (0, 0.9)$ (see Remark~\ref{rem-epsilon-assumption} below). On the event $\mathcal{E}_{t}$, we have $\|\mathbf Q_t \|_{\mathrm{op}} \leq 100$.
\end{Lemma}
\begin{proof}
For any vectors $x,\mathsf{y}$ with proper dimensions such that $x \mathbf{R}_t \mathsf{y}^{*}$ is well-defined (we may index $x$ with $(s,k,v)$ and $\mathsf{y}$ with $(s,k,\mathsf{v})$ respectively), we may write $x \mathbf{R}_t \mathsf{y}^{*}$ as
\begin{align*}
    & \sum_{r,s \leq t} \sum_{k \leq K_r,l \leq K_s} \sum_{v \in V \backslash A , \mathsf{u} \in \mathsf{V} \backslash \mathsf{A} } x(r,k,v) \mathsf{y}(s,l, \mathsf{u}) \mathbf{R}_t((r,k,v);(s,l,\mathsf{u}))\\
    = & \sum_{r,s \leq t} \sum_{k \leq K_r,l \leq K_s} \sum_{v \in V \backslash A, \mathsf{u} \in \mathsf{V} \backslash \mathsf{A} } x(r,k,v) \mathsf{y}(s,l, \mathsf{u}) \Hat{\mathbb{E}} \Big[ \langle \eta^{(r)}_k,D^{(r)}_v \rangle \langle \eta^{(s)}_l,\mathsf{D}^{(s)}_{\mathsf{u}} \rangle \Big]   \\
    = & \Hat{\mathbb{E}} \Big[ \sum_{s\leq t} \sum_{k \leq K_s} \sum_{v \in V \backslash A} x(s,k,v) \langle \eta^{(s)}_k,D^{(s)}_v \rangle, \sum_{s \leq t} \sum_{k \leq K_s} \sum_{\mathsf{v} \in \mathsf{V} \backslash \mathsf{A}} \mathsf{y}(s,k,\mathsf{v}) \langle \eta^{(s)}_k,\mathsf{D}^{(s)}_{\mathsf{v}} \rangle \Big]\,.
\end{align*}
Note that $\sum x(s,k,v) \langle \eta^{(s)}_k,D^{(s)}_v \rangle \in \mathrm{span} \{ \Hat{G}_{z,w}  \}$, $\sum \mathsf{y}(s,k,\mathsf{v}) \langle \eta^{(s)}_k,\mathsf{D}^{(s)}_{\mathsf{v}} \rangle \in \mathrm{span} \{ \Hat{\mathsf{G}}_{\mathsf{z},\mathsf{w}} \}$. In addition, for any $a_{u, v}, b_{u, v}\in \mathbb R$ we have
\begin{align*}
    &\mathbb{E} \Big[ \Big(\sum_{u,v \in V} a_{u,v} \Hat{G}_{u,v} \Big)  \Big(\sum_{u,v \in V} {b_{u,v}} \Hat{\mathsf{G}}_{\pi(u),\pi(v)} \Big) \Big] =\frac{1}{2} \epsilon \sum_{u,v \in V} (a_{u,v} b_{u,v} + a_{u,v} b_{v,u}) \\
    \leq &\epsilon \Big(\sum_{u,v \in V} a_{u,v}^2 \Big)^{\frac{1}{2}} \Big(\sum_{u,v \in V} {b_{u,v}^2} \Big)^{\frac{1}{2}}
    =\epsilon [\mathbb E( \sum_{u,v \in V} a_{u,v} \Hat{G}_{u,v} )^2]^{1/2} [\mathbb E ( \sum_{u,v \in V} {b_{u,v}} \Hat{\mathsf{G}}_{\pi(u),\pi(v)})^2]^{1/2}\,.
\end{align*}
Therefore, we have (below the sums are over $s\leq t, k\leq K_s, v\in V\setminus A$ (for $x$-vector) and $s\leq t, k\leq K_s, \mathsf v\in \mathsf V\setminus \mathsf A$ (for $y$-vector) respectively)
\begin{align*}
    x\mathbf{R}_t \mathsf{y}^{*}& = \Hat{\mathbb{E}} \Big[ \sum x(s,k,v) \langle \eta^{(s)}_k,D^{(s)}_v \rangle, \sum \mathsf{y}(s,k, \mathsf{v}) \langle \eta^{(s)}_k,\mathsf{D}^{(s)}_{\mathsf{v}} \rangle   \Big]     \\
    & \leq \epsilon \Hat{\mathbb{E}} \Big[ \Big( \sum x(s,k,v) \langle \eta^{(s)}_k,D^{(s)}_v \rangle \Big)^2 \Big]^{\frac{1}{2}} \Hat{\mathbb{E}} \Big[ \Big( \sum \mathsf{y}(s,k, \mathsf{v}) \langle \eta^{(s)}_k, \mathsf{D}^{(s)}_{\mathsf{v}} \rangle  \Big)^2 \Big]^{\frac{1}{2}}  \\
    &\leq \epsilon \left (\sum x(s,k,v)^2 (1+K_s \Delta_s) \right)^{\frac{1}{2}} \left ( \sum \mathsf{y}(k,s,\mathsf{v})^2 (1 + K_s \Delta_s) \right)^{\frac{1}{2}}   \\
    & \leq (\epsilon+ K_t\Delta_t) \| \mathsf{x} \|_2 \|  \mathsf{y} \|_2\,,
\end{align*}
where the second inequality follows from  \eqref{equ_degree_correlation_1}, \eqref{equ_degree_correlation_2}, \eqref{equ_degree_correlation_3}, \eqref{equ_degree_variance_1}  and \eqref{equ_degree_variance_2}.
This implies that
\begin{align}
    \| \mathbf{R}_t \|_{\mathrm{op}} , \| \mathbf{R}_t^{*}  \|_{\mathrm{op}} \leq \epsilon+ K_t \Delta_t\,.
    \label{equ_bound_op_norm_R_t}
\end{align}
By \eqref{equ_degree_variance_1}  and \eqref{equ_degree_variance_2} we see that $\| \mathbf{P}_t - \mathrm{I} \|_{\mathrm{op}} , \| \mathbf{L}_t - \mathrm{I} \|_{\mathrm{op}} \leq K_t \Delta_t$. In addition, we have
\begin{align*}
    & \Big\|  \begin{pmatrix}
    x & \mathsf{y}
    \end{pmatrix}
    \begin{pmatrix}
    \mathbf{P}_t & \mathbf{R}_t \\ 
    \mathbf{R}_t^* & \mathbf{L}_t   
    \end{pmatrix} \Big\|_2 
    =\| \begin{pmatrix}
    x \mathbf{P}_t + \mathsf{y} \mathbf{R}_t^{*} &
    x \mathbf{R}_t + \mathsf{y} \mathbf{L}_t
    \end{pmatrix} \|_2  \\
    &= \sqrt{ \|  x \mathbf{P}_t + \mathsf{y} \mathbf{R}_t^{*}  \|^2_2  +  \|  x \mathbf{R}_t + \mathsf{y} \mathbf{L}_t   \|^2_2  }  
    \geq \sqrt{ ( \|  x \mathbf{P}_t \|_2 -  \| \mathsf{y} \mathbf{R}_t^{*}  \|_2 )^2 + ( \|  \mathsf{y} \mathbf{L}_t \|_2  - \|  x \mathbf{R}_t  \|_2 )^2  }  \\
    & \geq \Big(  \| x \|_2 \Big( \| x \mathbf{P}_t \|_2 - \| \mathsf{y} \mathbf{R}_t^{*}  \|_2 \Big) + \| \mathsf{y} \|_2 \Big( \|  \mathsf{y} \mathbf{L}_t \|_2  - \|  x \mathbf{R}_t  \|_2 \Big)  \Big)   \Big/  \sqrt{ \| x \|_2^2 + \|  \mathsf{y} \|_2^2 } 
\end{align*}
by Cauchy-Schwartz inequality. Using the control on the operator norms of $\mathbf{P}_t, \mathbf{L}_t$ and $\mathbf{R}_t$, this is lower-bounded by
\begin{align*}
    & \frac{  \| x \|_2 \Big( (1 - K_t \Delta_t) \| x  \|_2 - (\epsilon + K_t \Delta_t) \| \mathsf{y} \|_2) + \| \mathsf{y} \|_2 \Big( (1 - K_t \Delta_t) \|  \mathsf{y}  \|_2 - (\epsilon + K_t \Delta_t) \|  x  \|_2 \Big)  }{  \sqrt{ \| x \|_2^2 + \|  \mathsf{y} \|_2^2 } }       \\
    & \geq (1 - \epsilon - 2K_t\Delta_t) \sqrt{ \| x \|^2 + \| \mathsf{y} \|^2 } = (1 - \epsilon - 2K_t\Delta_t) \| \begin{pmatrix}  x & \mathsf{y} \end{pmatrix}  \|_2\,.
\end{align*}
Therefore, the operator norm  of $\mathbf Q_t$ is upper-bounded by $\frac{1}{1-\epsilon-2K_t\Delta_t}$ as required. 
\end{proof}

\begin{Remark}\label{rem-epsilon-assumption} 
    In the lemma above, we made a technical assumption that $\epsilon\in (0,0.9)$.  The reason is that when $\epsilon$ approaches 1, the linear spaces spanned by $\{ \Hat{G}_{u,v} \}$ and by $\{ \Hat{\mathsf{G}}_{\mathsf{u,v}} \}$ tend to have non-trivial intersections, and thus our selected basis $\{ \langle \eta^{(t)}_k, D^{(t)}_v \rangle , \langle \eta^{(t)}_k, \mathsf{D}^{(t)}_{\mathsf{v}} \rangle  \}$ does not behave well. In particular, when $\epsilon = 1$, those variables are actually linearly dependent. To address this technical issue, we can employ the same operation of adding noise as described in Remark~\ref{rem-epsilon-u-0}. Also, when $\epsilon \to 1$ it is possible that simpler analysis shall work (e.g., for $\epsilon = 1$ the analysis can be hugely simplified) but we omit the consideration here. For a similar technical reason, we assume $\epsilon \leq 0.1$ in the lemma below.
  \end{Remark}

\begin{Lemma}{\label{lemma_1_norm_Q}}
    Suppose $\epsilon \leq 0.1$. On the event $\mathcal{E}_{t}$, we have $\| \mathbf{Q}_t  \|_{1} \leq \frac{10000 K_t^8}{\alpha^4} $.
\end{Lemma}
\begin{proof}
By a standard computation, we have that $\mathbf Q_t$ is equal to 
    \begin{equation}\label{eq-standard-linear-algebra-inverse}
    \begin{pmatrix} 
    \mathbf{P}_t^{-1} + \mathbf{P}_t^{-1} \mathbf{R}_t (\mathbf{L}_t - \mathbf{R}_t^{*} \mathbf{P}_t^{-1} \mathbf{R}_t)^{-1} \mathbf{R}_t^{*} \mathbf{P}_t^{-1} &  -\mathbf{P}_t^{-1} \mathbf{R}_t (\mathbf{L}_t-\mathbf{R}_t^{*} \mathbf{P}_t^{-1} \mathbf{R}_t)^{-1} \\
    - (\mathbf{L}_t-\mathbf{R}_t^{*} \mathbf{P}_t^{-1} \mathbf{R}_t)^{-1} \mathbf{R}_t^{*} \mathbf{P}_t^{-1} &(\mathbf{L}_t-\mathbf{R}_t^{*} \mathbf{P}_t^{-1} \mathbf{R}_t)^{-1}
    \end{pmatrix}\,.
    \end{equation}
By \eqref{equ_bound_correlation_dif_vertices} \eqref{equ_bound_correlation_same_vertices} and \eqref{equ_bound_correlation_same_vert_same_time}, the 1-norms of $\mathbf{R}_t$ and $\mathbf{R}^{*}_t$ are upper-bounded by $\frac{5K_t^2}{\alpha}$; by \eqref{equ_degree_variance_1} and \eqref{equ_degree_variance_2}, the 1-norms of $\mathbf{P}_t, \mathbf{P}_t^{-1}$ and $\mathbf{L}_t, \mathbf{L}_t^{-1}$ are upper-bounded by $1+2K_t \Delta_t$. Combined with \eqref{eq-standard-linear-algebra-inverse} and  the fact that $\| \mathbf{AB} \|_1 \leq \| \mathbf{A} \|_1  \| \mathbf{B} \|_1$, it yields that  
\begin{align*}
   \| \mathbf Q_t\|_{1}
    \leq \frac{100 K_t^4}{\alpha^2} \left \|  (\mathbf{L}_t - \mathbf{R}_t^{*} \mathbf{P}_t^{-1} \mathbf{R}_t)^{-1}  \right \|_{1}\,.
\end{align*}
So it suffices to show that 
\begin{equation}\label{eq-to-show-1-norm}
\| (\mathbf{L}_t -\mathbf{R}_t \mathbf{P}_t^{-1} \mathbf{R}_t^{*})^{-1} \|_1 \leq \frac{64 K_t^4}{\alpha^2}  \,.
\end{equation} 
Recall that for an $N*N$ matrix $\mathbf{A}$, the 1-norm $\| \mathbf{A} \|_1$ is the operator norm of the linear transform $\mathbf{A}$ on the Banach space $(\mathbb{R}^N,\| \cdot  \|_{\infty})$.  Let $N = \sum_{s \leq t} K_s (n-\kappa) = O(K_t n)$ be the dimension of $\mathbf{R}_t, \mathbf{L}_t$ and $\mathbf{P}_t$. So, in order to show \eqref{eq-to-show-1-norm}, it suffices to show that 
\begin{equation}\label{eq-to-show-infinite-norm}
\| (\mathbf{L}_t - \mathbf{R}_t \mathbf{P}_t^{-1} \mathbf{R}_t^{*}) x  \|_{\infty} \geq \frac{\alpha^2}{64K^4_t} \| x \|_{\infty} \mbox{ for all } x \in \mathbb{R}^N\,. 
\end{equation}

Without loss of generality, we may assume $\| x \|_{\infty}=1$. Recall from \eqref{equ_vector_unit}, \eqref{equ_bound_correlation_dif_vertices} \eqref{equ_bound_correlation_same_vertices} and \eqref{equ_bound_correlation_same_vert_same_time} that
\begin{align*}
    \Big| \mathbf{R}_t((s,k,v);(r,l,\mathsf{u})) \Big| = \Big| \Hat{\mathbb{E}} \Big[ \langle \eta^{(s)}_k,D^{(s)}_v \rangle \langle \eta^{(r)}_l, \mathsf{D}^{(r)}_{\mathsf{u}} \rangle \Big] \Big| \leq
    \begin{cases}
    2\varepsilon_s , & s=r, k=l, \mathsf{u}= \pi(v);  \\
    \frac{4K_s \Delta_s}{\alpha}  ,  & (s,k) \neq (r,l), s \geq r, \mathsf{u}= \pi(v) ; \\
    \frac{4K_s}{\alpha n}, & \mathsf{u} \neq \pi(v), s \geq r.
    \end{cases}
\end{align*}
Since $\mathbf{P}_t^{-1}$ is a diagonal matrix with entries  in $(1 - K_t \Delta_t, 1+K_t \Delta_t)$, we have
\begin{align*}
    \Big| \mathbf{R}_t \mathbf{P}_t^{-1} \mathbf{R}_t^{*} ((s,k,v);(r,l,u)) \Big| & \leq \sum_{p,m,\mathsf{w}} (1+ 2K_t\Delta_t) \Big| \mathbf{R}_t((s,k,v);(p,m, \mathsf{w}))\mathbf{R}_t((r,l,u);(p,m, \mathsf{w}))\Big| \\
    & \leq \begin{cases}
    5\varepsilon^2_s ,     & s=r, k=l,u=v ;   \\
    32 K^3_t \Delta_t / \alpha^2 , & (s,k) \neq (r,l), u=v; \\
    32 K_t^3 / (\alpha^2 n);       & u \neq v.
    \end{cases}
\end{align*}
We next divide our proof into two cases.

\noindent {\bf Case 1}: $\|x\|_1 \leq \frac{\alpha^2 n}{64 K^3_t}$. Since $\| x \|_\infty \leq 1$, we have
\begin{align*}
    \Big| \Big( \mathbf{R}_t \mathbf{P}_t^{-1} \mathbf{R}_t^{*} x \Big) (s,k,v) \Big| &= \Big| \sum_{r,l,u} \mathbf{R}_t \mathbf{P}_t^{-1} \mathbf{R}_t^{*} ((s,k,v);(r,l,u))x(r,l,u) \Big| \\
    &\leq 5\varepsilon_s^2 + \frac{32K^3_t\Delta_t}{\alpha^2} \cdot 2K_t + \frac{32K_t^3}{\alpha^2 n} \|x\|_1 \leq \frac{2}{3}\,.
\end{align*}
This implies that $\| (\mathbf{L}_t-\mathbf{R}_t \mathbf{P}_t^{-1} \mathbf{R}_t^{*})x   \|_{\infty} \geq \| \mathbf{L}_t x  \|_{\infty}- \| \mathbf{R}_t \mathbf{P}^{-1}_t \mathbf{R}_t^{*} x  \|_{\infty} \geq \frac{1}{4}$, which is stronger than \eqref{eq-to-show-infinite-norm}. 

\noindent {\bf Case 2}:  $\|x\|_1> \frac{\alpha^2 n}{64 K^3_t}$. In this case, using the dimension of $x$ is $N \in (K_t n, 4K_t n)$, we have $\| x \|_2 \geq \frac{\alpha^2 \sqrt{n}}{16K_t^{3.5}}$  by Cauchy-Schwartz inequality. By \eqref{equ_bound_op_norm_R_t} and \eqref{equ_degree_variance_1}, we have
\begin{align*}
    \| \mathbf{R}_t \mathbf{P}_t^{-1} \mathbf{R}_t^{*}  \|_{\mathrm{op}} \leq \| \mathbf{R}_t \|_{\mathrm{op}} \|  \mathbf{P}_t^{-1} \|_{\mathrm{op}}   \| \mathbf{R}_t^{*} \|_{\mathrm{op}}   \leq (\epsilon+3K_t\Delta_t)^2 (1 + 2 K_t\Delta_t) = \epsilon^2 (1+o(1))\,.
\end{align*}
This implies that $\|  (\mathbf{L}_t-\mathbf{R}_t \mathbf{P}_t^{-1} \mathbf{R}_t^{*})x  \|_2 \geq (1-\epsilon^2- o(1) ) \| x \|_2 \geq \frac{ \alpha^2 \sqrt{n}}{32K_t^{3.5}}$, and so there must exist an entry of $(\mathbf{L}_t-\mathbf{R}_t \mathbf{P}_t^{-1} \mathbf{R}_t^{*})x$ greater than $\frac{\alpha^2}{64 K_t^4}$, yielding \eqref{eq-to-show-infinite-norm}.
\end{proof}

\begin{Corollary}{\label{corollary_infty_2_norm_Q}}
    On $\mathcal E_t$ we have the following estimates on the entries of $\mathbf{Q}_t$:
    \begin{align}
        &\sum_{r,l,u} \mathbf{Q}_t( (s,k,v);(r,l,u) )^2 + \sum_{r,l,\mathsf{u}} \mathbf{Q}_t((s,k,v);(r,l,\mathsf{u}))^2 \leq 10000   \label{equ_bound_2_norm_row_Q} \,, \\
        &\sum_{r,l,u} \Big| \mathbf{Q}_t( (s,k,v);(r,l,u) ) \Big| + \sum_{r,l,\mathsf{u}} \Big| \mathbf{Q}_t ((s,k,v);(r,l,\mathsf{u}))\Big| \leq 10000K_t^8/ \alpha^4 \,.\label{equ_bound_1_norm_row_Q}
    \end{align}
\end{Corollary}
\begin{proof}
Since $\mathbf{Q}_t$ is symmetric, its operator norm is at least the 2-norm of each column (or each row). Thus, \eqref{equ_bound_2_norm_row_Q} follows from Lemma~\ref{lemma_op_norm_Q}.  In addition, \eqref{equ_bound_1_norm_row_Q} follows from  Lemma~\ref{lemma_1_norm_Q}.
\end{proof}

\subsection{Proof of proposition \ref{prop_admissible} }
The proof is by induction on $t$. Recall the definition of $\mathcal{E}_t$. Define
\begin{align*}
    \mathcal{A}_t & = \Big\{ \Big| \textup{PROJ}(\langle \eta^{(t)}_k, D^{(t)}_v \rangle; \mathfrak{S}_{t-1}) \Big| , \Big| \textup{PROJ}(\langle \eta^{(t)}_k, \mathsf{D}^{(t)}_{ \mathsf{v}} \rangle; \mathfrak{S}_{t-1}) \Big| \leq  K_t^{20} (\log n)   \Delta_t  \mbox{ for all } v, \mathsf v\Big\}\,,\\
    \mathcal{B}_t & = \Big\{ \Big| \langle \eta^{(t)}_k,D^{(t)}_v \rangle \Big|, \Big| \langle \eta^{(t)}_k, \mathsf{D}^{(t)}_{\mathsf{v}} \rangle \Big| \leq \log n  \mbox{ for all } v, \mathsf v    \Big\}  \,, \\
    \mathcal{C}_t & = \Big \{ \mathrm{Var}(\textup{GAUS}(\langle \eta^{(t)}_k, D^{(t)}_v \rangle; \mathfrak{S}_{t-1})), \mathrm{Var}(\textup{GAUS}(\langle \eta^{(t)}_k, \mathsf{D}^{(t)}_{\mathsf{v}} \rangle ; \mathfrak{S}_{t-1}) )  \leq K_{t}^{20} \Delta^2_{t}   \mbox{ for all } v, \mathsf v\Big \}\,,
\end{align*}
where in the above for all $v, \mathsf v$ means for all $v\in V \backslash A , \mathsf{v} \in \mathsf {V} \backslash \mathsf{A}$. Note that $\mathcal{A}_0, \mathcal{C}_0$ hold obviously since $\mathfrak{S}_{-1}$ denotes the trivial $\sigma$-algebra, and by a simple union bound $\mathcal B_0$ holds with probability at least $1 - O(n)e^{-(\log n)^2/2}$. Our inductive proof consists of the following steps:

\noindent {\bf Step 1.} $\mathcal{E}_0$ holds with high probability;

  \noindent{\bf Step 2.}  If $\mathcal{A}_t,\mathcal{B}_t,\mathcal{C}_t, \mathcal{E}_{t+1}$ holds for $t \leq t^{*}-1$, $\mathcal{A}_{t+1}$ holds with high probability;
   
  \noindent {\bf Step 3.}  If $\mathcal{B}_t, \mathcal{C}_t, \mathcal{E}_{t+1}, \mathcal{A}_{t+1}$ holds for $t \leq t^{*}-2$, $\mathcal{B}_{t+1},\mathcal{C}_{t+1}$ hold with high probability;
  
\noindent {\bf Step 4.} If $\mathcal{A}_{t+1}, \mathcal{B}_{t+1} , \mathcal{C}_{t+1}, \mathcal{E}_{t+1}$ holds for $t \leq t^{*}-2$, $\mathcal{E}_{t+2}$ holds with high probability.

\subsubsection{Step 1: $\mathcal{E}_0$ }
Since
$
    \frac{|\Gamma^{(0)}_k|}{n} - \alpha = \frac{1}{n} \sum_{v \in V \backslash A} ( \mathbf{I}_{v \in \Gamma^{(0)}_k} - \alpha) = \frac{1}{n} \sum_{v \in V \backslash A} ( \mathbf{I}_{ |G_{v,u_k}| \geq 10 }   - \mathbb{P}[ N(0,1) \geq 10] )
$
is a sum of mean zero i.i.d.\ Bernoulli random variables, by Lemma \ref{lemma_Berstein_ine} we have
\begin{align}
    \mathbb{P} \Big\{ \Big| \frac{|\Gamma^{(0)}_k|}{n} - \alpha \Big| > n^{-0.1} \Big\} \leq  2 \exp \{ - n^{0.8}  \}\,.
    \label{equ_concentration_gamma_0}
\end{align}
A similar concentration bound holds for $\frac{|\Pi^{(0)}_k|}{n} - \alpha$.  In addition, recalling $\alpha = \mathbb{P}[ |N(0,1)| \geq 10 ]$ and \eqref{equ_initial_set}, we see that $\frac{|\Gamma^{(0)}_k \cap \Gamma^{(0)}_l |}{n} - \alpha^2, \frac{|\Pi^{(0)}_k \cap \Pi^{(0)}_l |}{n} - \alpha^2, \frac{|\pi(\Gamma^{(0)}_k) \cap \Pi^{(0)}_l|}{n} -  \alpha^2, \frac{|\pi(\Gamma^{(0)}_k) \cap \Pi^{(0)}_k|}{n} - \phi(\epsilon)$ can all be written as sums of i.i.d.\ mean-zero Bernoulli variables. For instance,
\begin{align*}
    \frac{|\pi(\Gamma^{(0)}_k) \cap \Pi^{(0)}_k|}{n} - \phi(\epsilon) = \frac{1}{n} \sum_{v \in V \backslash A} \Big( \mathbf{I}_{ \Big\{ |G_{v,u_k}| , |\mathsf{G}_{\pi(v), \pi(u_k)}| \geq 10 \Big\} } - \phi(\epsilon) \Big)\,.
\end{align*}
Thus we can obtain a similar concentration bound as in \eqref{equ_concentration_gamma_0} using Lemma \ref{lemma_Berstein_ine}. By a union bound, we get that 
\begin{equation}\label{eq-mathcal-E-0-bound}
\mathbb{P} [\mathcal{E}_0] \geq 20 n K_0^2 \exp \{ - n^{0.8} \}\,.
\end{equation}

\subsubsection{Step 2: $\mathcal{A}_{t+1}$}
Our strategy  is to bound the conditional tail probability for $\eqref{equ_projection_part} \geq K_{t+1}^{20} (\log n) \Delta_{t+1}$ given $\mathfrak{S}_{t-1}$. Recall that in \eqref{equ_degree_correlation_1} we have shown $\Hat{\mathbb{E}} [\langle \eta^{(t+1)}_k,D^{(t+1)}_v \rangle \langle \eta^{(s)}_m , D^{(s)}_u \rangle ] =0 $ for any $s \leq t, u \in V$. Thus, we may expand the matrix product in \eqref{equ_projection_part} into a straightforward summation  as follows:
\begin{equation}
    \begin{aligned}
    \sum_{s,r=1}^{t} \sum_{l=1}^{K_r} \sum_{m=1}^{K_s} \sum_{\mathsf{u},w}  \Hat{\mathbb{E}} \Big[\langle \eta^{(t+1)}_k,D^{(t+1)}_v \rangle \langle \eta^{(s)}_m, \mathsf{D}^{(s)}_{\mathsf{u}} \rangle \Big] \mathbf{Q}_t((s,m,\mathsf{u});(r,l,w))\Big( \langle \eta^{(r)}_l,D^{(r)}_{w} \rangle \Big)  \\
    + \sum_{s,r=1}^{t} \sum_{l=1}^{K_r} \sum_{m=1}^{K_s} \sum_{\mathsf{u},\mathsf{w}} \Hat{\mathbb{E}} \Big[\langle \eta^{(t+1)}_k, D^{(t+1)}_v \rangle \langle \eta^{(s)}_m, \mathsf{D}^{(s)}_{\mathsf{u}} \rangle \Big]  \mathbf{Q}_t((s,m,\mathsf{u});(r,l,\mathsf{w}))\Big( \langle \eta^{(r)}_l, \mathsf{D}^{(r)}_{\mathsf{w}} \rangle  \Big)\,.
    \label{equ_explicit_projection}
    \end{aligned}
\end{equation}
Clearly \eqref{equ_explicit_projection} can be upper-bounded by
\begin{align*}
    4K_t^2 \max_{s,r,l,m} \Big \{ &\sum_{\mathsf{u},w} \Hat{\mathbb{E}} \Big[ \langle \eta^{(t+1)}_k,D^{(t+1)}_v \rangle \langle \eta^{(s)}_m,\mathsf{D}^{(s)}_{\mathsf{u}} \rangle \Big] \mathbf{Q}_t((s,m, \mathsf u);(r,l,w)) \Big(\langle \eta^{(r)}_l, D^{(r)}_{w} \rangle \Big) , \\
    & \sum_{\mathsf{u},\mathsf{w}} \Hat{\mathbb{E}} \Big[\langle \eta^{(t+1)}_k, D^{(t+1)}_v \rangle \langle \eta^{(s)}_m, \mathsf{D}^{(s)}_{\mathsf{u}} \rangle \Big] \mathbf{Q}_t((s,m, \mathsf u);(r,l,\mathsf{w}))\Big(\langle \eta^{(r)}_l, \mathsf{D}^{(r)}_{\mathsf{w}} \rangle \Big) \Big \}\,.
\end{align*}
Thus, it suffices to bound every term in the maximum above. For simplicity, we only demonstrate how to bound terms of the form:
\begin{align}
    \sum_{u,w} \Hat{\mathbb{E}} \Big[\langle \eta^{(t+1)}_k,D^{(t+1)}_v \rangle \langle \eta^{(s)}_m, \mathsf{D}^{(s)}_{\pi(u)} \rangle \Big] \mathbf{Q}_t((s,m,\pi(u));(r,l,w))
    \Big(\langle \eta^{(r)}_l,D^{(r)}_{w} \rangle \Big)\,.
    \label{equ_one_part_projection}
\end{align}
\begin{Lemma}{\label{claim_one_part_projection}}
    We have for all $s, m, l, r$
     $$\mathbb{P} \left[ \eqref{equ_one_part_projection} \geq \frac{1}{4} K_{t+1}^{20} K_t^{-2} (\log n)  \Delta_{t+1} ; \mathcal{A}_t,\mathcal{B}_t,\mathcal{C}_t,\mathcal{E}_{t+1}  \right] \leq O(n) e^{- (\log n)^2/2} \,.$$
\end{Lemma}
Since the bounds for other terms are similar, by Lemma~\ref{claim_one_part_projection} and a union bound, we get that $\mathbb{P}[\eqref{equ_explicit_projection} \geq K_{t+1}^{20} (\log n) \Delta_{t+1} ; \mathcal{A}_t,\mathcal{B}_t,\mathcal{C}_t,\mathcal{E}_{t+1} ] \leq  10 K_t^2 n^2 \exp \{ -( \log n )^2/2\} $. In addition, by a similar argument, we can bound the lower tail probability for \eqref{equ_explicit_projection}. Applying a union bound then yields that
\begin{equation}\label{prob-mathcal-A-t+1}
\mathbb{P}[\mathcal{A}_{t+1}^{c} ; \mathcal{A}_t,\mathcal{B}_t,\mathcal{C}_t,\mathcal{E}_{t+1} ] \leq  O(1)K_t^2 n^2 e^{-(\log n)^2/2}\,.
\end{equation}

\begin{proof}[Proof of Lemma~\ref{claim_one_part_projection}]
First we consider the terms in the sum  \eqref{equ_one_part_projection} with $u= v$, which can be written as
\begin{align}
    \sum_{w} \Hat{\mathbb{E}} \Big[ \langle \eta^{(t+1)}_k,D^{(t+1)}_v \rangle \langle \eta^{(s)}_m , \mathsf{D}^{(s)}_{\pi(v)} \rangle \Big] \mathbf{Q}_t ((s,m,\pi(v));(r,l,w))\Big( \langle \eta^{(r)}_l,D^{(r)}_w \rangle \Big)\,.
    \label{equ_proj_v=u}
\end{align}
By \eqref{equ_bound_correlation_same_vertices}, on the event $\mathcal{E}_{t+1}$ we have 
$$\Hat{\mathbb{E}} \left[ \langle \eta^{(t+1)}_k,D^{(t+1)}_v \rangle \langle \eta^{(s)}_m, \mathsf{D}^{(s)}_{\pi(v)} \rangle \right] = \eta^{(t+1)}_k \mathbf{P}_{\Gamma,\Pi}^{(t+1,s)} \left(\eta^{(s)}_m \right)^{*} \leq K_{t+1} \Delta_{t+1}   \,.$$ 
Also on $\mathcal{B}_{t}$ we know $| \langle \eta^{(r)}_l , D^{(r)}_w \rangle| \leq \log n $. Thus, 
\begin{align}
   \eqref{equ_proj_v=u} &\leq  K_{t+1} \Delta_{t+1} (\log n) \sum_{w} \Big| \mathbf{Q}_t((s,m,\pi(v));(r,l,w))\Big| \nonumber\\
   &\overset{\eqref{equ_bound_1_norm_row_Q}}{\leq}  10000 \alpha^{-4} K_t^8 K_{t+1}  (\log n) \Delta_{t+1} \leq \frac{1}{10} K_{t+1}^{20} K_t^{-2} (\log n) \Delta_{t+1}\,. \label{eq-bound-3.34}
\end{align}
Now consider the terms in the sum  \eqref{equ_one_part_projection} with $v \neq u$, which can be written as
\begin{align*}
\sum_{i,j} \eta^{(t+1)}_k (i) \eta^{(s)}_m (j)  \sum_{u ,w}  \frac{( \mathbf{I}_{ \pi(v) \in \Pi^{(s)}_j}-\alpha) ( \mathbf{I}_{u \in \Gamma^{(t+1)}_i} -\alpha) }{(\alpha - \alpha^2) n} \mathbf{Q}_t((s,m,\pi(u));(r,l,w)) \left(\langle \eta^{(r)}_l,D^{(r)}_w \rangle \right)\,.
\end{align*}
Since $ \sum_{i,j} |\eta^{(t+1)}_k (i) \eta^{(s)}_m (j)| \leq 4 K_{t+1}$, it suffices to bound 
\begin{align}
    \sum_{u ,w}  \frac{( \mathbf{I}_{ \pi(v) \in \Pi^{(s)}_j}-\alpha) ( \mathbf{I}_{u \in \Gamma^{(t+1)}_i} -\alpha) }{(\alpha - \alpha^2) n} \mathbf{Q}_t((s,m,\pi(u));(r,l,w))\left(\langle \eta^{(r)}_l,D^{(r)}_w \rangle \right)\,
    \label{equ_proj_v_not=u}
\end{align}
by $K_{t+1}^{10} (\log n) \Delta_{t+1}$ for every $i,j$.
We first consider the case when $r<t$. Conditioned on $\mathfrak{S}_{t-1}$, \eqref{equ_proj_v_not=u} can be organized as
\begin{equation}
    \sum_{u \in V} \lambda_u \left( \mathbf{I}_{\left \{ \left| \langle \sigma^{(t)}_i,D^{(t)}_u \rangle|{\mathfrak{S}_{t-1}} \right| \geq 10 \right \} } - \alpha  \right)\,,
    \label{equ_organized_v_not=u}
\end{equation}
where $\lambda_u$ is measurable with respect to $\mathfrak{S}_{t-1}$ and is given by
\begin{align*}
    \lambda_u = \sum_{w \in V} \frac{ \mathbf{I}_{ \pi(v) \in \Pi^{(s)}_j} - \alpha}{(\alpha - \alpha^2) n} \mathbf{Q}_t((s,m,\pi(u)); (r,l,w))\left( \langle \eta^{(r)}_l , D^{(r)}_w \rangle \right)\,.
\end{align*}
By \eqref{equ_bound_1_norm_row_Q}, on the event $\mathcal{A}_t \cap \mathcal{B}_t \cap \mathcal E_t$  we have 
\begin{align}
    | \lambda_u | \leq \sum_{w \in V} \frac{ \log n }{\alpha n} \Big| \mathbf{Q}_t ((s,m,\pi(u));(r,l,w)) \Big|  \leq \frac{ 10000 K^8_t (\log n)}{ \alpha^5 n }\,.
    \label{equ_bound_lambda_u}
\end{align}
Also, recall from Remark \ref{remark_projection_form} that
\begin{align*}
    \langle \eta^{(t)}_i , D^{(t)}_u \rangle |{\mathfrak{S}_{t-1}} \overset{d}{=} \textup{PROJ}(\langle \eta^{(t)}_i, D^{(t)}_u \rangle; \mathfrak{S}_{t-1}) + \left( \langle \eta^{(t)}_i , \Tilde{D}^{(t)}_u \rangle - \textup{GAUS}(\langle \eta^{(t)}_i, D^{(t)}_u \rangle; \mathfrak{S}_{t-1})  \right) \,.
\end{align*}
On the event $\mathcal{C}_t$ we have that $\textup{GAUS}(\langle \eta^{(t)}_i, D^{(t)}_u \rangle; \mathfrak{S}_{t-1})$ has variance bounded by $ K_{t}^{20} \Delta_{t}^2$.  By a union bound we have
\begin{align}\label{eq-Gaus-probability-bound}
    \mathbb P(\exists u\in V\setminus A: \textup{GAUS}(\langle \eta^{(t)}_i, D^{(t)}_u \rangle; \mathfrak{S}_{t-1}) \geq K_{t}^{10} (\log n) \Delta_{t}\,;\, \mathcal{C}_t) \leq n e^{- (\log n)^{2}/2}\,.
\end{align}
On the event $\mathcal{A}_t$, we have $|\textup{PROJ}(\langle \eta^{(t)}_i, D^{(t)}_u \rangle; \mathfrak{S}_{t-1})| \leq K_t^{20} (\log n)  \Delta_{t}$. In addition, recalling \eqref{equ_def_sigma} we have $|\textup{PROJ}(\langle \sigma^{(t)}_k, D^{(t)}_v \rangle; \mathfrak{S}_{t-1})|
, | \textup{GAUS}(\langle \sigma^{(t)}_k, D^{(t)}_v \rangle; \mathfrak{S}_{t-1})| \leq K_{t}^{21} (\log n) \Delta_{t} $. Thus, on the event $\mathcal{A}_t$ we have
\begin{align*}
    \left \{ \left| \langle \sigma^{(t)}_i,\Tilde{D}^{(t)}_u \rangle \right| \geq 10 + 2 K_{t}^{21} (\log n) \Delta_{t}  \right \} &\subset \left \{ \left | \langle \sigma^{(t)}_i,D^{(t)}_u \rangle |{\mathfrak{S}_{t-1}} \right | \geq 10 \right \}  \\
    &\subset \left \{ \left | \langle \sigma^{(t)}_i, \Tilde{D}^{(t)}_u \rangle \right | \geq 10 - 2K_{t}^{21} (\log n) \Delta_{t}  \right \}\,.
\end{align*}
Therefore, on the complement of the event described in \eqref{eq-Gaus-probability-bound} and on $\mathcal A_t$, 
we have
\begin{align}
    \eqref{equ_organized_v_not=u} \leq \sum_{u \in V} \lambda_u \left( \mathbf{I}_{ \left \{ \left| \langle \sigma^{(t)}_i,\Tilde{D}^{(t)}_u \rangle \right| + 2\mathrm{sgn}(\lambda_u) K_{t}^{21} (\log n) \Delta_{t} \geq 10 \right \} }  - \alpha    \right)\,,
    \label{equ_upper_bound_org}  \\
    \eqref{equ_organized_v_not=u} \geq \sum_{u \in V} \lambda_u \left( \mathbf{I}_{ \left \{ \left | \langle \sigma^{(t)}_i,\Tilde{D}^{(t)}_u \rangle \right | - 2\mathrm{sgn}(\lambda_u) K_{t}^{21} (\log n) \Delta_{t} \geq 10 \right \} }  - \alpha    \right)\,.
    \label{equ_lower_bound_org}
\end{align}
Since $\{  \langle \sigma^{(t)}_i,\Tilde{D}^{(t)}_u \rangle: u \in V \backslash A \}$ is a collection of independent variables, the concentration of (the right hand sides of)  \eqref{equ_upper_bound_org} and \eqref{equ_lower_bound_org} follows from their first and second moments (the concentration for \eqref{equ_lower_bound_org}  is used to control the lower tail probability for \eqref{equ_explicit_projection}, whose proof is omitted due to high similarity). To this end, on the complement of the event described in \eqref{eq-Gaus-probability-bound} and on $\mathcal A_t \cap \mathcal B_t$, the conditional first moment of \eqref{equ_upper_bound_org} is upper-bounded by
\begin{align*}
    & \left| \sum_{u \in V} \lambda_u \mathbb{E} \left [ \mathbf{I}_{ \left \{ \left | \langle \sigma^{(t)}_i,\Tilde{D}^{(t)}_u \rangle \right | + 2\mathrm{sgn}(\lambda_u) K_{t}^{21} (\log n) \Delta_{t} \geq 10 \right \} }  - \alpha \right ] \right| \\
    \leq & \sum_{u \in V} |\lambda_u| \left|\mathbb{E} \left [ \mathbf{I}_{ \left \{ \left | \langle \sigma^{(t)}_i, \Tilde{D}^{(t)}_u \rangle \right | + 2\mathrm{sgn}(\lambda_u) K_{t}^{21} (\log n) \Delta_{t} \geq 10 \right \}}  - \alpha \right ] \right| \\
    \leq & \frac{ 10000 K^{8}_t (\log n)}{ \alpha^5 n } \sum_{u \in V}  \left|\mathbb{E} \left [ \mathbf{I}_{ \left \{ \left | \langle \sigma^{(t)}_i,\Tilde{D}^{(t)}_u \rangle \right | + 2\mathrm{sgn}(\lambda_u) K_{t}^{21} (\log n) \Delta_{t} \geq 10 \right \} }  - \alpha \right ] \right| \\
    \leq & \frac{ 10000 K^8_t (\log n)}{ \alpha^5 } \cdot  \left \{ ( 1+ 2 K_{t}^{21} (\log n) \Delta_{t} ) -1 \right \} \leq 10^5 \alpha^{-5} K_{t}^{29} (\log n)^{2} \Delta_{t} \leq \Delta_{t+1} \,,
\end{align*}
where we recall $\Delta_t$ as in \eqref{eq-def-Delta-s}.
In addition, the variance of \eqref{equ_upper_bound_org} is upper-bounded by 
\begin{align*}
    \sum_{u \in V} \lambda_u^2 \mathrm{Var} \left( \mathbf{I}_{  \left \{ \left | \langle \sigma^{(t)}_i,\Tilde{D}^{(t)}_u \rangle \right | + 2 \mathrm{sgn}(\lambda_u) K_{t}^{21} (\log n) \Delta_{t} \geq 10 \right \} } \right) & \leq \sum_{u} \lambda_u^2 \\
    &  \overset{\eqref{equ_bound_lambda_u}}{\leq} \frac{ 10^8 K_{t}^{16} (\log n)^2 }{ \alpha^{10} n} \leq \frac{ K_{t+1}^{10} (\log n)^3 }{n}\,.
\end{align*}
Now, applying Lemma \ref{lemma_Berstein_ine} we get that with probability $1- o(e^{-n^{0.1}})$, \eqref{equ_upper_bound_org} is upper-bounded by $\frac{1}{10} K_{t+1}^{10} \Delta_{t+1}$. Combined with \eqref{eq-bound-3.34}, this completes the proof of the lemma in this case.

Next we consider the case when $r=t$. Conditioned on $\mathfrak{S}_{t-1}$, \eqref{equ_proj_v_not=u} can be organized as follows:
\begin{align}
    \sum_{u,w} \lambda_{u,w} \left( \mathbf{I}_{ \left \{ \left| \langle \sigma^{(t)}_i, D^{(t)}_u \rangle|{\mathfrak{S}_{t-1}} \right | \geq 10 \right \} }  - \alpha \right) \left(\langle \eta^{(t)}_l, D^{(t)}_w \rangle |{ \mathfrak{S}_{t-1} } \right)\,,
    \label{equ_org_proj_r=t}
\end{align}
where $\lambda_{u, w}$ is measurable with respect to $ \mathfrak{S}_{t-1}$ and is given by 
\begin{align*}
    \lambda_{u,w}= \frac{ \mathbf{I}_{ \pi(v) \in \Pi^{(s)}_j } - \alpha}{ (\alpha - \alpha^2) n} \mathbf{Q}_t((s,m,\pi(u));(t,l,w))\,.
\end{align*}
By \eqref{equ_bound_2_norm_row_Q} and \eqref{equ_bound_1_norm_row_Q}, on the event $\mathcal{E}_{t+1}$ we have
\begin{align}
    &\sum_{w} |\lambda_{u,w}| \leq \frac{2}{\alpha n} \sum_{w} \Big| \mathbf{Q}_t ( (s,m,\pi(u));(t,l,w) ) \Big| \leq \frac{ 10^5 K_t^8 }{ \alpha^5 n } \,,\label{equ_bound_lambda_u_w} \\
    &\sum_{w} \lambda_{u,w}^2 \leq \frac{4}{ (\alpha n)^2 } \sum_{w} \mathbf{Q}_t ( (s,m,\pi(u));(t,l,w) )^2 \leq \frac{ 10^5 }{ \alpha^2 n^2} \,.\label{equ_bound_square_lambda_u_w}
\end{align}
We adopt the similar approach as for the concentration of \eqref{equ_upper_bound_org}. Since the second moment is upper-bounded via \eqref{equ_bound_square_lambda_u_w}, the key is to control the first moment. Recalling  Remark~\ref{remark_projection_form},  it suffices to estimate
\begin{equation}
    \begin{aligned}
    \sum_{u,w} \lambda_{u,w} \Big( \mathbf{I}_{ \left \{ \left| \langle \sigma^{(t)}_i, \Tilde{D}^{(t)}_u \rangle - \textup{GAUS}(\langle \sigma^{(t)}_i, D^{(t)}_u \rangle; \mathfrak{S}_{t-1}) + \textup{PROJ}(\langle \sigma^{(t)}_i, D^{(t)}_u \rangle; \mathfrak{S}_{t-1})  \right| \geq 10 \right \} }  - \alpha \Big)  \\
    \times \left(\langle \eta^{(t)}_l, \Tilde{D}^{(t)}_w \rangle - \textup{GAUS}(\langle \eta^{(t)}_l, D^{(t)}_w \rangle \rangle; \mathfrak{S}_{t-1}) + \textup{PROJ}(\langle \eta^{(t)}_l, D^{(t)}_w \rangle; \mathfrak{S}_{t-1})  \right)\,.
    \label{equ_proj_r=t_conditioned}
    \end{aligned}
\end{equation}
By a similar argument as in the previous case, we may assume without loss of generality that $|\textup{GAUS}(\langle \eta^{(t)}_l, D^{(t)}_w \rangle; \mathfrak{S}_{t-1}) |, |\textup{PROJ} (\langle \eta^{(t)}_l, D^{(t)}_w \rangle; \mathfrak{S}_{t-1})| \leq K_t^{20} (\log n) \Delta_t $ for all $m , u$ (this assumption holds with probability at least $1 - ne^{(\log n)^2/2}$, which is why we may assume without loss; similar convention will be used in what follows for exposition convenience). Thus,
\begin{align*}
   \Big| \sum_{u,w} \lambda_{u,w} \Big( \mathbf{I}_{ \left \{ \left| \langle \sigma^{(t)}_i, \Tilde{D}^{(t)}_u \rangle - \textup{GAUS}(\langle \sigma^{(t)}_i, D^{(t)}_u \rangle; \mathfrak{S}_{t-1}) + \textup{PROJ}(\langle \sigma^{(t)}_i, D^{(t)}_u \rangle; \mathfrak{S}_{t-1})  \right| \geq 10 \right \} }  - \alpha \Big)  \\
    \times\left(- \textup{GAUS}(\langle \eta^{(t)}_l, D^{(t)}_w \rangle \rangle; \mathfrak{S}_{t-1}) + \textup{PROJ}(\langle \eta^{(t)}_l, D^{(t)}_w \rangle; \mathfrak{S}_{t-1})  \right)\Big|
\end{align*}
is upper-bounded by 
\begin{align*}
    2 K_t^{20} (\log n) \Delta_t \sum_{u,w} |\lambda_{u,w}| \overset{\eqref{equ_bound_lambda_u_w}}{\leq}  10^6 \alpha^{-5} K_t^{28} (\log n) \Delta_t \leq  \frac{1}{10}  \Delta_{t+1}\,.
\end{align*}
Now we deal with the part
\begin{align}
    \sum_{u,w} \lambda_{u,w} \left( \mathbf{I}_{ \left \{ \left| \langle \sigma^{(t)}_i, \Tilde{D}^{(t)}_u \rangle - \textup{GAUS}(\langle \sigma^{(t)}_i, D^{(t)}_u \rangle; \mathfrak{S}_{t-1}) + \textup{PROJ}(\langle \sigma^{(t)}_i, D^{(t)}_u \rangle; \mathfrak{S}_{t-1})  \right| \geq 10 \right \} }  - \alpha \right)  \cdot \left(\langle \eta^{(t)}_l, \Tilde{D}^{(t)}_w \rangle  \right)\,.
    \label{equ_uncleaned_proj_r=t}
\end{align}
We wish to approximate it with
\begin{align}
    \sum_{u,w} \lambda_{u,w} \left( \mathbf{I}_{ \left \{ \left| \langle \sigma^{(t)}_i, \Tilde{D}^{(t)}_u \rangle \right | \geq 10 \right \} }  - \alpha \right)    \langle \eta^{(t)}_l, \Tilde{D}^{(t)}_w \rangle\,.
    \label{equ_cleaned_proj_r=t}
\end{align}
By a union bound, we may assume without loss that $|\langle \eta^{(t)}_l, \Tilde{D}^{(t)}_w \rangle| \leq \log n$. In light of this, we define
\begin{equation}
\begin{aligned}
    &\mathcal{L}^{(t)}_{i} =  \Big\{  u : \Big| \langle \sigma^{(t)}_i, \Tilde{D}^{(t)}_u \rangle \Big| \geq 10 \Big\} \triangle \\
    &\Big\{ u : \Big| \langle \sigma^{(t)}_i, \Tilde{D}^{(t)}_u \rangle  - \textup{GAUS}(\langle \sigma^{(t)}_i, D^{(t)}_u \rangle; \mathfrak{S}_{t-1}) + \textup{PROJ}(\langle \sigma^{(t)}_i, D^{(t)}_u \rangle; \mathfrak{S}_{t-1}) \Big|  \leq 10    \Big\}\,,
    \label{equ_bias_set}
\end{aligned}
\end{equation}
where $\triangle$ denotes for the symmetric difference. Then, the approximation error $|\eqref{equ_uncleaned_proj_r=t}- \eqref{equ_cleaned_proj_r=t}|$ is upper-bounded by
\begin{align*}
\sum_{u,w} |\lambda_{u,w}| |\langle \eta^{(t)}_l, \Tilde{D}^{(t)}_w \rangle| \mathbf{I}_{ \left \{ u \in \mathcal{L}_i^{(t)}  \right \}}  \leq \sum_{u,w} \log n |\lambda_{u,w}| \mathbf{I}_{ \left \{ u \in \mathcal{L}_i^{(t)} \right \}}  \overset{\eqref{equ_bound_lambda_u_w}}{\leq} \frac{ 10^5 K_t^8 \log n}{ \alpha^5 n}  |\mathcal{L}^{(t)}_{i}|\,,
\end{align*}
which holds under various assumptions we made earlier. Recall our assumptions on $\textup{PROJ}(\langle \sigma^{(t)}_i, D^{(t)}_u \rangle ; \mathfrak{S}_{t-1})$ and $\textup{GAUS}(\langle \sigma^{(t)}_i, D^{(t)}_u \rangle ; \mathfrak{S}_{t-1})$, on the event $\mathcal{A}_{t}$ we have that with probability at least $1-n e^{-(\log n)^{2}}$ the following holds (which we then assume below):
\begin{align*}
    \mathcal{L}^{(t)}_i \subset \left \{ u : 10 - 2 K_t^{21} (\log n)  \Delta_{t} \leq \left| \langle \sigma^{(t)}_i, \Tilde{D}^{(t)}_u \rangle \right | \leq 10 + 2 K_t^{21} (\log n)  \Delta_{t}  \right \}\,.
\end{align*}
Therefore, $| (\ref{equ_uncleaned_proj_r=t}) - (\ref{equ_cleaned_proj_r=t}) |$ is upper-bounded by
\begin{align*}
    \frac{10^5 K_t^8 (\log n)}{\alpha^5 n} \left| \mathcal{L}^{(t)}_i \right| \leq \frac{ 10^5 K_t^8 (\log n)}{ \alpha^5 n} \sum_{u} \mathbf{I}_{ \left\{ 10 - 2 K^{21}_t (\log n)  \Delta_{t} \leq \left| \langle \sigma^{(t)}_i, \Tilde{D}^{(t)}_u \rangle \right | \leq  10 + 2 K^{21}_t (\log n)  \Delta_{t}  \right \} }\,,
\end{align*}
which, by Lemma \ref{lemma_Berstein_ine} again, is bounded by $10^6 \alpha^{-5} K_t^{29} (\log n)^{3} \Delta_{t} < \frac{1}{10} \Delta_{t+1}$ with probability at least $1 - e^{ -n^{0.1} }$. 

It remains to estimate \eqref{equ_cleaned_proj_r=t} as incorporated in Lemma~\ref{lemma_tail_estimate_cleaned} below (which then completes the proof of Lemma~\ref{claim_one_part_projection}).
\end{proof}
\begin{Lemma}{\label{lemma_tail_estimate_cleaned}}
Under the assumptions of \eqref{equ_bound_lambda_u_w}  and \eqref{equ_bound_square_lambda_u_w}, we have that with probability $o(e^{-n^{0.1}})$, 
\begin{align*}
    \Big| \sum_{u,w} \lambda_{u,w} \Big( \mathbf{I}_{ \left \{ \left| \langle \sigma^{(t)}_i, \zeta_u \rangle \right | \geq 10 \right \} }  - \alpha \Big)    \langle \eta^{(t)}_l, \zeta'_w \rangle\Big| \leq  \Delta_{t+1}\,,
\end{align*}
where either $\zeta_u = \Tilde D^{(t)}_u$ for all $u$ or $\Tilde {\mathsf D}^{(t)}_{\pi(u)}$ for all $u$, and either $\zeta'_w = \Tilde D^{(t)}_w$ for $w$ or $\zeta'_w = \Tilde {\mathsf D}^{(t)}_{\pi(w)}$ for all $w$. 
\end{Lemma}

\begin{proof}
We prove the case when $\zeta_u = \Tilde D^{(t)}_u$ and $\zeta'_w = \Tilde D^{(t)}_w$ and the other cases follow similarly. For notation convenience, denote $Z_u = \langle \eta^{(t)}_l, \Tilde{D}^{(t)}_u \rangle$, $W_u =  \langle \sigma^{(t)}_i, \Tilde{D}^{(t)}_u \rangle$ and  $b_u  = \mathbf{I}_{  \{ | W_u | \geq 10  \} } - \alpha$. Then $ b_u$'s and $ Z_u$'s are all mean-zero sub-Gaussian random variables. Let $\Lambda$ be matrix with $\Lambda_{u,w} = \lambda_{u,w}$. Then \eqref{equ_cleaned_proj_r=t} can be written as $\sum_{u,w} \lambda_{u,w} b_u Z_w = b \Lambda Z^{*}$, where $b= (b_u)_{u \in V \backslash A}$ and $Z=(Z_u)_{u \in V \backslash A}$. Define $\mathcal{F}_b = \sigma\{W_u: u \in V \backslash A \}$. Since $\{ W_u: u\in V\setminus A\}$ is a collection of independent Gaussian variables, we have
\begin{align*}
    \mathbb{E}[ Z_u | \mathcal{F}_b  ] = \sum_{w \in V \backslash A} \mathbb{E}[  Z_u W_w   ] W_w 
\end{align*}
where
\begin{equation}\label{eq-projection-coefficient-bound-Z-W}
 \mathbb{E}[ Z_u W_u ] = O(1) \mbox{ and }\mathbb{E}[ Z_u W_w ] = O( \frac{K_t}{n} ) \mbox{ for } w\neq u. 
 \end{equation}
 Now split (\ref{equ_cleaned_proj_r=t}) into
\begin{align}
    \eqref{equ_cleaned_proj_r=t} = \sum_{u,w} \lambda_{u,w} b_u Z_w &= \sum_{u,w} \lambda_{u,w} \left( \mathbf{I}_{ |W_u| \geq 10} - \alpha \right) ( \sum_{ v \not = u } \mathbb{E}[ Z_w W_v  ] W_v ) \label{eq-3.47-part-1} \\
    & + \sum_{u,w} \lambda_{u,w} \left( \mathbf{I}_{ |W_u| \geq 10} - \alpha \right) \left(  \mathbb{E}[ Z_w W_u  ] W_u \right)  \label{eq-3.47-part-2}\\
    &+ \sum_{u,w} \lambda_{u,w} \left ( \mathbf{I}_{ |W_u| \geq 10} - \alpha \right)  \left( Z_w - \mathbb{E}[ Z_w | \mathcal{F}_b ]  \right)\label{eq-3.47-part-3}\,.
\end{align}
First, we consider \eqref{eq-3.47-part-1}, which equals to
$\sum_{u \not = v} ( \sum_{w} \lambda_{u,w} \mathbb{E}[ Z_w W_v ]  )   ( \mathbf{I}_{ |W_u| \geq 10} - \alpha ) W_v
$.
Let $\Sigma$ be a matrix with $\Sigma_{u, v} = \mathbb{E}[ Z_u W_v ]$, and denote $\Xi$ the matrix with the same non-diagonal entries as $\Lambda \Sigma$ and with diagonal entries being 0. Then the sum can be written as $b \Xi W^{*}$, where $W$ is a vector with entries given by $W_u$.  We may write $\Sigma = D + H$, where $D$ is a diagonal matrix with entries bounded by 1 and $H$ is a matrix with entries bounded by $\frac{ K_t }{n}$. We then have $ \| \Lambda D \|_{\mathrm{HS}}^2 \leq \| \Lambda  \|^2
_{\mathrm{HS}} \leq \frac{K_t}{n}$  (by \eqref{equ_bound_square_lambda_u_w}, where we also need to sum over $u$)  and $ \| \Lambda H \|_{\mathrm{HS}}^2 \leq \| \Lambda  \|_{\mathrm{HS}}^2 \|H \|_{\mathrm{HS}}^2 \leq \frac{K_t^{3}}{n} $. Thus,  $\| \Lambda \Sigma  \|^2_{\mathrm{HS}}\leq \frac{K_t^3}{n}$. Applying Lemma \ref{lemma_decoupling_subGaussian}, we get that $|\eqref{eq-3.47-part-1}| \geq \frac{1}{10} \Delta_{t+1}$ with probability $o(e^{-n^{0.1}})$ (with room to spare).

Next, we consider \eqref{eq-3.47-part-2}, which equals to 
$\sum_{u} ( \sum_{w} \lambda_{u,w} \mathbb{E}[ Z_w W_u ] ) ( \mathbf{I}_{ |W_u| \geq 10} - \alpha ) W_u\,.
$
By \eqref{equ_bound_lambda_u_w} and \eqref{eq-projection-coefficient-bound-Z-W}, we have $\sum_{w} \lambda_{u,w} \mathbb{E}[ Z_w W_u ] \leq \frac{O(K_t^{10})}{n}$. Also note that $ \{  ( \mathbf{I}_{|W_u| \geq 10} - \alpha ) W_u: u\in V\setminus A  \} $ is a collection of mean-zero independent random variables, so by the standard Chernoff bound we know $|\eqref{eq-3.47-part-2}| \geq \frac{1}{10} \Delta_{t+1}$ with probability $o(e^{-n^{0.1}})$.

Finally, we consider \eqref{eq-3.47-part-3}, which equals  $b \Lambda \Tilde{Z}^{*}$, where $\Tilde{Z}$ is given by $\Tilde{Z}_w = Z_w - \mathbb{E}[ Z_w | \mathcal{F}_b ]$. There must exist a matrix $T$ and i.i.d.\ standard Gaussian vector $\Hat{Z}$ such that $\Tilde{Z} = \Hat{Z} T $, and $\Hat{Z}$ is independent with $\mathcal F_b$. So $\eqref{eq-3.47-part-3} = b \Lambda T^{*} (\Hat{Z})^{*}$ where $b,\Hat{Z}$ are sub-Gaussian random vectors with independent entries. By Corollary~ \ref{corollary_Hanson_Wright_independent_subGaussian}, it suffices to control the Hilbert-Schmidt norm of $\Lambda T^{*}$. To this end, note that $ T^{*} T = \mathbb{E} \{    ( \Hat{Z} T )^{*} \Hat{Z} T   \} = \mathbb{E} \{ \Tilde{Z}^{*} \Tilde{Z}  \} $. In addition,
\begin{align*}
    \mathbb{E}[ \Tilde{Z}_v \Tilde{Z}_w ] &= \mathbb{E} [ ( Z_w - \mathbb{E}[ Z_w | \mathcal{F}_b ] ) ( Z_v - \mathbb{E}[ Z_v | \mathcal{F}_b ] )  ] = - \mathbb{E}[ \mathbb{E}[ Z_w | \mathcal{F}_b ] \mathbb{E}[ Z_v | \mathcal{F}_b ]  ] \\
    &= - \mathbb{E} [ ( \sum_{h \in V \backslash A} \mathbb{E}[  Z_w W_h   ] W_h ) ( \sum_{h \in V \backslash A} \mathbb{E}[  Z_v W_h   ] W_h )  ] = O( \frac{K_t^2}{n} )
\end{align*}
for $v \not = w$,  where we have used \eqref{eq-projection-coefficient-bound-Z-W}. So $\|T^{*} T\|^2_{\mathrm{op}} \leq K_t^4$. Thus, we have
\begin{align*}
    \| \Lambda T^{*}   \|^2_{\mathrm{HS}} = \mathrm{tr} ( \Lambda T^{*} T \Lambda^{*} ) \leq K_t^4 \mathrm{tr} ( \Lambda \Lambda^{*} ) \leq K_t^4 \|  \Lambda \|^2_{\mathrm{HS}} \leq \frac{O(K_t^5)}{n}\,.
\end{align*}
Applying Corollary \ref{corollary_Hanson_Wright_independent_subGaussian}, we have  $|\eqref{eq-3.47-part-3}| \geq \frac{1}{10}  \Delta_{t+1}$ with probability $o(e^{-n^{0.1}})$. Combining preceding bounds on \eqref{eq-3.47-part-1}, \eqref{eq-3.47-part-2} and \eqref{eq-3.47-part-3}, we complete the proof of the lemma.
\end{proof}

\subsubsection{Step 3: $\mathcal{B}_{t+1}, \mathcal{C}_{t+1}$ }

Recalling Remark \ref{remark_projection_form}, we have that:
\begin{align*}
    \textup{GAUS}( \langle \eta^{(t+1)}_k, D^{(t+1)}_v \rangle; \mathfrak{S}_{t} ) = 
    \begin{pmatrix}
    H_t & \mathsf{H}_t
    \end{pmatrix}
    \mathbf{Q}_t  
    \begin{pmatrix}
    \Tilde{Y}_t^{*} \\ \Tilde{\mathsf{Y}}_t^{*}
    \end{pmatrix}\,.
\end{align*}
Thus, the variance of $\textup{GAUS}( \langle \eta^{(t+1)}_k, D^{(t+1)}_v \rangle ; \mathfrak{S}_t )$ is given by
\begin{align*}
    \begin{pmatrix}
    H_t & \mathsf{H}_t
    \end{pmatrix}
    \mathbf{Q}_t  
    \begin{pmatrix}
    H_t^{*} \\ \mathsf{H}_t^{*}
    \end{pmatrix} 
    \leq \| \mathbf{Q}_t \|_{\mathrm{op}} \| \begin{pmatrix} H_t & \mathsf{H}_t \end{pmatrix} \|_2^2  \leq K_{t+1}^{20} \Delta_{t+1}^2\,,
\end{align*}
where the last inequality holds on the event  $\mathcal{E}_{t+1}\cap \mathcal{A}_{t+1}$, thanks to
Lemma \ref{lemma_op_norm_Q} and \eqref{equ_bound_2_norm_H_t}.
This implies that 
\begin{equation}\label{eq-mathcal-C-containment}
\mathcal E_{t+1} \cap \mathcal A_{t+1} \subset \mathcal C_{t+1}\,.
\end{equation}
Now we control $\mathcal{B}_{t+1}$ assuming  $\mathcal{A}_{t+1}$. Recall that on the event $\mathcal{A}_{t+1}$ 
$$|\textup{PROJ}( \langle \eta^{(t+1)}_k, D^{(t+1)}_v \rangle ; \mathfrak{S}_t )| \leq K_{t+1}^{21} (\log n) \Delta_{t+1}  = o(1)\,.$$ 
In addition, by Lemma \ref{lemma_Gaussian_convar} and \eqref{equ_degree_variance_1} the variance of $\langle \eta^{(t+1)}_k, \Tilde{D}^{(t+1)}_v \rangle - \textup{GAUS}( \langle \eta^{(t+1)}_k, D^{(t+1)}_v \rangle ; \mathfrak{S}_t )$ is bounded by $1+K_{t+1} \Delta_{t+1} $ . Thus,
\begin{align*}
   & \mathbb{P}[ | \langle \eta^{(t+1)}_k, D^{(t+1)}_v \rangle | > \log n; \mathcal A_{t+1} ]\\ &\leq \mathbb{P}[ | \langle \eta^{(t+1)}_k, \Tilde{D}^{(t+1)}_v \rangle - \textup{GAUS}( \langle \eta^{(t+1)}_k, D^{(t+1)}_v \rangle ; \mathfrak{S}_t ) | > \log n -1 ]  \\
    &\leq \exp \{ - \frac{1}{4}(\log n -1 )^2  \}\,.
\end{align*}
By a union bound we derive that 
\begin{equation}\label{eq-mathcal-B-happens}
\mathbb P(\mathcal B_{t+1}^c \cap \mathcal A_{t+1 }) \leq K_{t+1} n e^{- (\log n-1)^2/4}\,.
\end{equation}

\subsubsection{Step 4: $\mathcal{E}_{t+2}$}
The goal of this subsection is to prove
\begin{equation}\label{eq-mathcal-E-t+1-holds}
\mathbb{P}[\mathcal{E}_{t+2}^{c}; \mathcal{A}_{t+1},\mathcal{B}_{t+1},\mathcal{C}_{t+1}] \leq \frac{1}{(\log n)^2}\,.
\end{equation}
Note that (iv.), (vi.) and (viii.) hold since we already proved $\mathcal{E}_0$. We thus focus on the other requirements. For $\zeta, \zeta' \in \cup_{k , v}\{\langle \sigma^{(t+1)}_k, D^{(t+1)}_v \rangle\} \bigcup \cup_{k , \mathsf v}\{ \langle \sigma^{(t+1)}_k, \mathsf{D}^{(t+1)}_{\mathsf v} \rangle\}$, we get from \eqref{eq-conditional-covariance-projection}  that
\begin{align}
    \mathbb{E} \big[ \big( \zeta|\mathfrak{S}_t \big) \big( \zeta'|\mathfrak{S}_t \big)   \big] 
    &=  \Hat{\mathbb{E}} \big[\zeta \zeta'] - \Hat{\mathbb{E}} \big[ \textup{GAUS}( \zeta ; \mathfrak{S}_t ) \textup{GAUS}( \zeta' ; \mathfrak{S}_t ) \big]  \nonumber\\
    &= \Hat{\mathbb{E}} [\zeta\zeta' ] + o(K_{t+1}^{22} \Delta_{t+1}^2)\,,
    \label{equ_conditional_cov}
\end{align}
where the last inequality holds on the event $\mathcal{C}_{t+1}$ and we have also applied the Cauchy-Schwartz inequality as well as the following fact:
\begin{align*}
    & \mathrm{Var}( \textup{GAUS}( \langle \sigma^{(t+1)}_k, D^{(t+1)}_v \rangle ; \mathfrak{S}_t )  ) \overset{\eqref{equ_def_sigma}}{=} \frac{12}{K_{t+1}} \mathrm{Var} \Big(  \sum_{j=1}^{\frac{1}{12} K_{t+1}} \beta_k^{(t+1)} (j) \textup{GAUS}( \langle \eta^{(t+1)}_j, D^{(t+1)}_v \rangle ; \mathfrak{S}_t )   \Big) \\
    & \leq \frac{12}{K_{t+1}} \Big ( \sum_{j=1}^{\frac{K_{t+1}}{12}} \mathrm{Var}( \textup{GAUS}( \langle \eta^{(t+1)}_j, D^{(t+1)}_v \rangle ; \mathfrak{S}_t )  ) \Big)^2 \overset{\mathcal{C}_{t+1}}{\leq} K_{t+1}^{21} \Delta_{t+1}^2 \,.
\end{align*}
In light of \eqref{equ_conditional_cov}, we consider the Gaussian density under perturbation, as follows. 
\begin{Lemma}{\label{lemma_estimation_perturbed_Gaussian_density}}
For  $d\geq 1$, let $\Sigma, \Sigma_\delta$ be positive definite $d*d$ matrices and let $\mu, \mu_\delta \in \mathbb R^d$ such that $\|\Sigma - \Sigma_\delta\|_\infty, \|\mu - \mu_\delta\|_\infty \leq \delta \leq 1$ and $\|\Sigma\|_\infty, \|\Sigma^{-1}\|_\infty, \|\Sigma^{-1}_{\delta}\|_{\infty} \leq M$. Let $p_{\mu, \Sigma}$ denote the density for a normal vector with mean $\mu$ and covariance matrix $\Sigma$.  Then there exists a constant $C= C(M, d)$ such that 
    \begin{align*}
     \exp \{ -C  ( \| x \| + \| \mu \| + 1)^2 \delta \}  \leq  \frac{ p_{\mu_{\delta} ,\Sigma_{\delta} } (x) }{ p_{\mu , \Sigma} (x) } \leq \exp \{ C  ( \| x \| + \| \mu \| + 1)^2 \delta \} \,.
    \end{align*}
\end{Lemma}
\begin{proof}
In what follows we let $C= C(M,d)$ be a constant depending only on $M$ and $d$ whose exact value may change from line to line. Note that for $x \in \mathbb R^d$
\begin{align*}
     p_{\mu , \Sigma} (x) = \frac{1}{\sqrt{(2\pi)^d} \mathrm{det}(\Sigma)} \exp \{ -\frac{1}{2}(x- \mu) \Sigma^{-1} (x - \mu)^{*}  \} \,,
\end{align*}
where $\mathrm{det}$ denotes the determinant of a matrix.
Obviously an analogous formula holds for $p_{\mu_\delta, \Sigma_\delta}$. Thus,
\begin{align*}
    \frac{ p_{\mu_{\delta},  \Sigma_{\delta}} (x) }{ p_{ \mu , \Sigma} (x) } = \sqrt{ \frac{ \mathrm{det}(\Sigma) }{ \mathrm{det}(\Sigma_{\delta}) } } \exp &\{ -\frac{1}{2} (x- \mu_{\delta}) (\Sigma_{\delta}^{-1} - \Sigma^{-1}) (x- \mu_{\delta})^{*} \\
   & - (\mu - \mu_{\delta}) \Sigma^{-1} ( x - \mu )^{*}  - \frac{1}{2} (\mu_{\delta} - \mu) \Sigma^{-1} (\mu_{\delta} - \mu)^{*}   \}\,.
\end{align*}
Note that 
\begin{align*}
    \sqrt{ \frac{ \mathrm{det}(\Sigma) }{ \mathrm{det}(\Sigma_{\delta}) } } = \sqrt{ \frac{1}{\mathrm{det}( \mathrm{I} + \Sigma^{-1}( \Sigma_{\delta} -\Sigma) )} } = \sqrt{ \frac{1}{ 1 + O( C(M,d) \delta) } }= 1 + O( C(M,d) \delta)\,.
\end{align*}
In addition, 
\begin{align*}
    & |(\mu_{\delta} - \mu ) \Sigma^{-1} ( x - \mu )^{*}| \leq \| \Sigma^{-1}  \|_{\mathrm{op}} \| \mu_{\delta}  - \mu \| \| x - \mu  \| \leq C(M,d) \delta  (\| x \| + \| \mu \| + 1 )\,,\\
    & |(\mu_{\delta} - \mu ) \Sigma^{-1} ( \mu_{\delta} - \mu )^{*}| \leq \| \Sigma^{-1}  \|_{\mathrm{op}} \| \mu_{\delta}  - \mu \|^2  \leq C(M,d) \delta^2 \,,\\
    & |(x- \mu_{\delta})( \Sigma_{\delta}^{-1}  - \Sigma^{-1} )(x- \mu_{\delta})^{*} | \leq \| \Sigma^{-1}_{\delta} \|_{\mathrm{op}}  \|\Sigma^{-1}\|_{\mathrm{op}}  \|  \Sigma_{\delta} - \Sigma  \|_{\mathrm{op}} (1+  \| x \|+ \| \mu \|)^2  \\
    &\leq C(M,d) \delta(1+  \| x \|+ \| \mu \|)^2 \,.
\end{align*}
Altogether, this completes the proof of the upper bound on the density ratio. The lower bound follows similarly.
\end{proof}

\begin{Corollary}{\label{corollary_estimation_perturbed_Gaussian_tail}}
We continue to make assumptions as in Lemma~\ref{lemma_estimation_perturbed_Gaussian_density}.  Let $Z, Z'$ be normal vectors with parameters $(\mu, \Sigma)$ and $(\mu_\delta, \Sigma_\delta)$ respectively. For any constant $C_1>0$, there exists $C_2 = C_2(C_1, d, M)> 0$ such that the following holds for all $\delta, L$:
$$\Big| \mathbb{P}[  | Z| \geq \gamma  ] - \mathbb{P} [ | Z' | \geq \gamma  ] \Big | \leq C_2 e^{- C_2^{-1}L^2} + ( e^{ C_2 (L+1)^2 \delta} -1 ) \mbox{ for } \|\gamma\|_\infty, \|\mu\|_\infty \leq C_1\,.$$
\end{Corollary}
\begin{proof}
We need to compare $ \int_{|x| \geq \gamma }  p_{\mu_{\delta} , \Sigma_{\delta}} (x) dx$ against   $\int_{|x| \geq \gamma } p_{\mu , \Sigma} (x) dx
$.
Note we may truncate the integration region to $\|x\|_\infty \leq L$, since the integral outside is bounded by $e^{- C_2^{-1}L^2}$.  For $\|x\|_\infty \leq L$, by Lemma~\ref{lemma_estimation_perturbed_Gaussian_density}
\begin{align*}
   | \frac{ p_{\mu_{\delta} ,  \Sigma_{\delta}} (x) }{ p_{\mu , \Sigma} (x) }  - 1| \leq C_2 \delta/4+ ( \exp \{ C_2( (\| x \|+ \| \mu \|)^2 \delta) /4  \} -1 ) =  (e^{ C_2( (L^2 +1 )\delta )}-1) \,.
\end{align*}
This completes the proof of the corollary by a standard computation.
\end{proof}
In our application of Corollary~\ref{corollary_estimation_perturbed_Gaussian_tail} later,  we will choose $L = \log n$ and $ \frac{1}{n} \leq \delta \leq n^{-0.01}$, so the approximation error will be bounded by $O( (\log n)^2 \delta)$.

Now we return to the proof of $\mathcal{E}_{t+2}$, i.e., to inductively verify admissible conditions. In what follows, we always assume that $\mathcal A_{t+1}, \mathcal B_{t+1}, \mathcal C_{t+1}$ hold.
Recalling Remark~\ref{remark_projection_form}, we have $\langle \sigma^{(t+1)}_k, D^{(t+1)}_v \rangle | \mathfrak{S}_t $ has the same distribution as
\begin{align*}
     \langle \sigma^{(t+1)}_k, \Tilde{D}^{(t+1)}_v \rangle - \textup{GAUS}(\langle \sigma^{(t+1)}_k, D^{(t+1)}_v \rangle; {\mathfrak{S}_t}) + \textup{PROJ}(\langle \sigma^{(t+1)}_k, D^{(t+1)}_v \rangle; \mathfrak{S}_t)\,.
\end{align*}
Note that
\begin{align*}
    \mathrm{Var}( \langle \sigma^{(t+1)}_k, \Tilde{D}^{(t+1)}_v \rangle ) = \mathrm{Var} \Big( \sqrt{\frac{12}{K_{t+1}}} \sum_{j=1}^{\frac{1}{12} K_{t+1}} \beta_k^{(t+1)} (j) \langle \eta^{(t+1)}_j, \Tilde{D}^{(t+1)}_v \rangle  \Big) \nonumber\\
    \overset{\eqref{equ_degree_correlation_2}}{=} \frac{12}{K_{t+1}} \sum_{j=1}^{\frac{1}{12} K_{t+1}} \mathrm{Var} ( \langle \eta^{(t+1)}_j, \Tilde{D}^{(t+1)}_v \rangle ) \overset{\eqref{equ_degree_variance_1}}{=} 1 + o(K_{t+1} \Delta_{t+1})\,.
\end{align*}
Recalling \eqref{eq-conditional-covariance-projection}, we see that the variance of $\langle \sigma^{(t+1)}_k, \Tilde{D}^{(t+1)}_v \rangle - \textup{GAUS} (\langle \sigma^{(t+1)}_k, D^{t+1}_v \rangle ; \mathfrak{S}_t )$ is 
\begin{equation}\label{equ_var_sigma_D}
1 + o(K_{t+1} \Delta_{t+1}) - o(K_{t+1}^{22} \Delta_{t+1}^2) = 1 + o(K_{t+1} \Delta_{t+1})\,,
\end{equation}
on the event $\mathcal C_{t+1}$. 
In addition, on the event $\mathcal A_{t+1}$, we have 
\begin{equation}\label{equ-mean-sigma-D}
|\textup{PROJ}(\langle \sigma^{(t+1)}_k, D^{(t+1)}_v \rangle; \mathfrak{S}_t)| \leq K_{t+1}^{21} (\log n) \Delta_{t+1}\,.
\end{equation}
Combined with Corollary \ref{corollary_estimation_perturbed_Gaussian_tail}, it yields that
\begin{align*}
    \Big| \mathbb{P} \Big[ v \in \Gamma^{(t+2)}_k \mid \mathfrak S_t \Big] - \alpha \Big| 
    &= o( K_{t+1}^{21} ( \log n )^3 \Delta_{t+1} ) \,.
\end{align*}
By \eqref{eq-conditional-covariance-projection}, we 
$$\mathrm{Cov}(\langle \sigma^{(t+1)}_k, D^{(t+1)}_u \rangle|{\mathfrak{S}_t}, \langle \sigma^{(t+1)}_k, D^{(t+1)}_v \rangle|{\mathfrak{S}_t}) = o(K_{t+1}^{22} \Delta_{t+1}^2)\,.$$ So
applying Corollary \ref{corollary_estimation_perturbed_Gaussian_tail} again gives that
\begin{align*}
     \Big| \mathbb{P} \Big[ u,v \in \Gamma^{(t+2)}_k \mid \mathfrak S_t\Big] - \mathbb{P} \Big[ u \in \Gamma^{(t+2)}_k \mid \mathfrak S_t\Big] \mathbb{P} \Big[ v \in \Gamma^{(t+2)}_k \mid \mathfrak S_t \Big] \Big|  
    = o( K_{t+1}^{22} ( \log n )^2 \Delta^2_{t+1})\,.
\end{align*}
So by Chebyshev inequality we have
\begin{align*}
    \mathbb{P} \Big[  \Big| \frac{|\Gamma^{(t+2)}_k |}{n} - \alpha\Big| > \Delta_{t+2}  \mid \mathfrak S_t \Big] \leq \frac{1}{K_{t+2}^2 (\log n)^4}\,.
\end{align*}
By a union bound this verifies (i.). Similarly we can verify (ii.). 

The concentration results concerning $\frac{|\Gamma^{(t+2)}_i \cap \Gamma^{(t+2)}_j|}{n}, \frac{|\Pi^{(t+2)}_i \cap \Pi^{(t+2)}_j|}{n}, \frac{|\Gamma^{(t+2)}_i \cap \Pi^{(t+2)}_j|}{n}$ (which correspond to  (iii.), (v.) and (vii.)  respectively) can be proved similarly. 

For $\frac{|\Gamma^{(t+2)}_k \cap \Gamma^{(t+2)}_l|}{n}$, it suffices to estimate \begin{align}
    \frac{1}{n} \sum_{v \in V} \left(  \mathbb{P} \left[ \left | \langle \sigma^{(t+1)}_k, D^{(t+1)}_v \rangle | \mathfrak{S}_{t} \right| , \left| \langle \sigma^{(t+1)}_l, D^{(t+1)}_v \rangle | \mathfrak{S}_{t} \right| \geq 10\right] - \phi( \frac{12}{K_{t+1}} \langle \beta^{(t+1)}_k, \beta^{(t+1)}_l \rangle )  \right) \label{eq-Gamma-Gamma-first-moment}
    \end{align}
and
\begin{equation}\label{eq-Gamma-Gamma-second-moment}
\begin{aligned}
    \frac{1}{n^2} \sum_{u \not = v} \mathrm{ Cov } \Big(  \mathbf{I}_{ \left\{  \left| \langle \sigma^{(t+1)}_k , D^{(t+1)}_v \rangle | \mathfrak{S}_{t} \right| , \left| \langle \sigma^{(t+1)}_l , D^{(t+1)}_v \rangle | \mathfrak{S}_{t} \right| \geq 10 \right \} } ,  \\
    \mathbf{I}_{ \left\{   \left| \langle \sigma^{(t+1)}_k , D^{(t+1)}_u \rangle | \mathfrak{S}_{t} \right| , \left| \langle \sigma^{(t+1)}_l , D^{(t+1)}_u \rangle | \mathfrak{S}_{t} \right| \geq 10    \right \} }   \Big)\,.
\end{aligned}
\end{equation}
Note that
\begin{align*}
    &\mathrm{Cov}( \langle \sigma^{(t+1)}_k, \Tilde{D}^{(t+1)}_v \rangle, \langle \sigma^{(t+1)}_l, \Tilde{D}^{(t+1)}_v \rangle ) \\
    =& \frac{12}{K_{t+1}} \sum_{i=1}^{\frac{1}{12}K_{t+1}} \sum_{j=1}^{\frac{1}{12}K_{t+1}}  \beta^{(t+1)}_k (i) \beta^{(t+1)}_l (j) \mathrm{Cov} ( \langle \eta^{(t+1)}_i, \Tilde{D}^{(t+1)}_v \rangle, \langle \eta^{(t+1)}_j, \Tilde{D}^{(t+1)}_v \rangle ) \\
    \overset{\eqref{equ_degree_correlation_2}}{=} & \frac{12}{K_{t+1}} \sum_{i=1}^{\frac{1}{12}K_{t+1}} \beta^{(t+1)}_k(i) \beta^{(t+1)}_l(i) \mathrm{Var}(\langle \eta^{(t+1)}_i, \Tilde{D}^{(t+1)}_v \rangle) \\ \overset{\eqref{equ_degree_variance_1}}{=} & \frac{12}{K_{t+1}} \langle \beta^{(t+1)}_k, \beta^{(t+1)}_l \rangle + o(K_{t+1} \Delta_{t+1})\,.
\end{align*}
Recalling \eqref{equ_var_sigma_D} and \eqref{equ-mean-sigma-D}, we see that on the event 
$\mathcal{A}_{t+1} \cap \mathcal C_{t+1}$ we have (here $N(\mu, \Sigma)$ denotes for a normal vector with mean $\mu$ and covariance matrix $\Sigma$)
\begin{align*}
    &\begin{pmatrix}
    \langle \sigma^{(t+1)}_k, D^{(t+1)}_v \rangle | \mathfrak{S}_{t} \\
    \langle \sigma^{(t+1)}_l, D^{(t+1)}_v \rangle | \mathfrak{S}_{t}
    \end{pmatrix}
    \overset{d}{=} 
    N \left( 
    \begin{pmatrix}
    \mu_k \\ \mu_l
    \end{pmatrix},
    \begin{pmatrix}
    1 + \sigma_k & \langle \beta^{(t+1)}_k, \beta^{(t+1)}_l \rangle + \rho_{k,l} \\
    \langle \beta^{(t+1)}_k, \beta^{(t+1)}_l \rangle + \rho_{k,l} & 1 + \sigma_l
    \end{pmatrix}
    \right )\,,
\end{align*}
where $|\mu_k|,|\mu_l| \leq K^{21}_{t+1} (\log n) \Delta_{t+1}$ and $|\sigma_k|,|\sigma_l|,|\rho_{k,l}| \leq K_{t+1}^{20} \Delta_{t+1}$. 
So by Corollary \ref{corollary_estimation_perturbed_Gaussian_tail}, we have $|\eqref{eq-Gamma-Gamma-first-moment}| = o(K_{t+1}^{21} (\log n)^3 \Delta_{t+1})$.  By \eqref{eq-mathcal-E-t+1-holds}, for $u\neq v$ the pairwise covariance between $\langle \sigma^{(t+1)}_i, D^{(t+1)}_v \rangle | \mathfrak{S}_{t}, \langle \sigma^{(t+1)}_i, D^{(t+1)}_v \rangle | \mathfrak{S}_{t}$ and $\langle \sigma^{(t+1)}_i, D^{(t+1)}_u \rangle | \mathfrak{S}_{t}, \langle \sigma^{(t+1)}_i, D^{(t+1)}_u \rangle | \mathfrak{S}_{t}$ is $o(K_{t+1}^{22} \Delta_{t+1}^2)$. Thus, by Corollary~\ref{corollary_estimation_perturbed_Gaussian_tail} again we have $\eqref{eq-Gamma-Gamma-second-moment} = o(K_{t+1}^{22} (\log n)^2 \Delta_{t+1}^2)$. Therefore, applying Chebyshev's inequality and a union bound yields the desired concentration for $\frac{ |\Gamma^{t+2}_i \cap \Gamma^{(t+2)}_j| }{n}$. Furthermore, we can control $\frac{|\Pi^{(t+2)}_i \cap \Pi^{(t+2)}_j|}{n}$ in the same way.

For $\frac{|\pi(\Gamma^{(t+2)}_k) \cap \Pi^{(t+2)}_l|}{n}$, it suffices to estimate \begin{align}\label{eq-Gamma-Pi-first-moment}
    \frac{1}{n} \sum_{v \in V} \left(  \mathbb{P} \left[ \left | \langle \sigma^{(t+1)}_k, D^{(t+1)}_v \rangle | \mathfrak{S}_{t} \right| , \left| \langle \sigma^{(t+1)}_l, \mathsf{D}^{(t+1)}_{\pi(v)} \rangle | \mathfrak{S}_{t} \right| \geq 10 \right] - \phi( \frac{12}{K_{t+1}} \langle \Hat{\beta}^{(t+1)}_k, \Hat{\beta}^{(t+1)}_l \rangle )  \right)
\end{align}
and 
\begin{equation}\label{eq-Gamma-Pi-second-moment}
\begin{aligned}
    \frac{1}{n^2} \sum_{u \not = v} \mathrm{ Cov } \Big( \mathbf{I}_{ \left\{   \left| \langle \sigma^{(t+1)}_k , D^{(t+1)}_v \rangle | \mathfrak{S}_{t} \right| , \left| \langle \sigma^{(t+1)}_l , \mathsf{D}^{(t+1)}_{\pi(v)} \rangle | \mathfrak{S}_{t} \right| \geq 10     \right\}  }  ,  \\
    \mathbf{I}_{ \left\{   \left| \langle \sigma^{(t+1)}_k , D^{(t+1)}_u \rangle | \mathfrak{S}_{t} \right| , \left| \langle \sigma^{(t+1)}_l , \mathsf{D}^{(t+1)}_{\pi(u)} \rangle | \mathfrak{S}_{t} \right| \geq 10     \right\}  }  \Big)\,.
\end{aligned}
\end{equation}
Note that
\begin{align*}
    &\mathrm{Cov}( \langle \sigma^{(t+1)}_k, \Tilde{D}^{(t+1)}_v \rangle, \langle \sigma^{(t+1)}_l, \Tilde{\mathsf{D}}^{(t+1)}_{\pi(v)} \rangle ) \\
    =& \frac{12}{K_{t+1}} \sum_{i=1}^{\frac{1}{12}K_{t+1}} \sum_{j=1}^{\frac{1}{12}K_{t+1}}  \beta^{(t+1)}_k (i) \beta^{(t+1)}_l (j) \mathrm{Cov} ( \langle \eta^{(t+1)}_i, \Tilde{D}^{(t+1)}_v \rangle, \langle \eta^{(t+1)}_j, \Tilde{\mathsf{D}}^{(t+1)}_{\pi(v)} \rangle ) \\
    = & \frac{12}{K_{t+1}} \sum_{i=1}^{\frac{1}{12}K_{t+1}} \sum_{j=1}^{\frac{1}{12}K_{t+1}} \beta^{(t+1)}_k(i) \beta^{(t+1)}_l(j) \eta^{(t)}_i \mathrm{P}^{(t+1,t+1)}_{\Gamma,\Pi} \left( \eta^{(t)}_j \right)^{*}  \\
    \overset{\eqref{equ_concentration_P_Gamma_Pi} \eqref{equ_vector_orthogonal}}{=} & \frac{12}{K_{t+1}} \sum_{i=1}^{\frac{1}{12}K_{t+1}}  \beta^{(t+1)}_k(i) \beta^{(t+1)}_l(i) \eta^{(t)}_i \Psi^{(t+1)} \left( \eta^{(t)}_i \right)^{*} + o(K_{t+1}^2 \Delta_{t+1})  \\
    = & \frac{12}{K_{t+1}} \langle \Hat{\beta}^{(t+1)}_k, \Hat{\beta}^{(t+1)}_l \rangle + o(K_{t+1}^2 \Delta_{t+1})  \,.
\end{align*}   
Recalling \eqref{equ_var_sigma_D} and \eqref{equ-mean-sigma-D} (as well as their analogues for $\tilde{\mathsf D}$), we see that on the event 
$\mathcal{A}_{t+1} \cap \mathcal C_{t+1}$
\begin{align*}
    &\begin{pmatrix}
    \langle \sigma^{(t+1)}_k, D^{(t+1)}_v \rangle | \mathfrak{S}_{t} \\
    \langle \sigma^{(t+1)}_l, \mathsf D^{(t+1)}_{\pi(v)} \rangle | \mathfrak{S}_{t}
    \end{pmatrix}
    \overset{d}{=} 
    N \left( 
    \begin{pmatrix}
    \mu_k \\ \mu_l
    \end{pmatrix},
    \begin{pmatrix}
    1 + \sigma_k & \langle \Hat{\beta}^{(t+1)}_k, \Hat{\beta}^{(t+1)}_l \rangle + \rho_{k,l}  \\
    \langle \Hat{\beta}^{(t+1)}_k, \Hat{\beta}^{(t+1)}_l \rangle + \rho_{k,l} & 1 + \sigma_l
    \end{pmatrix}
    \right )\,,
\end{align*}   
where $|\mu_k|,|\mu_l| \leq K_{t+1}^{21} (\log n) \Delta_{t+1}$ and $|\sigma_k|, |\sigma_l|, |\rho_{k,l}| \leq K^{20}_{t+1} \Delta_{t+1}$. So by Corollary \ref{corollary_estimation_perturbed_Gaussian_tail} we have $\eqref{eq-Gamma-Pi-first-moment} =  o(K_{t+1}^{21} (\log n)^3 \Delta_{t+1})$. By \eqref{eq-mathcal-E-t+1-holds}, for $u\neq v$ the pairwise covariance between $\langle \sigma^{(t+1)}_i, D^{(t+1)}_v \rangle | \mathfrak{S}_{t}, \langle \sigma^{(t+1)}_i, \mathsf D^{(t+1)}_{\pi(v)} \rangle | \mathfrak{S}_{t}$ and $\langle \sigma^{(t+1)}_i, D^{(t+1)}_u \rangle | \mathfrak{S}_{t}, \langle \sigma^{(t+1)}_i, \mathsf{D}^{(t+1)}_{\pi(u)} \rangle | \mathfrak{S}_{t}$ is $o(K_{t+1}^{22} \Delta_{t+1}^2)$. Applying Corollary \ref{corollary_estimation_perturbed_Gaussian_tail} again we have $\eqref{eq-Gamma-Pi-second-moment} = o(K_{t+1}^{22} (\log n)^2 \Delta_{t+1}^2)$. What remains is a standard application of Chebyshev's inequality and a union bound as above.

Furthermore, we control the concentration of $\frac{ | \Gamma^{(t+2)}_k \cap \Pi^{(s)}_l |}{n}, \frac{ | \Gamma^{(t+2)}_k \cap \Gamma^{(s)}_l |}{n}, \frac{ | \Pi^{(t+2)}_k \cap \Pi^{(s)}_l |}{n} $ (which correspond to (ix.), (x.) and (xi.) respectively). Note that under $\mathfrak{S}_{t}$, $\Pi^{(s)}_l$'s are fixed subsets for $s\leq t$. In addition, on the event $\mathcal{E}_{t+1}$ we have $\Big| \frac{| \Pi^{(s)}_l |}{n} - \alpha \Big| <\Delta_s$. So, 
\begin{align*}
    \frac{ | \Gamma^{(t+2)}_k \cap \Pi^{(s)}_l |}{n} - \alpha^2 = \alpha \Big( \frac{1}{\alpha n}  \sum_{ u \in \Pi^{(s)}_l } \Big( \mathbf{I}_{ u \in \Gamma^{(t+2)}_k } - \alpha \Big) \Big) + \alpha \Big( \frac{|\Pi^{(s)}_l|}{n} - \alpha \Big)\,.
 \end{align*}
 Since $|\alpha \Big( \frac{|\Pi^{(s)}_l|}{n} - \alpha \Big)| \leq \Delta_s$ on the event $\mathcal{E}_{t+1}$, the above can be handled similarly to that for $\frac{|\Gamma^{(t+2)}_k |}{n} - \alpha$. The same applies to the other two items here. We omit further details since the modifications are minor.
 
 Putting all above together, we finally complete the proof of \eqref{eq-mathcal-E-t+1-holds}.

\subsubsection{Conclusion}
By putting together \eqref{eq-mathcal-E-0-bound}, \eqref{prob-mathcal-A-t+1}, \eqref{eq-mathcal-C-containment}, 
\eqref{eq-mathcal-B-happens} and \eqref{eq-mathcal-E-t+1-holds},  we have proved {\bf Step 1}--{\bf Step 4} listed at the beginning of this subsection. In addition, since $t^*\leq \log \log n$, our quantitative bounds imply that all these hold simultaneously for $t = 0, \ldots, t^*$ except with probability $O(1/\log n)$. We first apply \eqref{eq-mathcal-E-t+1-holds} with $t=-1$ to derive that $\mathcal E_1$ holds with high probability and then by the inductive logic explained at the beginning of this subsection, we complete the proof of the proposition. We also point out that in addition we have shown that $\mathcal A_{t^*}, \mathcal B_{t^*}, \mathcal C_{t^*}$ hold with probability $1-o(1)$, which will be used in Section~\ref{sec-proof-main-prop}.

\subsection{Proof of Proposition~\ref{main-prop}}\label{sec-proof-main-prop}

It remains to show that on $\mathcal E^{\diamond} = \mathcal A_{t^*} \cap \mathcal B_{t^*} \cap \mathcal C_{t^*} \cap \mathcal{E}_{t^*}$, our matching algorithm succeeds with probability $1-o(1)$. At this point, the approach is fairly straightforward: since we have accumulated enough signals over the iterative steps thanks to the success of $ \mathcal{E}^\diamond$, we just need to use this to prove large deviation bounds for statistics employed in the finishing step of our algorithm.

For convenience we will omit the index $t^*$. So we will denote $\varepsilon_{t^*},K_{t^*}$ as $\varepsilon,K$.  Recalling Remark~\ref{remark_projection_form}, we have 
\begin{align*}
    \langle \eta_k, D_v \rangle |{ \mathfrak{S} } \overset{d}{=} \langle \eta_k , \Tilde{D}_v \rangle - \textup{GAUS}( \langle \eta_k , D_v \rangle ; \mathfrak{S}_{t^*-1} ) + \textup{PROJ}( \langle \eta_k , D_v \rangle ; \mathfrak{S}_{t^*-1} )\,.
\end{align*}
Note $\mathrm{Var}(\textup{GAUS}( \langle \eta_k , D_v \rangle ; \mathfrak{S}_{t^*-1} )) \leq K^{20} \Delta^2$ and $|\textup{PROJ}( \langle \eta_k , D_v \rangle ; \mathfrak{S}_{t^*-1})| \leq K^{21} (\log n) \Delta$ on  $\mathcal E^\diamond$. So on  $\mathcal E^\diamond$ we have $|\textup{GAUS}( \langle \eta_k , D_v \rangle ; \mathfrak{S}_{t^*-1} )|, |\textup{PROJ}( \langle \eta_k , D_v \rangle ; \mathfrak{S}_{t^*-1} ) |\leq n^{-0.05}$ with probability at least $1 - K n e^{-n^{0.01}}$. On this event,  for $\pi(v)=\mathsf{v}$,
\begin{align*}
    \sum_{k=1}^{\frac{K}{12}} \langle \eta_k, D_v \rangle \langle \eta_k , \mathsf{D}_{\pi(v)} \rangle |_{\mathfrak{S}_{t^*-1}}
    \overset{d}{=}  \sum_{k=1}^{\frac{K}{12}} \langle \eta_k, \Tilde{D}_v \rangle  \langle \eta_k , \Tilde{\mathsf{D}}_{\pi(v)} \rangle + o(n^{-0.01})
\end{align*}
is a quadratic form of Gaussian variables with expectation larger than $ \frac{K \varepsilon}{20} $ and variance bounded by $K$; for $\pi(v) \not = \mathsf{v}$,
\begin{align*}
    \sum_{k=1}^{\frac{K}{12}} \langle \eta_k, D_v \rangle \langle \eta_k , \mathsf{D}_{\mathsf{v}} \rangle |_{\mathfrak{S}_{t^*-1}}
    \overset{d}{=} \sum_{k=1}^{\frac{K}{12}} \langle \eta_k, \Tilde{D}_v \rangle  \langle \eta_k , \Tilde{\mathsf{D}}_{\mathsf{v}} \rangle + o(n^{-0.01})
\end{align*}
is also a quadratic form of Gaussian variables, with expectation $O(\frac{1}{n})$ and variance bounded by $K$. Thus, applying Corollary \ref{corollary_Hanson_Wright_Gaussian}, we get that for $\mathsf v\neq \pi(v)$
\begin{align*}
    \mathbb{P}& \Big[ \sum_{ k=1 }^{ \frac{K}{12} } \langle \eta_k, D_v \rangle \langle \eta_k, \mathsf{D}_{\mathsf{v}} \rangle > \frac{K \varepsilon}{100} ; \mathcal E^\diamond\Big] \leq \mathbb{P} \Big[  \sum_{ k=1 }^{ \frac{K}{12} } \langle \eta_k, \Tilde{D}_v \rangle \langle \eta_k, \Tilde{\mathsf{D}}_{\mathsf{v}} \rangle > \frac{K \varepsilon}{200}  \Big]  + K n e^{-n^{0.01}}\\
    & \leq  \exp \Big \{ - \Big(  \frac{ (\frac{K \varepsilon}{200})^2 }{ \frac{K}{12}  }  \Big)^{\frac{1}{2}}  \Big \} + K n e^{-n^{0.01}} \leq \exp \Big\{ - \Big( \frac{K \varepsilon^2}{ 40000 } \Big)^{\frac{1}{2}}  \Big \} + K n e^{-n^{0.01}} \overset{ (\ref{equ_estimation_K_t_Varepsilon_t}) }{<} \frac{1}{n^5}\,,
\end{align*}
and (similarly) that
\begin{equation*}
    \mathbb{P} \Big[ \sum_{ k=1 }^{ \frac{K}{12} } \langle \eta_k, D_v \rangle \langle \eta_k, \mathsf{D}_{\pi(v)} \rangle \leq \frac{K \varepsilon}{100} ; \mathcal E^\diamond \Big] \leq \frac{1}{n^5}\,.
\end{equation*}
Combining the above two results, we get that
\begin{equation*}
    \mathbb{P}[ \Hat{\pi} (v) \not = \pi(v) ; \mathcal E^\diamond] \leq \frac{1}{n^4}\,.
\end{equation*}
At this point, a union bound completes the proof of the proposition.


\small
\begin{thebibliography}{10}

\bibitem{Alon09}
N.~Alon.
\newblock Perturbed identity matrices have high rank: proof and applications.
\newblock {\em Combin. Probab. Comput.}, 18(1-2):3--15, 2009.

\bibitem{BCL19}
B.~Barak, C.-N. Chou, Z.~Lei, T.~Schramm, and Y.~Sheng.
\newblock (nearly) efficient algorithms for the graph matching problem on
  correlated random graphs.
\newblock In {\em Advances in Neural Information Processing Systems},
  volume~32. Curran Associates, Inc., 2019.

\bibitem{BM13}
F.~Barthe and E.~Milman.
\newblock Transference principles for log-{S}obolev and spectral-gap with
  applications to conservative spin systems.
\newblock {\em Comm. Math. Phys.}, 323(2):575--625, 2013.

\bibitem{BM10}
M.~Bayati and A.~Montanari.
\newblock The dynamics of message passing on dense graphs, with applications to
  compressed sensing.
\newblock {\em IEEE Trans. Inform. Theory}, 57(2):764--785, 2011.

\bibitem{BBM05}
A.~Berg, T.~Berg, and J.~Malik.
\newblock Shape matching and object recognition using low distortion
  correspondences.
\newblock In {\em 2005 IEEE Computer Society Conference on Computer Vision and
  Pattern Recognition (CVPR'05)}, volume~1, pages 26--33 vol. 1, 2005.

\bibitem{Bourgain99}
J.~Bourgain.
\newblock Random points in isotropic convex sets.
\newblock In {\em Convex geometric analysis ({B}erkeley, {CA}, 1996)},
  volume~34 of {\em Math. Sci. Res. Inst. Publ.}, pages 53--58. Cambridge Univ.
  Press, Cambridge, 1999.

\bibitem{BSH19}
M.~Bozorg, S.~Salehkaleybar, and M.~Hashemi.
\newblock Seedless graph matching via tail of degree distribution for
  correlated erdos-renyi graphs.
\newblock Preprint, arXiv:1907.06334.

\bibitem{CJMNZ22+}
S.~Chen, S.~Jiang, Z.~Ma, G.~P. Nolan, and B.~Zhu.
\newblock One-way matching of datasets with low rank signals.
\newblock Preprint, arXiv:2204.13858.

\bibitem{CSS07}
T.~Cour, P.~Srinivasan, and J.~Shi.
\newblock Balanced graph matching.
\newblock In B.~Sch\"{o}lkopf, J.~Platt, and T.~Hoffman, editors, {\em Advances
  in Neural Information Processing Systems}, volume~19. MIT Press, 2006.

\bibitem{CK17}
D.~Cullina and N.~Kiyavash.
\newblock Exact alignment recovery for correlated {E}rdos-{R}\'enyi graphs.
\newblock Preprint, arXiv:1711.06783.

\bibitem{CK16}
D.~Cullina and N.~Kiyavash.
\newblock Improved achievability and converse bounds for erdos-renyi graph
  matching.
\newblock In {\em Proceedings of the 2016 ACM SIGMETRICS International
  Conference on Measurement and Modeling of Computer Science}, SIGMETRICS '16,
  pages 63--72, New York, NY, USA, 2016. Association for Computing Machinery.

\bibitem{CKMP19}
D.~Cullina, N.~Kiyavash, P.~Mittal, and H.~V. Poor.
\newblock Partial recovery of erdos-r\'enyi graph alignment via $k$-core
  alignment.
\newblock SIGMETRICS '20, pages 99--100, New York, NY, USA, 2020. Association
  for Computing Machinery.

\bibitem{DCKG19}
O.~E. Dai, D.~Cullina, N.~Kiyavash, and M.~Grossglauser.
\newblock Analysis of a canonical labeling algorithm for the alignment of
  correlated erdos-r\'{e}nyi graphs.
\newblock {\em Proc. ACM Meas. Anal. Comput. Syst.}, 3(2), jun 2019.

\bibitem{DKN10}
I.~Diakonikolas, D.~M. Kane, and J.~Nelson.
\newblock Bounded independence fools degree-2 threshold functions.
\newblock In {\em 2010 {IEEE} 51st {A}nnual {S}ymposium on {F}oundations of
  {C}omputer {S}cience---{FOCS} 2010}, pages 11--20. IEEE Computer Soc., Los
  Alamitos, CA, 2010.

\bibitem{DD22+}
J.~Ding and H.~Du.
\newblock Detection threshold for correlated erdos-renyi graphs via densest
  subgraph.
\newblock Preprint, arXiv:2203.14573.

\bibitem{DD22+b}
J.~Ding and H.~Du.
\newblock Matching recovery threshold for correlated random graphs.
\newblock Preprint, arXiv:2205.14650.

\bibitem{DDG22+}
J.~Ding, H.~Du, and S.~Gong.
\newblock A polynomial-time approximation scheme for the maximal overlap of two
  independent {E}rdos-{R}\'{e}nyi graphs.
\newblock Preprint, arXiv:2210.07823.

\bibitem{DMWX21}
J.~Ding, Z.~Ma, Y.~Wu, and J.~Xu.
\newblock Efficient random graph matching via degree profiles.
\newblock {\em Probab. Theory Related Fields}, 179(1-2):29--115, 2021.

\bibitem{DP09}
D.~P. Dubhashi and A.~Panconesi.
\newblock {\em Concentration of measure for the analysis of randomized
  algorithms}.
\newblock Cambridge University Press, Cambridge, 2009.

\bibitem{FMWX22a}
Z.~Fan, C.~Mao, Y.~Wu, and J.~Xu.
\newblock Spectral graph matching and regularized quadratic relaxations:
  Algorithm and theory.
\newblock {\em Foundations of Computational Mathematics}, 2022.

\bibitem{FMWX22b}
Z.~Fan, C.~Mao, Y.~Wu, and J.~Xu.
\newblock Spectral graph matching and regularized quadratic relaxations {II}:
  Erdos-r\'enyi graphs and universality.
\newblock {\em Foundations of Computational Mathematics}, 2022.

\bibitem{FQRM+16}
S.~Feizi, G.~Quon, M.~Medard, M.~Kellis, and A.~Jadbabaie.
\newblock Spectral alignment of networks.
\newblock Preprint, arXiv:1602.04181.

\bibitem{FR13}
S.~Foucart and H.~Rauhut.
\newblock {\em A mathematical introduction to compressive sensing}.
\newblock Applied and Numerical Harmonic Analysis. Birkh\"{a}user/Springer, New
  York, 2013.

\bibitem{Garmarnik21}
D.~Gamarnik.
\newblock The overlap gap property: A topological barrier to optimizing over
  random structures.
\newblock {\em Proceedings of the National Academy of Sciences},
  118(41):e2108492118, 2021.

\bibitem{GM20}
L.~Ganassali and L.~Massouli\'e.
\newblock From tree matching to sparse graph alignment.
\newblock In J.~Abernethy and S.~Agarwal, editors, {\em Proceedings of Thirty
  Third Conference on Learning Theory}, volume 125 of {\em Proceedings of
  Machine Learning Research}, pages 1633--1665. PMLR, 09--12 Jul 2020.

\bibitem{GML21}
L.~Ganassali, L.~Massouli\'e, and M.~Lelarge.
\newblock Impossibility of partial recovery in the graph alignment problem.
\newblock In M.~Belkin and S.~Kpotufe, editors, {\em Proceedings of Thirty
  Fourth Conference on Learning Theory}, volume 134 of {\em Proceedings of
  Machine Learning Research}, pages 2080--2102. PMLR, 15--19 Aug 2021.

\bibitem{GML22}
L.~Ganassali, L.~Massouli\'{e}, and M.~Lelarge.
\newblock Correlation detection in trees for planted graph alignment.
\newblock In {\em 13th {I}nnovations in {T}heoretical {C}omputer {S}cience
  {C}onference}, volume 215 of {\em LIPIcs. Leibniz Int. Proc. Inform.}, pages
  Art. No. 74, 8. Schloss Dagstuhl. Leibniz-Zent. Inform., Wadern, 2022.

\bibitem{GMS22+}
L.~Ganassali, L.~Massouli\'e, and G.~Semerjian.
\newblock Statistical limits of correlation detection in trees.
\newblock {\em arXiv:2209.13723}.

\bibitem{HNM05}
A.~Haghighi, A.~Ng, and C.~Manning.
\newblock Robust textual inference via graph matching.
\newblock In {\em Proceedings of Human Language Technology Conference and
  Conference on Empirical Methods in Natural Language Processing}, pages
  387--394, Vancouver, British Columbia, Canada, Oct 2005.

\bibitem{HM20}
G.~Hall and L.~Massouli\'e.
\newblock Partial recovery in the graph alignment problem.
\newblock Preprint, arXiv:2007.00533.

\bibitem{HW71}
D.~L. Hanson and F.~T. Wright.
\newblock A bound on tail probabilities for quadratic forms in independent
  random variables.
\newblock {\em Ann. Math. Statist.}, 42:1079--1083, 1971.

\bibitem{KHG15}
E.~Kazemi, S.~H. Hassani, and M.~Grossglauser.
\newblock Growing a graph matching from a handful of seeds.
\newblock {\em Proc. VLDB Endow.}, 8(10):1010--1021, jun 2015.

\bibitem{Latala06}
R.~Lata\l a.
\newblock Estimates of moments and tails of {G}aussian chaoses.
\newblock {\em Ann. Probab.}, 34(6):2315--2331, 2006.

\bibitem{LFP14}
V.~Lyzinski, D.~E. Fishkind, and C.~E. Priebe.
\newblock Seeded graph matching for correlated {E}rdos-{R}\'{e}nyi graphs.
\newblock {\em J. Mach. Learn. Res.}, 15:3513--3540, 2014.

\bibitem{MRT21+}
C.~Mao, M.~Rudelson, and K.~Tikhomirov.
\newblock Exact matching of random graphs with constant correlation.
\newblock Preprint, arXiv:2110.05000.

\bibitem{MWXY22+}
C.~Mao, Y.~Wu, J.~Xu, and S.~H. Yu.
\newblock Random graph matching at {O}tter's threshold via counting
  chandeliers.
\newblock Preprint, arXiv:2209.12313.

\bibitem{MWXY21+}
C.~Mao, Y.~Wu, J.~Xu, and S.~H. Yu.
\newblock Testing network correlation efficiently via counting trees.
\newblock Preprint, arXiv:2110.11816.

\bibitem{MX20}
E.~Mossel and J.~Xu.
\newblock Seeded graph matching via large neighborhood statistics.
\newblock {\em Random Structures Algorithms}, 57(3):570--611, 2020.

\bibitem{NS08}
A.~Narayanan and V.~Shmatikov.
\newblock Robust de-anonymization of large sparse datasets.
\newblock In {\em 2008 IEEE Symposium on Security and Privacy (sp 2008)}, pages
  111--125, 2008.

\bibitem{NS09}
A.~Narayanan and V.~Shmatikov.
\newblock De-anonymizing social networks.
\newblock In {\em 2009 30th IEEE Symposium on Security and Privacy}, pages
  173--187, 2009.

\bibitem{PG11}
P.~Pedarsani and M.~Grossglauser.
\newblock On the privacy of anonymized networks.
\newblock In {\em Proceedings of the 17th ACM SIGKDD International Conference
  on Knowledge Discovery and Data Mining}, KDD '11, pages 1235--1243, New York,
  NY, USA, 2011. Association for Computing Machinery.

\bibitem{PSSZ22}
G.~Piccioli, G.~Semerjian, G.~Sicuro, and L.~Zdeborov\'{a}.
\newblock Aligning random graphs with a sub-tree similarity message-passing
  algorithm.
\newblock {\em J. Stat. Mech. Theory Exp.}, (6):Paper No. 063401, 44, 2022.

\bibitem{RS20+}
M.~Z. Racz and A.~Sridhar.
\newblock Correlated randomly growing graphs.
\newblock to appear in Ann. Appl. Probab.

\bibitem{RS21}
M.~Z. Racz and A.~Sridhar.
\newblock Correlated stochastic block models: Exact graph matching with
  applications to recovering communities.
\newblock In {\em Advances in Neural Information Processing Systems}, 2021.

\bibitem{RV14}
M.~Rudelson and R.~Vershynin.
\newblock Hanson-{W}right inequality and sub-{G}aussian concentration.
\newblock {\em Electron. Commun. Probab.}, 18:no. 82, 9, 2013.

\bibitem{SW22}
T.~Schramm and A.~S. Wein.
\newblock Computational barriers to estimation from low-degree polynomials.
\newblock {\em Ann. Statist.}, 50(3):1833--1858, 2022.

\bibitem{SGE17}
F.~Shirani, S.~Garg, and E.~Erkip.
\newblock Seeded graph matching: Efficient algorithms and theoretical
  guarantees.
\newblock In {\em 2017 51st Asilomar Conference on Signals, Systems, and
  Computers}, pages 253--257, 2017.

\bibitem{SXB08}
R.~Singh, J.~Xu, and B.~Berger.
\newblock Global alignment of multiple protein interaction networks with
  application to functional orthology detection.
\newblock {\em Proceedings of the National Academy of Sciences of the United
  States of America}, 105:12763--8, 10 2008.

\bibitem{VCP15}
J.~T. Vogelstein, J.~M. Conroy, V.~Lyzinski, L.~J. Podrazik, S.~G. Kratzer,
  E.~T. Harley, D.~E. Fishkind, R.~J. Vogelstein, and C.~E. Priebe.
\newblock Fast approximate quadratic programming for graph matching.
\newblock {\em PLOS ONE}, 10(4):1--17, 04 2015.

\bibitem{WWXY22+}
H.~Wang, Y.~Wu, J.~Xu, and I.~Yolou.
\newblock Random graph matching in geometric models: the case of complete
  graphs.
\newblock Preprint, arXiv:2202.10662.

\bibitem{Wright73}
F.~T. Wright.
\newblock A bound on tail probabilities for quadratic forms in independent
  random variables whose distributions are not necessarily symmetric.
\newblock {\em Ann. Probability}, 1(6):1068--1070, 1973.

\bibitem{WXY21+}
Y.~Wu, J.~Xu, and S.~H. Yu.
\newblock Settling the sharp reconstruction thresholds of random graph
  matching.
\newblock Preprint, arXiv:2102.00082.

\bibitem{WXY20+}
Y.~Wu, J.~Xu, and S.~H. Yu.
\newblock Testing correlation of unlabeled random graphs.
\newblock Preprint, arXiv:2008.10097.

\bibitem{YG13}
L.~Yartseva and M.~Grossglauser.
\newblock On the performance of percolation graph matching.
\newblock In {\em Proceedings of the First ACM Conference on Online Social
  Networks}, COSN '13, pages 119--130, New York, NY, USA, 2013. Association for
  Computing Machinery.

\end{thebibliography}
\end{document}